\documentclass[aps,epsfig,twocolumn,superscriptaddress,prl]{revtex4}

\usepackage{graphicx}% Include figure files
\usepackage{amsthm} % for proofs
\usepackage{amsmath} % for environments like align
\usepackage{mathrsfs}
\usepackage{xcolor}
\usepackage{dcolumn}% Align table columns on decimal point
\usepackage[ruled,vlined]{algorithm2e}
\usepackage{braket}
\usepackage{bm}% bold math
\usepackage[colorlinks,linkcolor=blue,citecolor=blue,urlcolor=blue]{hyperref}% add hypertext capabilities
\usepackage{amsfonts} % for \mathbb{Z}

\newtheorem{theorem}{Theorem}

\newtheorem{lemma}[theorem]{Lemma}

\theoremstyle{definition}
\newtheorem{definition}[theorem]{Definition}

\newcommand{\tr}{\mathrm{Tr}}

\newcommand{\C}{\mathcal{C}}

\newcommand{\HarvardPhysics}{Department of Physics, Harvard University, Cambridge, MA 02138, USA}
\newcommand{\CTP}{Center for Theoretical Physics, Massachusetts Institute of Technology, Cambridge, MA 02139, USA}
\begin{document}

\title{Enhancing Detection of Topological Order by Local Error Correction}

\author{Iris Cong\vspace{.35mm}}
\thanks{These authors contributed equally to this work.}
\affiliation{\vspace{0.5mm}
\HarvardPhysics}

\author{Nishad Maskara}
\thanks{These authors contributed equally to this work.}
\affiliation{\vspace{0.5mm}
\HarvardPhysics}

\author{Minh C. Tran}
\affiliation{\vspace{0.5mm}
\HarvardPhysics}
\affiliation{\vspace{0.5mm}
\CTP}

\author{Hannes Pichler}
\affiliation{\vspace{0.5mm}
Institute for Theoretical Physics, University of Innsbruck, 6020 Innsbruck, Austria}
\affiliation{Institute for Quantum Optics and Quantum Information of the Austrian Academy of Sciences, 6020 Innsbruck, Austria\vspace{1.5mm}
}

\author{Giulia Semeghini}
\affiliation{\vspace{0.5mm}
\HarvardPhysics}

\author{Susanne F. Yelin}
\affiliation{\vspace{0.5mm}
\HarvardPhysics}

\author{Soonwon Choi}
\affiliation{\vspace{0.5mm}
\CTP}

\author{Mikhail D. Lukin\vspace{0.75mm}}
\affiliation{\vspace{0.5mm}
\HarvardPhysics}

\date{\today
%\vspace{2mm}
}

\preprint{MIT-CTP/5462}
%TC:ignore
\begin{abstract}
The exploration of topologically-ordered states of matter is a long-standing goal at the interface of several subfields of the physical sciences. Such states feature intriguing physical properties such as long-range entanglement, emergent gauge fields and non-local correlations, and can aid in realization of  scalable fault-tolerant quantum computation. However, these same features also make creation, detection, and characterization of topologically-ordered states particularly challenging. Motivated by recent experimental demonstrations, we introduce a new paradigm for quantifying topological states---locally error-corrected decoration (LED)---by combining methods of  error correction with ideas of renormalization-group flow. Our approach allows for efficient and robust identification of topological order, and is applicable in the presence of incoherent noise sources, making it particularly  suitable for realistic experiments. We demonstrate the power of LED using numerical simulations of the toric code under a variety of perturbations. We subsequently apply it to an experimental realization, providing new insights into a quantum spin liquid created on a Rydberg-atom  simulator.  Finally, we extend LED to  generic topological phases, including those with non-abelian order. 
\end{abstract}
%TC:endignore

\maketitle

Topological ordering is an exotic phenomenon which can occur when quantum fluctuations and local constraints stabilize a state with long-range entanglement~\cite{wen_colloquium_2017}.
With their non-local correlations, topologically-ordered states feature many remarkable properties and can be used for protecting quantum information non-locally~\cite{wen_colloquium_2017,nayak_non-abelian_2008,terhal_quantum_2015}. 
Yet, because these states 
appear to be liquid-like at short length-scales~\cite{sachdev_kagome-_1992}, they cannot be identified or characterized using any local order parameters. 
Instead, the canonical  approach to discern topological order is to measure operators supported on large closed loops, the Wilson loops~\cite{hastings_quasiadiabatic_2005,wilson_confinement_1974,wen_colloquium_2017,haah_invariant_2016}. 
However, such operators are challenging to identify or measure: while they have simple forms in certain fixed-point models, this is generally not the case for states realized experimentally in the presence of noise or other perturbations.  
In these cases, the expectation values of the simple or `bare' Wilson loop operators described above decay exponentially with the loop's perimeter, which hinders the experimental certification of topological order.

\begin{figure*}[t]
\includegraphics[width=0.8\textwidth]{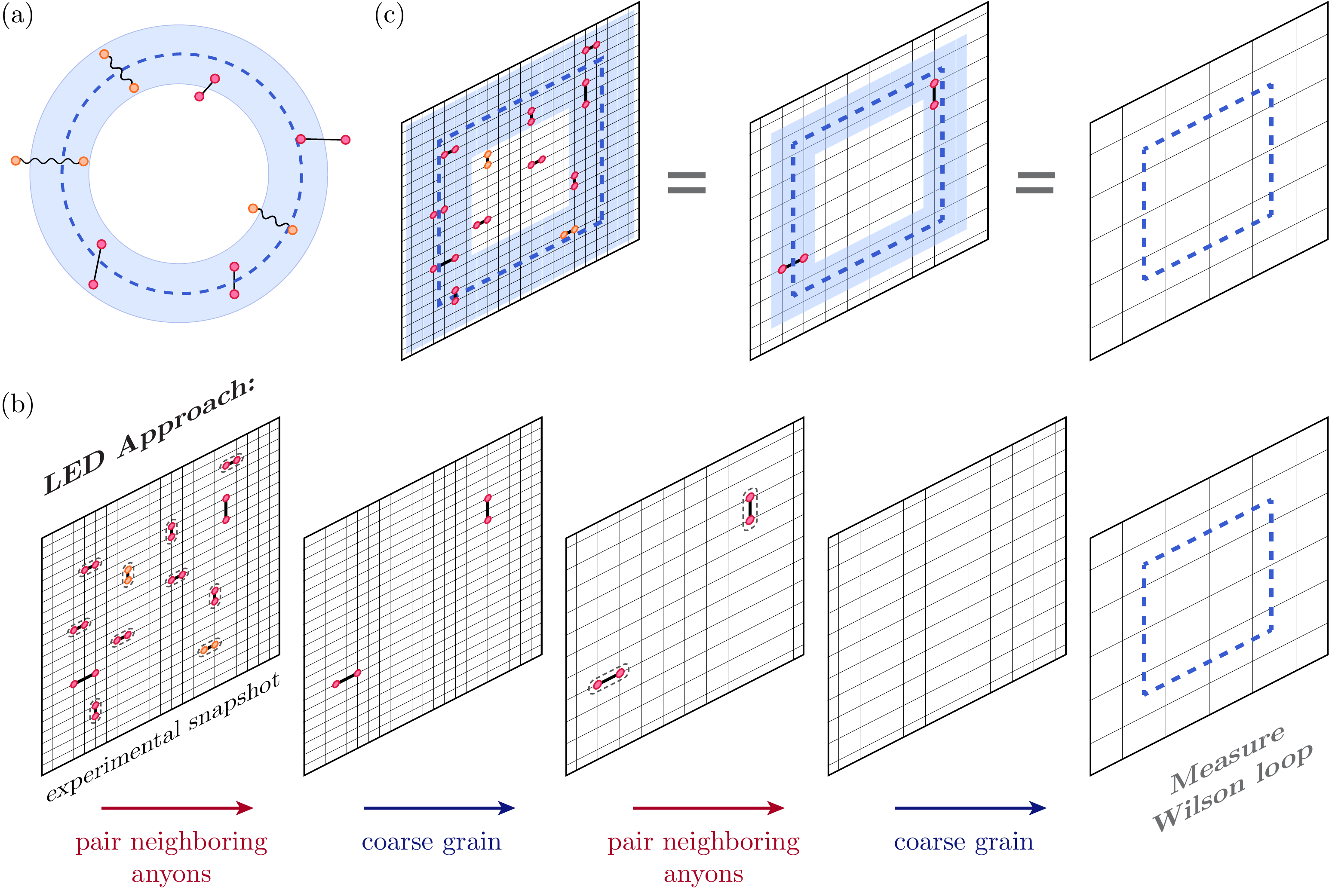}
\caption{Detecting topological phases via snapshot-based LED.
(a) In the absence of perturbations, a topologically-ordered state with zero correlation length such as Kitaev's toric code state~\cite{kitaev_fault-tolerant_2003}
is characterized by $+1$ expectation values of `bare' Wilson loop operators, which are typically tensor products of single-site operators (dotted blue loop). In realistic systems, however, coherent perturbations give rise to virtual anyon pairs (red dots/straight lines), and incoherent errors introduce physical anyon pairs (orange dots/wavy lines); this causes the expectation value of bare Wilson operators to decay exponentially with the loop's perimeter. To account for these local fluctuations, one can measure `fattened' Wilson operators supported on an annulus (blue); the LED loops constitute one realization of this.
(b) LED method to measure decorated Wilson loop observables for $\mathbb{Z}_2$ topological order in a system where qubits live on the links of a square lattice, and stabilizers are associated with vertices. Given an experimental snapshot of all qubits in the $Z$ or $X$ basis, one can obtain values for all stabilizer operators in that basis, thereby identifying the locations of all $e$ or $m$ anyons, respectively.  In the first step, neighboring anyons are paired using a local decoder (dashed pairings), and each pair is removed by flipping the value(s) of qubit(s) lying on a path of minimal length connecting the two anyons; subsequently, the lattice is coarse-grained so that only a fraction of the original qubits remain. These two steps are iterated $n$ times (here, $n=2$), after which a bare Wilson loop is evaluated on the final, coarse-grained state.
(c) The final, bare Wilson loop operator evaluated on the final state is equivalent to decorated Wilson loop operators evaluated at earlier iterations (see Methods).}
\label{fig:fig1}
\end{figure*}

To address these challenges, several methods have been developed to construct `fattened' Wilson loops which do not decay with loop size. 
These include a systematic method utilizing quasi-adiabatic connections to the fixed-point models~\cite{hastings_quasiadiabatic_2005}, as well as variational and tensor-network-based approaches~\cite{bridgeman_detecting_2016,iqbal_entanglement_2021,duivenvoorden_entanglement_2017,jamadagni_learning_2022}. Nevertheless, these methods are challenging to apply in realistic experiments, especially in the presence of incoherent noise (e.g., spontaneous emission).  
Other signatures,
such as topological entanglement entropy~\cite{kitaev_topological_2006,levin_detecting_2006}
are likewise difficult to measure in large systems.

Motivated by these considerations, we introduce a novel and powerful framework, {\it locally error-corrected decoration (LED)}, for studying and characterizing topologically-ordered states. By leveraging the error-correcting properties of topological phases, LED provides a systematic method to construct and efficiently measure `decorated' Wilson loop operators,  a variant of the fattened loop operators. This enables the identification and characterization of topological order at large length-scales in the presence of both coherent perturbations and incoherent noise, which are particularly challenging or impossible using conventional methods. 

In its most general form, LED corresponds to a class of hierarchically-structured  quantum circuits which resemble the classification of quantum phases using RG flow~\cite{schuch_classifying_2011,chen_classification_2011}. 
Yet, for a wide range of experiments
where the prepared state is known to approximate a fixed-point state with zero correlation length~\footnote{See Supplementary Information.}, there is an efficient `snapshot-based' realization of LED using only 
classical post-processing of experimental measurements in a few fixed bases.
In this work, we primarily focus on snapshot-based LED due to current experimental limitations and the hardness of simulating 2D quantum circuits.

\vspace{-1.5mm}
\section{LED Approach}
\vspace{-1.5mm}
The key idea of LED can be understood by considering Kitaev's toric code model, a canonical example of topological order. 
Specifically, we consider qubits localized on the edges of a square lattice. The ideal, fixed-point Hamiltonian is defined as~\cite{kitaev_fault-tolerant_2003}:
\begin{align}
\label{eq:tc-ham}
    H_\textrm{TC} &= -J \sum_v A_v -J \sum_p B_p
\end{align}
where $A_v = \prod_{i \in \mathrm{adj}(v)} X_i$, $B_p = \prod_{i \in \mathrm{adj}(p)} Z_i$, and $\mathrm{adj}(v)$ (resp., $\mathrm{adj}(p)$) denote the set of edges touching a given vertex $v$ (plaquette $p$) of the lattice. The ground state space, given by the simultaneous $+1$ eigenspace of all {\it stabilizer operators} $\{ A_v, B_p \}$, forms a quantum error-correcting code: all local operators either act trivially on ground states or couple them to excited states~\cite{kitaev_fault-tolerant_2003}. By measuring stabilizers, one can detect the presence of excitations
and apply a recovery procedure to return the system back to this ground state space.

In this model, contractible Wilson loops can be constructed by multiplying stabilizers, so their expectation values in any ground state of $H_\textrm{TC}$ are $+1$, independent of loop size.
However, in realistic situations, the prepared state differs from the fixed-point state by local fluctuations such as coherent perturbations and incoherent errors (Figure~\ref{fig:fig1}a).
This causes bare Wilson loops to decay exponentially with the number of locations where a fluctuation can intersect the loop (i.e., its perimeter). 

The snapshot-based LED approach 
begins with a measurement of all qubits in the same (Pauli-$Z$ or Pauli-$X$) basis. For each measurement snapshot, one can calculate the stabilizer and Wilson loop values.
Local fluctuations appear as stabilizer violations, which are identified with anyonic excitations~\cite{kitaev_fault-tolerant_2003} (Figure~\ref{fig:fig1}b). 
A local decoder partially removes such fluctuations by flipping measured qubits using only nearby stabilizer values. 
The simplest such local decoder can remove single-qubit errors, by flipping a qubit if and only if both adjacent vertices (resp., plaquettes) are occupied by an $m$ ($e$-anyon). 
However, it cannot remove higher-weight errors, which flip two or more adjacent qubits.
Subsequently, the lattice is coarse-grained, 
which can also be done efficiently on measurement snapshots (see methods). 
Together, the anyon-pairing and coarse-graining steps are repeated for $n$ layers. 
Crucially, the weight of uncorrected errors is reduced in each layer, so that all local errors eventually become single-qubit errors which the decoder can correct; this mimics a real-space RG flow towards the fluctuation-free fixed-point state (see Methods). Finally, a bare Wilson loop is measured on the final, corrected and coarse-grained state (Figure~\ref{fig:fig1}b).

This bare operator measured on the final state is equivalent to a decorated Wilson loop operator measured on the original state (Figure~\ref{fig:fig1}).
Moreover, it is determined solely by the fixed-point state and is independent of the specific fluctuations in the system; this crucially differentiates LED from prior approaches to construct fattened loop operators~\cite{hastings_quasiadiabatic_2005,bridgeman_detecting_2016,iqbal_entanglement_2021}.
Notice that 
all steps in snapshot-based LED can be performed in post-processing (see Methods), making it uniquely suited for integration into experimental measurement procedures. More general LED operators can be constructed through the quantum circuit formulation; one example is presented in the Supplementary Information.

The hierarchical LED procedure is also inspired by the quantum convolutional neural network (QCNN) approach to phase classification, and the decorated Wilson loop operators resemble the ``multiscale string order parameters'' studied in Ref.~\cite{cong_quantum_2019}. 
However, in this context, the LED framework is more general: 
one can construct LED Wilson operators of diameter $L$  with any desired correction distance $d \ll L$ (Figure~\ref{fig:fig2}a) by choosing any local decoder which pairs anyons up to distance $d$ 
(see Methods)~\footnote{Note that technically, the correction distance $d$ is related to the annulus thickness $cd$ by a decoder-dependent constant $c > 1$, since the range of information propagation is generally larger than the range of allowed anyon-pairings (see Methods).}. 
The construction of Figure~\ref{fig:fig1}b with alternating local-decoding and coarse-graining layers is a particularly efficient way to construct local decoders and LED loops with longer-range (e.g., $d,L \propto 2^n$).

We emphasize that the locality of our procedure ensures that only topologically-ordered states can flow to the fixed-point state. Thus, LED gives rise to a new sufficient condition or {\it witness} for topological order. This distinguishes LED from general decoders, which do not typically respect locality and hence cannot be used to certify topological order.

%\vspace{0.5mm}
\section{Numerical Detection of Topological Order with Coherent Perturbations}

To demonstrate the applicability of LED for coherent local perturbations to $H_\textrm{TC}$, we consider a family of states
\begin{align}
\label{eq:tc-peps}
    \ket{\psi(g_X, g_Z)} &= \frac{1}{\mathcal{N}} e^{g_X \sum_i X_i + g_Z \sum_i Z_i} \ket{\psi_\textrm{TC}},
\end{align}
generated by imaginary-time evolution of a toric code ground state $\ket{\psi_\textrm{TC}}$~\cite{chen_local_2010,haegeman_gauging_2015,zhu_gapless_2019}.
As each operator $Z_i$ (resp., $X_i$) creates a pair of $m$ anyons ($e$ anyons), $\ket{\psi(g_X, g_Z)}$ contains virtual anyon fluctuations on top of  $|\psi_\textrm{TC}\rangle$. 
In the special case where $g_X = 0$, topological order is known to survive for perturbations $g_Z \leq g_c = 0.220343$, beyond which the $m$-anyons condense, driving a second-order phase transition into the $Z$-paramagnet state~\cite{castelnovo_quantum_2008}. More generally, $\ket{\psi(g_X\neq 0, g_Z \neq 0)}$ is also topologically-ordered for small $g_X$ and $g_Z$, but the transitions to paramagnetic phases can occur at points which differ from $g_c$.

\begin{figure*}[t]
\includegraphics[width=0.98\textwidth]{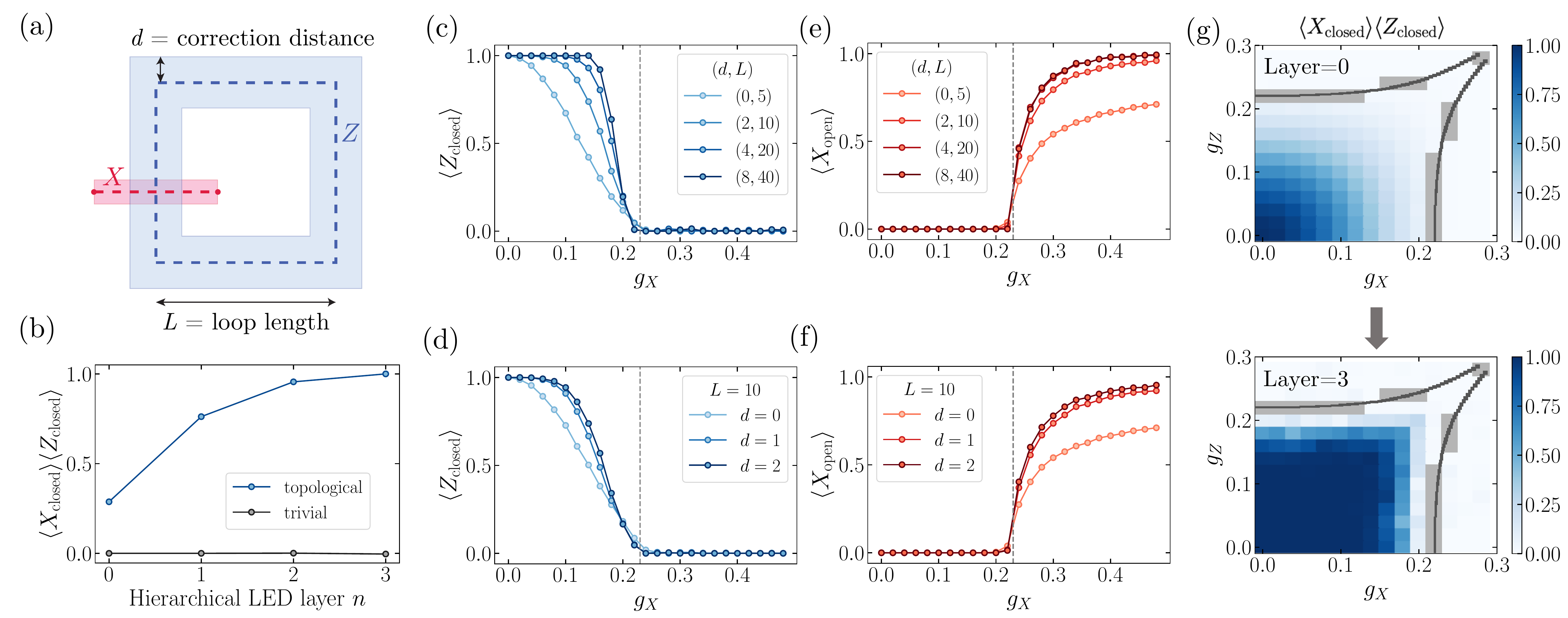}
\caption{Numerical demonstration with coherently perturbed toric code states.  (a) In a general construction of our LED Wilson loop operator, we use a local decoder which pairs anyons
within a region of radius $d$ (blue annulus). Conjugate LED open string operators (red stripe) anti-commute with Wilson loops, and hence must vanish in the topological phase.    (b) Order parameter $\langle X_{\textrm{loop}} \rangle \langle Z_{\textrm{loop}} \rangle$ for a trivial state ($g_Z=0.0, g_X=0.26$) and a topological state ($g_Z = 0.12, g_X=0.12$), upon varying $n$, using a distance-four patch decoder and coarse-graining blocksize two respectively (see Methods).  (c) Output at different $n$ along the $g_Z=0.14$ line of the phase diagram. Gray dotted line is conjectured phase transition region. 
(d) Expectation values of generic LED Wilson loops with the same diameter $L$, using the pairing decoder ($d=1$) and distance-four patch decoder ($d=2$) without coarse-graining.
(e,f) Corresponding expectation values of bare and decorated open string operators.
(g) Order parameter values constructed from bare Wilson loops ($n=0$) and LED Wilson loops ($n=3$), using the same LED procedure as (c,e), across varying values of $(g_X,g_Z)$.
Dark gray regions are numerical estimates for the phase boundary between topological and trivial (Methods).
Light gray regions correspond to locations where sampling is expensive due to large correlation length.
}
\label{fig:fig2}
\end{figure*}

In testing LED, we numerically simulate projective measurements of $\ket{\psi(g_X, g_Z)}$ and use them as input ``experimental snapshots'' in Figure~\ref{fig:fig1}b (see Methods).
Figure~\ref{fig:fig2}b shows the value of the LED order parameter for a trivial and a topological state with $g_Z = 0.14$, when $n$ is varied (and $d, L \propto 2^n$). Clearly, the order parameter stays at $0$ for the trivial state, but increases from a small, finite value to one for the topological state as $n$ 
is increased.
Similar behavior is also observed throughout a one-dimensional parameter space in Figures~\ref{fig:fig2}c and \ref{fig:fig2}d, whenever the correction distance $d$ is increased, while keeping $d \ll L$ to prevent overcorrection (see Methods). Importantly, amplification occurs only if the input state is topological, and the order parameter approaches $0$ for all trivial states.

Another important set of observables for characterizing topological order are $X$ and $Z$ open string operators, which detect the transition from the topological phase to the trivial, paramagnet phase.
Because LED Wilson $Z$-loop operators (resp., $X$-loop operators) are linear combinations of $Z$ ($X$) closed loops supported on an annulus, they anti-commute with conjugate $X$ ($Z$) open strings connecting the interior and exterior of the annulus.
As such, the expectation value of long, open strings must flow to zero in the topological phase, where closed-loop LED operators flow to unity with increasing $d$~\footnote{This holds for any LED open string.}.
The topological-to-trivial phase transition occurs when certain long, open $X$ or $Z$ strings acquire non-zero expectation value, due to the condensation of $m$ or $e$ anyons, respectively.
Indeed, deep in the paramagnetic phase the state $\lim_{g_x \rightarrow \infty} \ket{(g_x, g_z)}$ is polarized along the $X$ direction, and $X$ open strings become unity.
However, for generic trivial states, open strings also decay exponentially with length, due to local fluctuations of the opposite type; nevertheless, LED can still amplify their expectation values by removing the effect of local fluctuations.
This behavior is demonstrated in our simulations: in Figure~\ref{fig:fig2}e,f, open string expectation values stay at $0$ in the topological phase, but are amplified and saturate to a non-zero value in the trivial (paramagnetic) phase. Because LED amplifies the contrast between trivial states and a large class of topological states,
the topological order can be detected using with lower sample complexity---that is, by using substantially fewer experimental repetitions~\cite{cong_quantum_2019,haah_sample-optimal_2017} 
(see Supplementary Information).

Let us note that the boundary dividing the states whose LED operators approach zero and one does not necessarily correspond to the topological phase boundary: in general, it depends on the choice of decoder and coarse-graining length-scale. For instance, this is observed in Figure~\ref{fig:fig2}g, where closed loops are nearly one after $n=3$ layers for a large region within, but not fully encompassing, the topological phase. Hence, LED is not always a necessary condition for topological order.

\section{Effect of Incoherent Errors}

We next demonstrate the application of LED in the presence of incoherent local noise such as spontaneous emission or dephasing, which commonly occur in experiments. 
Because local decoders can recover topologically encoded information in the presence of small, local error channels~\cite{dennis_topological_2002,kitaev_topological_2006} 
it is reasonable to ask whether mixed states prepared in these systems exhibit topological ordering.

\begin{figure*}[t]
\includegraphics[width=0.75\textwidth]{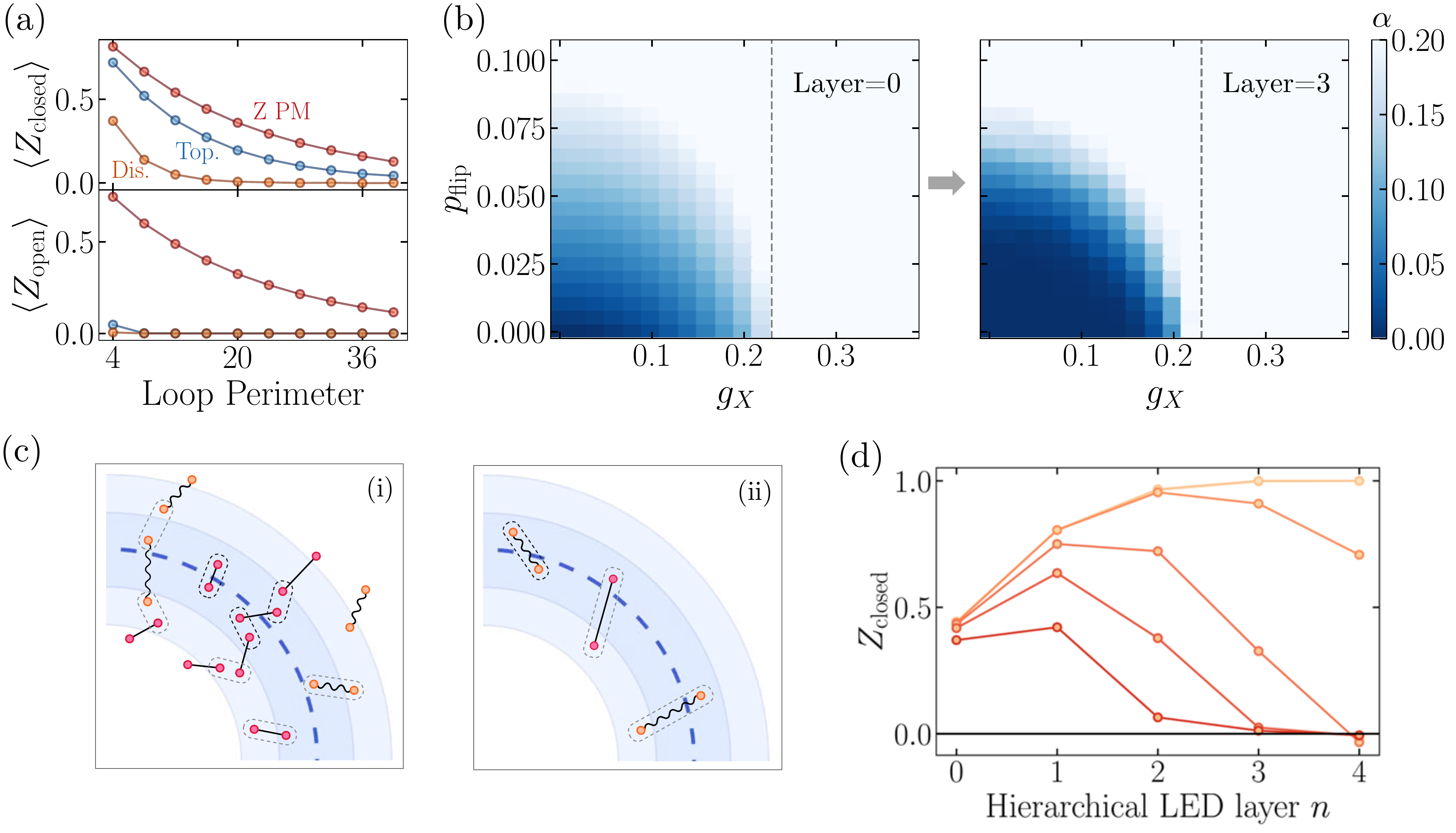}
\caption{Application to mixed states. (a) Without error correction, generic points in the topological and trivial, disordered phase ($g_Z=0.12$, $g_X=0.18$, $p_{\mathrm{flip}}=0.0$ and $g_Z=0.06$, $g_X=0.0$, $p_{\mathrm{flip}}=0.11$ resp. shown in the plot) appear very similar qualitatively, as closed loops decay exponentially with loop perimeter in both cases, while open strings remain close to zero (see Methods). In contrast, in the trivial, paramagnet phase ($g_Z=0.32, g_X=0.2, p_{\mathrm{flip}}=0$), open strings decay with the same perimeter-law as closed loops.
(b) $g_Z=0.14$ slice of mixed state phase diagram, containing topological, disordered, and 
$X$-paramagnetic phases. These phases are  associated with  fixed-point states $g_x$\;$=$\;$g_z$\;$=$\;$p_{\textrm{flip}}$\;$=$\;$0$, $g_z$\;$\rightarrow$\;$\infty$, $g_x$\;$\rightarrow$\;$\infty$, and $p_{\textrm{flip}}$\;$\rightarrow$\;$0.5$, respectively.
The flow of the closed-loop decay exponent $\alpha$ under LED provides a sharp divider between two kinds of perimeter-law decay, observed in different regimes of the mixed-state phase diagram.
(c) In the uncorrectable regime (i), the local decoder of LED pairs anyons incorrectly, resulting in perimeter-law decay with large $\alpha$ in disordered and paramagnetic phases. Moreover, the probability of such an incorrect pairing can increase with the number $n$ of LED iterations. Here, the black pairings are made by LED at or before one specific value of $n$, and  gray pairings are made upon performing one additional LED iteration. In the correctable (topological) regime (ii), increasing $n$ can reduce $\alpha$ to zero, as fluctuations of higher characteristic length $\xi$ can be reliably corrected using only local information. In the conceptual framework where an LED operator is embedded in a surface code on an annulus (see Figure~\ref{fig:fig2}a), incorrect pairings corresponds to logical errors (e.g. $X_L$). 
(d) Expectation values of LED loop observables upon increasing  $n$ ($d, L \propto 2^n$), in thermal states of varying temperatures (between $0$ and $0.35$, with darker colors indicating higher temperatures) and $p_{\mathrm{flip}}=0.02$.}
\label{fig:fig3}
\end{figure*}

To study such examples, 
we introduce incoherent bit- and phase-flip errors
by independently flipping, with probability $p_{\textrm{flip}}$, each measured qubit in a snapshot of $\ket{\psi(g_X, g_Z)}$. Here, we associate topological order with states that can be transformed into a ground state of $H_{\textrm{TC}}$ via local operations.
Our analysis then suggests that the resulting mixed-state phase space contains a $\mathbb{Z}_2$-topological phase, a $Z$-paramagnet, an $X$-paramagnet, and a disordered phase with large incoherent error rates. 
However, it is especially difficult to distinguish the topological and disordered phases using measurements of bare operators alone:  in both phases, open strings remain close to zero, while bare Wilson loops decay exponentially with perimeter as $e^{-\alpha L}$, where the exponent $\alpha$ interpolates smoothly between the phases (Figure~\ref{fig:fig3}a,b). This is in contrast to the paramagnet phases, where closed loops exhibit similar behavior, but certain open strings decay with the same exponent $\alpha$ as the closed loops~\cite{fredenhagen_charged_1983}.

Upon studying the behavior of LED operators, one finds that the mixed-state phase space exhibits two qualitatively different regimes (Figure~\ref{fig:fig3}b). 
LED reduces $\alpha$ to $0$ with increasing $d$ in the `correctable' regime, while $\alpha$ grows in the `uncorrectable' regime.
Further, correctable states with 
small $p_{\mathrm{flip}}$ are connected to topologically-ordered pure states, suggesting these mixed states are topologically-ordered as well. Indeed, we show that correctability implies the input state cannot be prepared from a product state using only local operations. In particular, if LED Wilson loops are amplified to above $1-\epsilon$ under depth $d$ correction, this certifies topological order up to length-scale $O(\mathscr{L}-d)$ where $\mathscr{L} \sim 1/\sqrt{\epsilon}$. Furthermore,  we argue (see Methods) that, under plausible conditions, this implies the \textit{entanglement negativity} of the input state contains a topological term; this connects the LED characterization of mixed state topological order to other studies ~\cite{peres_separability_1996,horodecki_separability_1996,lee_entanglement_2013}.

The ability of LED to distinguish between the topological and disordered phases can be understood by analogy to quantum error correction.
Conceptually, since any given LED loop operator is supported on an annulus, we can consider this operator as being embedded in a surface code on this annulus with open boundary conditions, which supports a logical qubit. Then, an LED $Z$-loop operator  corresponds to a logical-$Z$ operator for this qubit, while an $X$-string  connecting the interior of the annulus to the exterior corresponds to a logical-$X$ operator (see Figure~\ref{fig:fig2}a).
In this framework, the decay rate $\alpha$ of Wilson loops corresponds to a local logical error rate per unit length, and in the correctable phase, LED-based decoding succeeds with high probability as long as the code distance $d$ is sufficiently large (Figure~\ref{fig:fig3}c).
However, in the uncorrectable phase, such as when $p_{\mathrm{flip}}$ is above the error correction threshold or when long, open strings condense in a paramagnetic phase, decoding cannot correctly pair anyons, resulting in a high rate of logical errors~\cite{dennis_topological_2002}.

The above results are deeply rooted in the stability of topological order against local perturbations. In contrast, any finite temperature destroys long-range topological order as it leads to freely propagating thermal anyons.
In Figure~\ref{fig:fig3}d, we consider the toric code model at finite temperature, with local incoherent errors, and find that the LED loop operators indeed approach zero upon increasing $n$. 
Interestingly, their expectation values flow non-monotonically, being amplified at small $n$ before eventually turning to $0$. 
This occurs because of a competition between two effects: thermal anyons are uncorrectable, so their density accumulates under RG flow; however, local fluctuations are corrected at early layers, which initially amplifies LED loop expectation values. 
Because loops at different $n$ probe correlations at different length-scales, the turning point in these curves can be used to identify the characteristic length-scale of separation between thermal anyons, or equivalently, the system's temperature.

\section{Experimental Realization in Rydberg Atom Arrays}

\begin{figure}[t]
\includegraphics[width=0.46\textwidth]{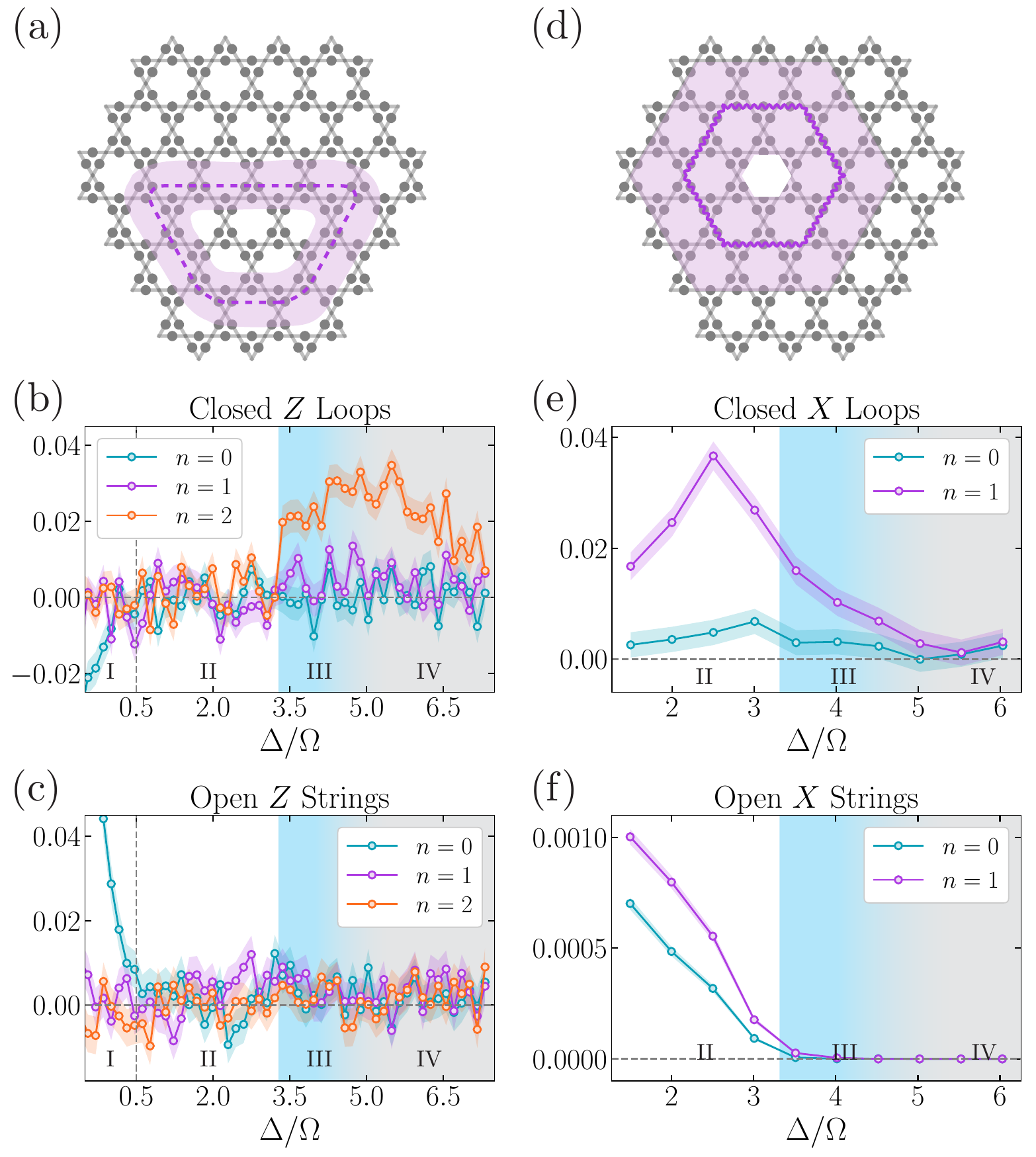}
\caption{Enhancing experimental detection of $\mathbb{Z}_2$ spin liquid. (a,d) In the experiment~\cite{semeghini_probing_2021}, 219 qubits are placed on the links of a kagome lattice. 
Upon applying LED, the $Z$ and $X$ closed-loop observables are amplified for certain ranges of $\Delta/\Omega$. The shaded purple regions show the support of large, decorated Wilson loops after one layer of correction with $n=1$. (b,e) Expectation value of Wilson loops depicted in (a,c) for different correction layers $n$. Plotted error bars (shaded regions) show expected variation (one standard error) of the mean. The regime in which both types of loops are amplified corresponds nicely to the spin-liquid regime identified in Ref.~\cite{semeghini_probing_2021} (shaded blue region). (c,f) The behavior of expectation values of open $Z$- and $X$-strings under LED further confirms our findings, as both types of open strings stay at $0$ in the spin-liquid regime. Here, the measured open strings are half of the Wilson loops. By considering the behavior of all types of loops and strings---closed and open, $Z$ and $X$---we find that there are four regimes (I-IV), corresponding to four phases:  (I) $Z$-paramagnet, (II) $X$-paramagnet, (III) topological spin liquid (blue), and (IV) a phase which is consistent with strong decoherence effects (gray). In our analysis, the progression from Regime (III) to (IV) appears to be smooth.}
\label{fig:fig4}
\end{figure}

We now use LED to characterize and provide new insights into the  $\mathbb{Z}_2$-topologically-ordered states recently realized on a 219-qubit programmable quantum simulator~\cite{semeghini_probing_2021}. In the experiment, qubits are encoded in ground states and $n=70$ Rydberg states of neutral $^{87}$Rb atoms and 
placed in an array on the links of a kagome lattice (Figure~\ref{fig:fig4}). 
This model maps onto a dimer model, where each Rydberg atom can be viewed as a dimer covering the two adjacent vertices of the kagome lattice~\cite{verresen_prediction_2020}: the Rydberg blockade interaction between nearby atoms enforces a ``dimer constraint'' by preventing, with high probability, any vertex from being covered by more than one dimer~\cite{saffman_quantum_2010}.

This dimer model is predicted to support a $\mathbb{Z}_2$-topologically-ordered state of the resonating valence bond (RVB) type,  involving the equally-weighted superposition of all dimer coverings \cite{verresen_prediction_2020,misguich_quantum_2002,poilblanc_topological_2012}.
In this model, $Z$-stabilizers are given by $(-1)$ times the product of single-qubit $Z$-operators on the edges touching a vertex, $X$-stabilizers are given by the product of off-diagonal operators supported on the triangles bordering a hexagon (see Methods), and the RVB state forms a fixed-point state. An $e$ (resp., $m$) anyon arises when a $Z$ ($X$) stabilizer is violated~\cite{samajdar_emergent_2022,tarabunga_gauge-theoretic_2022,verresen_unifying_2022}~\footnote{The $(-1)$ factor for $Z$-stabilizers ensures stabilizer expectation values of $+1$, because each vertex is touched by exactly one dimer.}.

In the experiment, a topologically-ordered state  is prepared by
quasi-adiabatically adjusting the detuning $\Delta$ and Rabi frequency $\Omega$ of a global laser drive~\cite{semeghini_probing_2021}. The onset of topological order is observed by 
studying the expectation values of  Wilson loops and open strings~\cite{fredenhagen_charged_1983,bricmont_order_1983,gregor_diagnosing_2011,verresen_prediction_2020,semeghini_probing_2021}. A state consistent with $\mathbb{Z}_2$ topological order emerges when using a quasi-adiabatic sweep from initial $\Delta/\Omega <0$ to a final value of $\Delta/\Omega$ in the range  $3.3 \lesssim \Delta/\Omega \lesssim 4.5$.
In practice, several factors make quantitative characterization of such states difficult, as they
cause the prepared state to differ from the ideal fixed-point state for the dimer model. In particular, the Rydberg interaction Hamiltonian is only an approximation of the parent Hamiltonian of the fixed-point state~\footnote{For example, the $1/r^6$ interaction between Rydberg atoms gives rise to long-range tails in the interaction Hamiltonian. These long-range tails also destabilize the spin-liquid ground state, which could cause  a first-order phase transition between regions (II) and (IV) in Figure~\ref{fig:fig4}. Nonetheless, a spin-liquid state can be prepared by using finite ramp speeds, as was done in the experiments~\cite{giudici_dynamical_2022,cheng_variational_2022,verresen_prediction_2020}.}. Moreover, the finite sweep speed  and  experimental imperfections (e.g., off-resonant scattering, laser phase noise, spontaneous emission events) also modify the experimentally created state. 
These factors correspond to both coherent and incoherent perturbations, similar to those considered in our toric code simulations. As a result, while topological order can be discerned at modest length-scales, the expectation values of large, bare Wilson loop observables have nearly vanishing signal for almost all final values of $\Delta/\Omega$  (Figure~\ref{fig:fig4}b,e).

To circumvent these imperfections, 
we measure LED loops on the experimentally prepared states. 
Due to the limited experimental system size, it is not possible to consider loops which strictly satisfy the limit  where
$\xi \ll d \ll L$, resulting 
in relatively small expectation values for the LED loop operators. Nonetheless, we clearly observe a range of values of $\Delta/\Omega$ where both $Z$- and $X$-loops are amplified, which corresponds to the spin-liquid interval identified in Ref.~\cite{semeghini_probing_2021} (blue shaded region in Figure~\ref{fig:fig4}).  In particular, some of the largest loops within the system acquire non-zero expectation values in this parameter regime. To further confirm our findings in this intermediate system size setting, we also examine the behavior of open $Z$- and $X$-strings under LED, and we find that there are four regimes (I-IV).
Regimes I, II, and III correspond to the $Z$-paramagnet, $X$-paramagnet, and spin-liquid regime, in agreement with the prior interpretation of experimental results~\cite{semeghini_probing_2021}. We emphasize that our analysis of Regime III goes beyond that of~\cite{semeghini_probing_2021}, showing non-trivial coherence in closed loops at significantly longer length-scales. Furthermore, LED provides novel insights into the nature of Regime IV: because LED does not amplify open or closed string expectation values, our analysis appears to be consistent with a decoherence-dominated disordered phase (see also Supplementary Information). Such a phase is analogous to the disordered part of the mixed-state phase diagram (see Figure~\ref{fig:fig3}c), which has a high density of dephasing ($Z$) errors, in contrast to the valence-bond solid (VBS) phase predicted for the   ground state~\cite{verresen_prediction_2020}.

\section{Circuit-Based LED and Generic Topological Phases}
\begin{figure*}[t]
\includegraphics[width=0.85\textwidth]{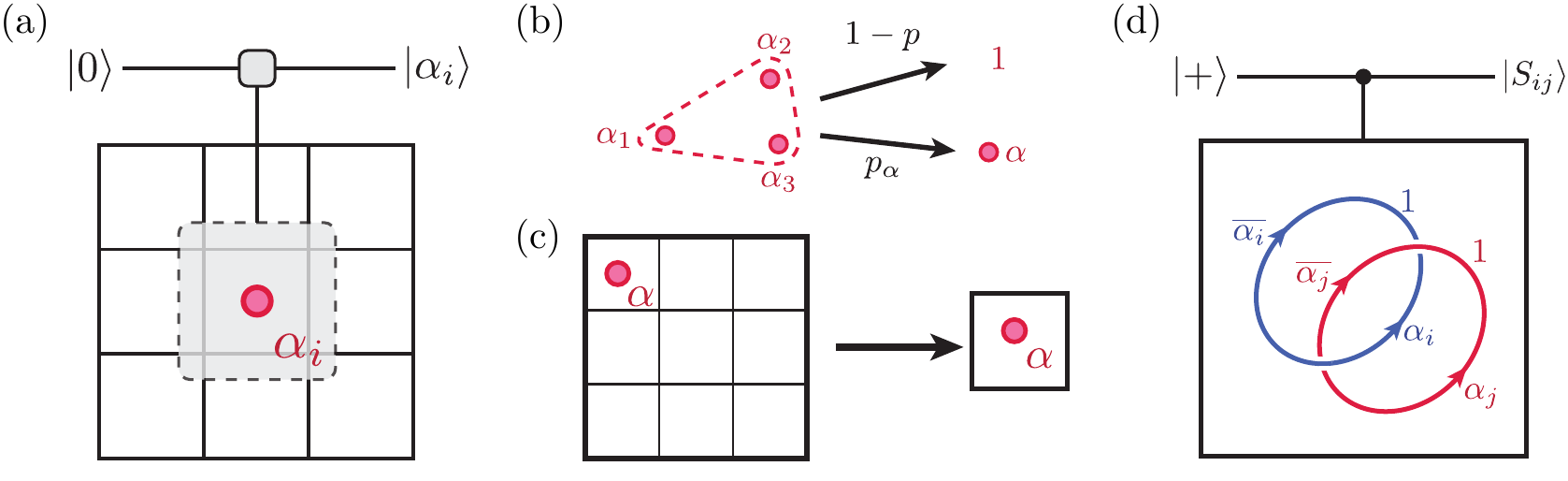}
\caption{LED for generic string-net models. (a) An ancilla qudit is used to measure the topological charge within each local region $\mathcal{R}$: we initialize the ancilla in $|0\rangle$, apply a local unitary $U = \sum_{i,j=0}^{N-1} |(i+j)\, \textrm{mod}\, N \rangle \langle j|_\textrm{anc} \otimes P_{i}$, where $P_{i}$ projects $\mathcal{R}$ onto the subspace with topological charge $\alpha_i$, and finally measure the ancilla's state. (b) Local error correction is performed by inputting the fusion rules of $\mathcal{C}$ into a maximum-likelihood patch-based decoder. Given any $l \times l$ patch, one  identifies possible groupings of anyons (including groupings to the boundary) that can remove all nontrivial topological charges within the patch. The decoder performs the grouping of highest probability by fusing anyons or dragging them to the boundary of the patch~\cite{zhu_universal_2020}. If $\mathcal{C}$ is non-abelian, the vacuum topological charge may only be attained probabilistically with probability $1-\sum_\alpha p_\alpha$, or a nontrivial topological charge $\alpha$ remains with some probability $p_\alpha$. (c) The system is then coarse-grained by applying a quantum circuit corresponding to a multiscale entanglement renormalization ansatz (MERA) representation of the fixed-point state~\cite{konig_exact_2009}. (d) At the final layer,  $S$- and $T$-matrix elements can be measured by introducing an ancilla qubit in the $|+\rangle$ state and applying controlled-anyon-braiding operations. 
More details on implementing Steps (c) and (d) can be found in Methods.
}
\label{fig:fig5}
\end{figure*}

While our current LED analysis uses classical 
post-processing of $Z$- and $X$-basis experimental snapshots, the most generic LED formulation involves a 
quantum circuit model following the QCNN framework of Ref.~\cite{cong_quantum_2019}. 
Here, the entropy associated with both incoherent and coherent fluctuations are systematically removed by introducing ancillary degrees of freedom and applying local unitary transformations, ultimately leaving a purified  state supported on fewer degrees of freedom. 
Notably, 
this enables the application of LED to a large class of non-abelian topological orders known as string-net models~\cite{kitaev_fault-tolerant_2003, levin_string-net_2005}. The anyon content of these models is characterized by a {\it modular tensor category (MTC)} $\C=\mathcal{Z}(\mathcal{A})$, where $\mathcal{Z}$ denotes the Drinfeld center~\cite{wang_topological_2010,bakalov_lectures_2001} of a unitary fusion category $\mathcal{A}$~\cite{kitaev_fault-tolerant_2003}. 
Here, the possible topological charges (i.e., anyon types) are given by the simple objects $\alpha_0, \alpha_1, ... , \alpha_{N-1}$ of $\C$. It is conjectured that any MTC is uniquely determined by modular
$S$ and $T$ matrices which capture its anyon braiding statistics:
\begin{equation}
\vcenter{\hbox{\includegraphics[width = 0.225\textwidth]{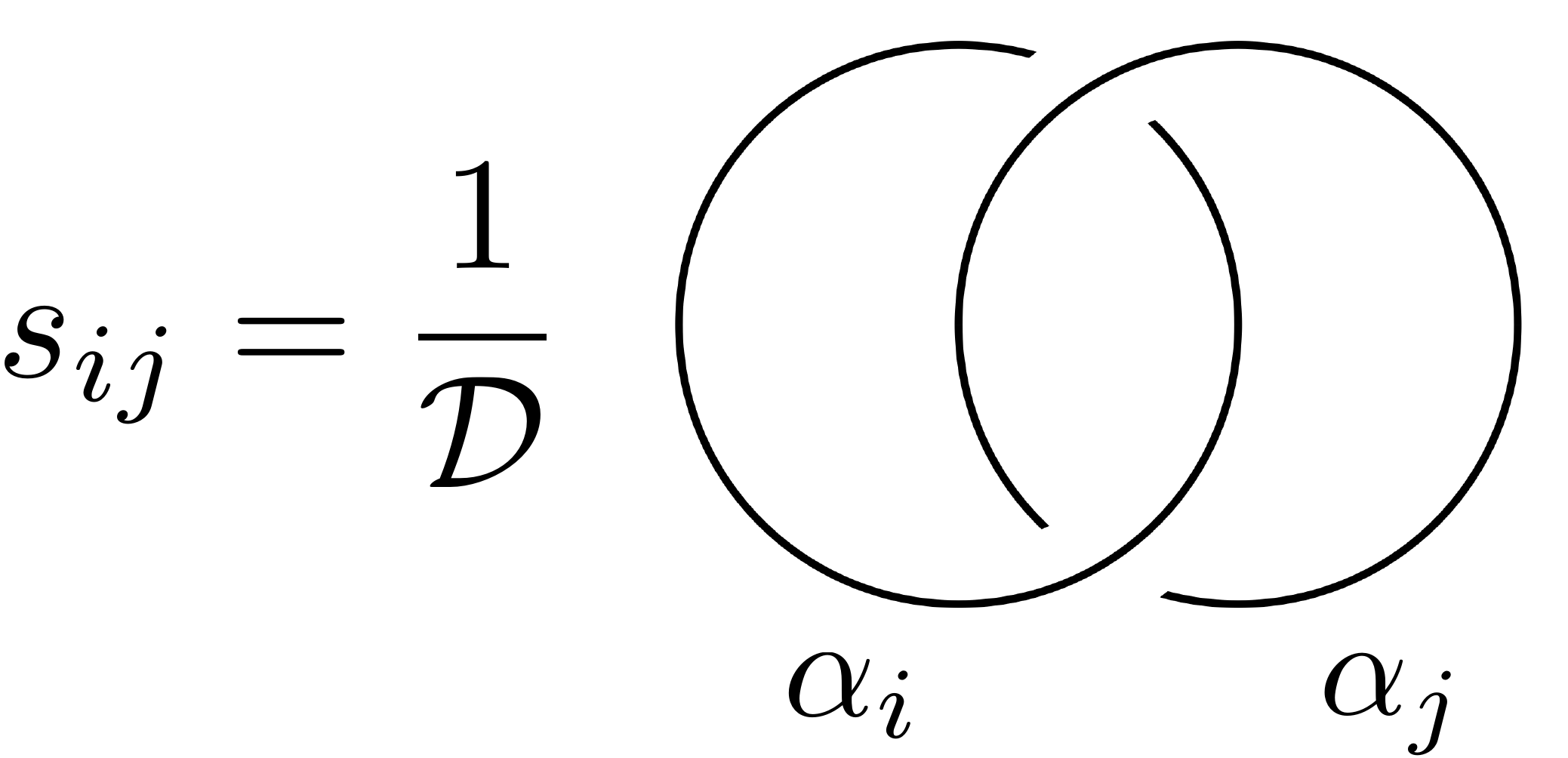}}}
\end{equation}
\vspace{-2mm}
\begin{equation}
\label{eq:twist}
\vcenter{\hbox{\includegraphics[width = 0.16\textwidth]{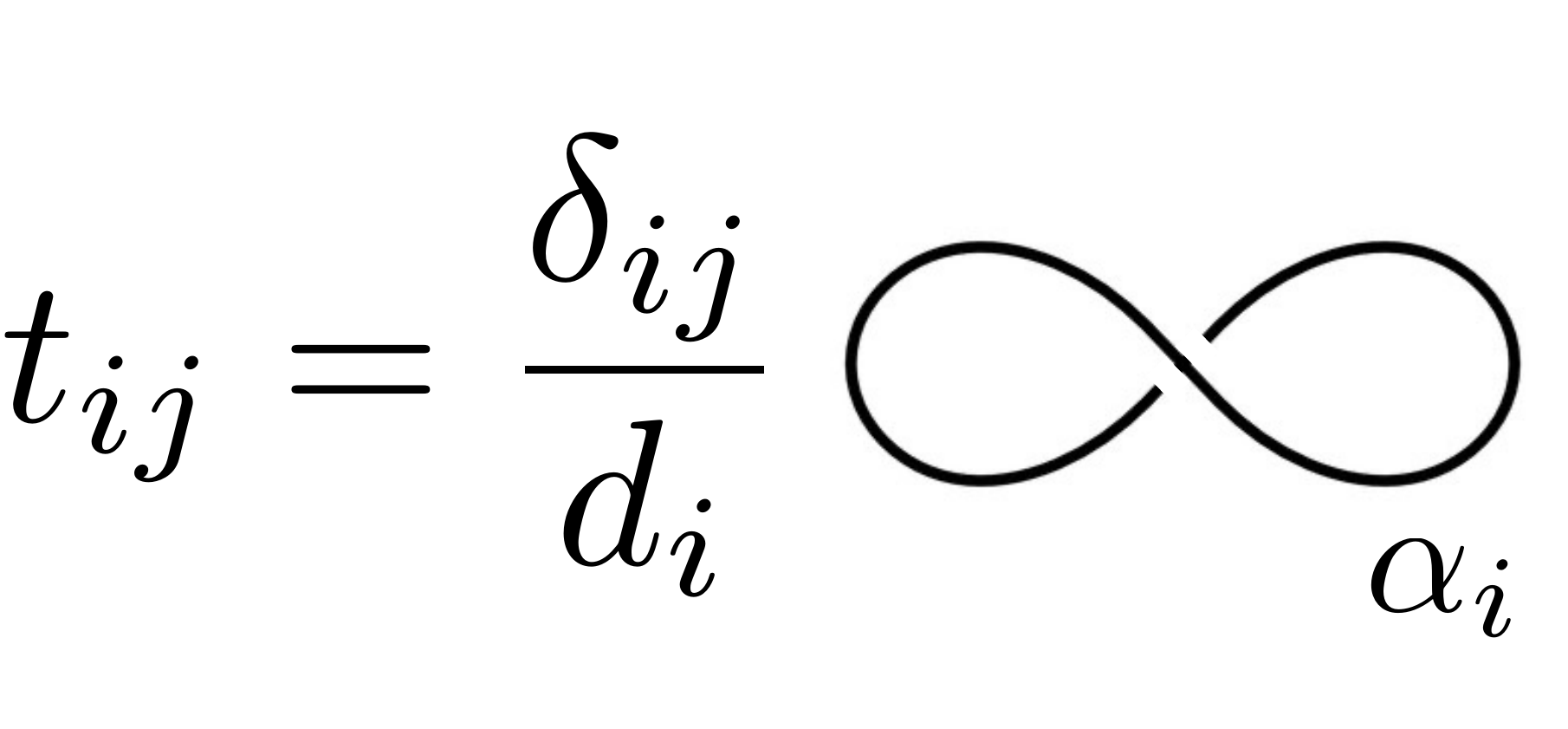}}}
\end{equation}
where $d_i$ is the quantum dimension of $\alpha_i$ and $\mathcal{D} = \sqrt{\sum_i d_i^2}$. For example, a key signature of the toric code MTC $\mathcal{C} = \mathfrak{D}(\mathbb{Z}_2)$ is the $-1$ twist product between $e$ and $m$ anyon loops ($s_{em} = -1$).

Direct measurements of $s_{ij}$ and $t_{ij}$ involve braiding anyons along large loops, and hence are affected by coherent perturbations and incoherent errors. The inability to extract their precise values prevents accurate identification of the topological phase.  To circumvent this, we use a hierarchical LED circuit which systematically detects and identifies errors (anyons) at each site by using ancillary qubits, removes them by inputting the fusion rules of $\C$ into a maximum-likelihood decoder, and applies an entanglement renormalization circuit to coarse-grain the system~\cite{konig_exact_2009}. After multiple layers, the $S$ and $T$ matrices can be measured with much higher accuracy and efficiency (Figure~\ref{fig:fig5}). We note that circuit-based LED is required for the detection and removal of non-abelian anyons. 
More details on circuit-based LED and generic topological phases can be found in Methods and Supplementary Information.

\vspace{-2mm}
\section{Outlook}
\vspace{-1.5mm}

These results demonstrate that 
LED constitutes an exceptionally promising approach to 
enhance the detection and characterization of topological order. 
Several generalizations and future avenues can be considered. For example, the variational methods of QCNN circuits can  
enable adaptive measurement procedures, which can recognize a much larger portion of the topological phase. This opens the door towards achieving a {\it necessary and sufficient} criterion for topological order using LED, which cannot be done using any fixed linear observable~\cite{huang_provably_2022}.  
Moreover, our results indicate that LED is applicable to generic topological orders in higher dimensions, which is challenging to analyze using any currently known techniques. LED can also potentially serve as an order parameter for efficiently characterizing glassy gauge models~\cite{Wang_gauge_glass_2003}, through a mapping shown in Methods.
Additionally, while our present work analyzes a spin-liquid state prepared using a Rydberg-atom quantum simulator, LED is also directly applicable to other platforms such as superconducting qubits~\cite{satzinger_realizing_2021} or trapped ions~\cite{stricker_experimental_2020}.

Another promising direction is to further study whether the ``correctability'' of states in our mixed-state phase diagram can be used to characterize topological order in mixed states more generally~\cite{lee_entanglement_2013,jamadagni_learning_2022,jamadagni_operational_2022,bao_mixed_2023,hastings_topological_2011}. In particular, it could be intriguing to further explore the dependence of the  correctable regime  on the choice of local error correction and/or coarse-graining procedure. Finally, while our approach can be directly applied to any string-net topological order,  it could be interesting to consider more general topological phases, fracton phases or gauge theories with continuous gauge groups~\cite{verresen_efficiently_2022,verresen_unifying_2022}. Such methods can then become indispensable parts of quantum simulation toolboxes for understanding exotic states of entangled quantum matter. 

\vspace{-1.5mm}
\section{Methods}
\vspace{-1.5mm}
\subsection{Numerical Simulations for the Toric Code}
\vspace{-1.5mm}

\begin{figure*}
    \centering
    \includegraphics[width=0.80\textwidth]{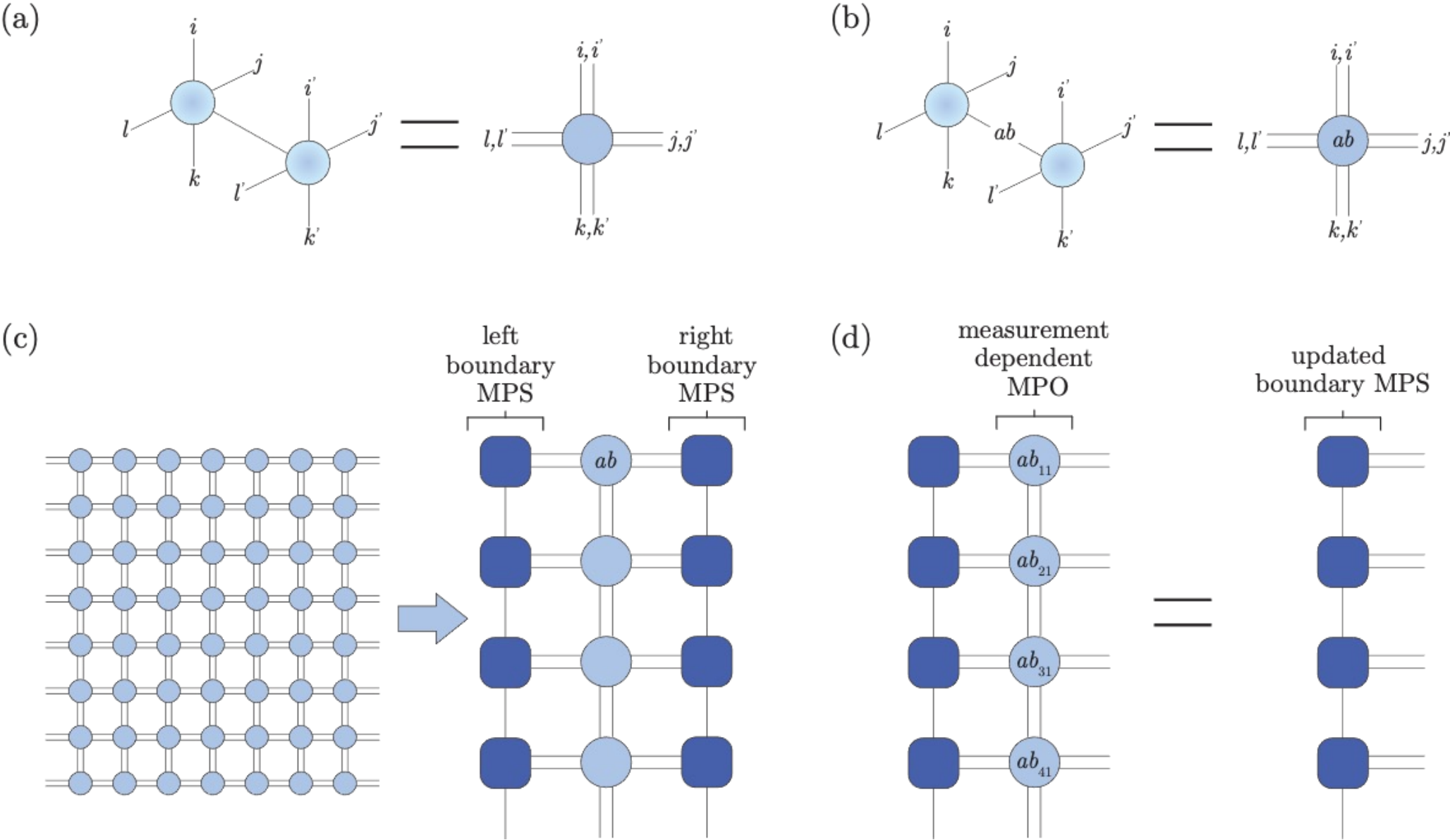}
    \caption{PEPS sampling algorithm. Expectation values are computing with respect to both $\ket{\psi}$ (back) and $\bra{\psi}$ (front). (a) Tracing, or averaging over measurement outcomes can be done by contracting the physical indices, and is needed to compute marginal probabilities. (b) To compute the probability of a particular $Z$ basis measurement, the physical index is assigned a particular value $ab$. 
    (c) We can efficiently contract a 2D PEPS tensor network on an infinite strip of finite height, by using a left and right boundary MPS (only top four rows shown). The probability distribution for projective measurements on a particular site, e.g. $xy = 11$ can then be computed efficiently. 
    (d) Once an entire column has been sampled, the measurement-dependent MPO can applied to the boundary MPS. Although performing this contraction exactly causes the bond-dimension to grow rapidly, away from phase boundaries, finite bond dimension is sufficient for accurate simulation. See Supplementary Information for more details.
    }
    \label{fig:methods-fig1-PEPS}
\end{figure*}

In this section, we explain how the numerical simulations underlying Figures~\ref{fig:fig2} and~\ref{fig:fig3} are performed. 
We begin by constructing a projected entangled pair state (PEPS) representation of the exact toric code ground state~\cite{schuch_resonating_2012}. This construction 
utilizes a parity tensor $P$ defined as
\begin{align}
    P_{ijkl}&=\begin{cases}
    1 & \mathrm{if\;} i+j+k+l=0\mod2\\
    0 & \mathrm{otherwise}
    \end{cases}
\end{align}
where each index $i,j,k,l \in \{0,1\}$ (i.e., the tensor $P$ has bond dimension two).
Because the toric code is defined with qubits on the links of a square lattice, our PEPS representation of the state has one PEPS tensor with two physical indices per unit cell. Letting $p,q$ be the physical indices and $ijkl$ be the virtual indices, the toric code PEPS tensor $A$ is then given by $A_{ijkl}^{pq}=\delta^p_{i}\delta^q_{j} P_{ijkl}$. Our perturbed states $|\psi(g_X,g_Z)\rangle$ are constructed from the toric code state by applying imaginary time evolution to each site $L(g_X,g_Z)=e^{g_X X + g_Z Z}$:
\begin{align}
    A(g_X, g_Z)_{ijkl}^{pq} = \sum_{p',q'} L(g_X,g_Z)_{p'}^{p} L(g_X,g_Z)_{q'}^{q} A_{ijkl}^{p'q'}.
\end{align}
Notice that this operation does not change the PEPS bond dimension, thereby allowing for efficient simulation.

Our goal is to simulate projective $Z$-basis measurements to serve as the ``experimental snapshot'' input in Figure~\ref{fig:fig1}b.
The key ingredient which enables efficient sampling is an algorithm for efficiently computing marginal and conditional probabilities, which can be implemented as follows:
We first label every unit cell by its coordinate $(x,y)$. There are four possible measurement outcomes at each unit cell, and we compute the probability $P(\sigma_{(1,1)}=ab)$ that measurement of the first site $(x,y)=(1,1)$ yields the outcome $ab = 00, 01, 10$, or $11$.
Next, we select a sample $ab_{11}$ based on this probability distribution, 
compute the conditional probability distribution on the second site, $P(\sigma_{(2,1)}=ab | \sigma_{(1,1)}=ab_{11})$, and sample the second measurement outcome $ab_{21}$.
The process then repeats, with each subsequent distribution being conditioned on all prior measurements. 

Computing the probabilities requires contracting a 2D tensor network (Figure~\ref{fig:methods-fig1-PEPS}), which is in general $\#P$-hard~\cite{schuch_computational_2007}. In practice, however, the states we encounter have finite correlation length, and the computation becomes remarkably efficient throughout much of the phase diagram~\cite{cirac_matrix_2021}. In particular, we work on a strip of finite height $L_x$ and infinite length $L_y$, and introduce boundary matrix product states (MPS) to efficiently capture the effect of the environment---that is, the sites different from the one currently being sampled~\cite{napp_efficient_2022}. Because singular-value decomposition truncation is used at each step to prevent the bond dimension of the boundary MPS from growing exponentially~\cite{vidal_class_2008}, the method is approximate; however, we only discard singular values  $<10^{-8}$, so truncation errors are insignificant.
Details of the boundary conditions and contraction ordering are discussed in the Supplementary Information.

In our simulations, we choose $L_x = 300$ unit cells and sample 1000 columns, giving us access to very large snapshots with 600,000 qubits. To minimize boundary effects, we compute observables supported on sites at least 30 unit cells away from the boundaries.
Near the phase boundaries, the bond dimension (entanglement) of the boundary MPS becomes large due to the large correlation length, which increases the computational demands for sampling (gray data points in Figure~\ref{fig:fig2}g). We numerically confirm this phase boundary with an independent calculation (see Supplementary Information).

\subsection{Details on Error-Correction and Coarse-Graining Procedures}

Here, we explain the details of
the LED decoding and coarse-graining procedures and demonstrate how bare Wilson loops become decorated under the LED protocol. Without loss of generality, we consider
$Z$-basis measurements, from which we can calculate plaquette stabilizers $B_u$. Here, each plaquette is labelled by the 
2D coordinate of its unit cell $u = (x,y)$.
Since there are two qubits per unit cell, each qubit carries a coordinate
and a link label $v$ or $h$, depending on whether its corresponding edge in the square lattice is
vertical or horizontal, respectively.
Finally, the projective measurement outcomes are denoted by $\sigma \in \{+1,-1\}$ (see Figure~\ref{fig:methods-fig2}).

To illustrate local error correction, we consider the ``pairing decoder,'' which flips a qubit if and only if its two neighboring plaquettes are simultaneously occupied. 
Importantly, to preserve locality, we first compute all stabilizer values and then flip qubits based on these values. The decision of whether to flip any qubit then depends only on its value, and the values of the six adjacent qubits with which it shares a plaquette.
Equivalently,
this error correction procedure 
corresponds to an operator transformation
\begin{align}
    \sigma_{u+\hat{x},v} \rightarrow \sigma_{u+\hat{x},v} \left(1 + B_u + B_{u+\hat{x}} - B_u B_{u+\hat{x}} \right)/2 \label{eq:inverse1} \\
    \sigma_{u+\hat{y},h} \rightarrow \sigma_{u+\hat{y},v} \left(1 + B_u + B_{u+\hat{y}} - B_u B_{u+\hat{y}} \right)/2 \label{eq:inverse2}
\end{align}
To ensure all local errors are removed after a finite number of LED steps, we also pair anyons which occupy two plaquettes separated by a diagonal, such as $B_u$ and $B_{u+\hat{x}+\hat{y}}$.
The locality of the decoder ensures that the support of any local operator only grows by a finite amount with each step. Subsequently, the coarse-graining procedure replaces each $b \times b$ block of plaquettes with a single plaquette whose value is the product of $b^2$ plaquettes; microscopically, this can be done by defining new qubits as a product of $b$ corresponding qubits in the original lattice. 
The combination of a local pairing step and a coarse-graining step
forms a layer of real-space RG; with each additional layer, one can correct errors of higher and higher weight.

The bare Wilson loops measured in the final state are equivalent to decorated loop operators acting on the original state. 
These decorated operators can be efficiently computed from projective measurement data, since their eigenstates are product states in the $Z$ and $X$ bases, respectively. 
Furthermore, in the operator transformation picture, any loop or string of length $L$ maps onto a linear combination of exponentially many ($2^{O(L)}$) loops or strings, respectively. Thus, while the operator transformation picture is helpful for conceptual reasons, it is computationally much easier to use the original picture of error-correction and coarse-graining.

A few remarks are in order. First, one important property of LED is that it preserves commutation relations: consider two anti-commuting $X$ and $Z$ strings which intersect at a single point, far from the strings' endpoints. Upon applying LED, the resulting decorated strings still anti-commute. This is
because the correction is computed only using stabilizers, so it decorates $Z$-operators by a linear combination of closed $Z$-loops, and similarly for $X$. 
Moreover, other local decoding algorithms, such as cellular automata and RG decoders, can also be used to generate different LED operators~\cite{duclos-cianci_fault-tolerant_2013}.
In the following section, we describe a flexible, ``patch-based'' local decoder for the toric code, which allows LED to classify a wider range of states as topological.

\begin{figure*}
    \centering
    \includegraphics[width=0.8\textwidth]{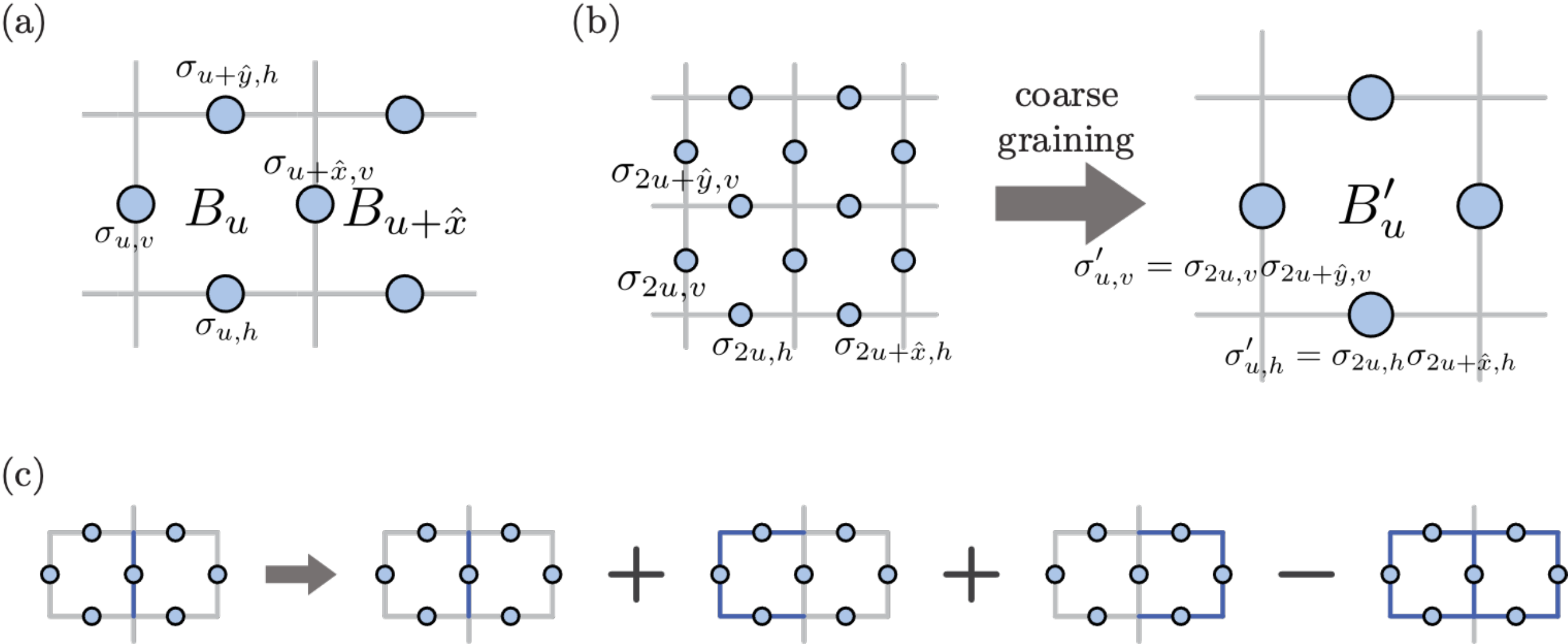}
    \caption{LED coarse-graining and operator transformation. (a) In the toric code model, qubits are located on the links of a square lattice,  and the stabilizer associated with any plaquette is given by a product of four single-qubit Pauli-$Z$ operators. (b)  Coarse-graining maps a $b \times b$ block of plaquettes to a single plaquette whose value is the product of the $b^2$ plaquettes (here $b=2$).
 Microscopically, coarse-grained qubits $\sigma'$ are products of $b$ lower-level qubits $\sigma$. Coarse-grained stabilizers $B'_u$ are therefore equivalent to a product of $b \times b$ stabilizers at the lower level. (c) Pairing correction flips a qubit conditioned on the state of its two neighboring stabilizers. This is equivalent to an operator transformation where the qubit is decorated by products of closed loops. }
    \label{fig:methods-fig2}
\end{figure*}

\subsection{Patch-based decoder}

The patch-based decoder with variable correction distance $d$ is based on a local 
minimum-weight perfect matching (MWPM) procedure. 
In the first decoding step, a local
MWPM decoder is convolved with all $l$ by $l$ square regions of the toric code, where $l \sim d$; for each region, MWPM takes as input the location of the enclosed anyons.
Because both $e$ and $m$ anyons can freely move into and out of the region, this is analogous to decoding a surface code with open boundaries. Therefore, MWPM pairs any given anyon either with another anyon or with the boundary.

The second step aggregates MWPM pairings. Since the square regions can overlap, a pair may appear more than once.
As such, after choosing a natural indexing of the plaquettes,
we create a list of all MWPM pairings between two plaquettes $(p,q)$ with $p < q$; pairings with the boundary are not included (Figure~\ref{fig:methods-fig3}).
For each plaquette $p$ containing an anyon, the patch-based decoder then performs the pairing $(p,q)$ which occurs most often. 
This procedure naturally favors pairings that flip fewer qubits, because shorter-range pairings can be included in more local patches.

A critical property of this decoder is that it
preserves locality. In the first step, MWPM only uses  information from local $l$ by $l$ patches, while the distance between partner plaquettes in the second step is always less than $l$.
Aggregation can thus be performed using only the results from a small number of overlapping local patches.

\begin{figure*}
    \centering
    \includegraphics[width=0.95\textwidth]{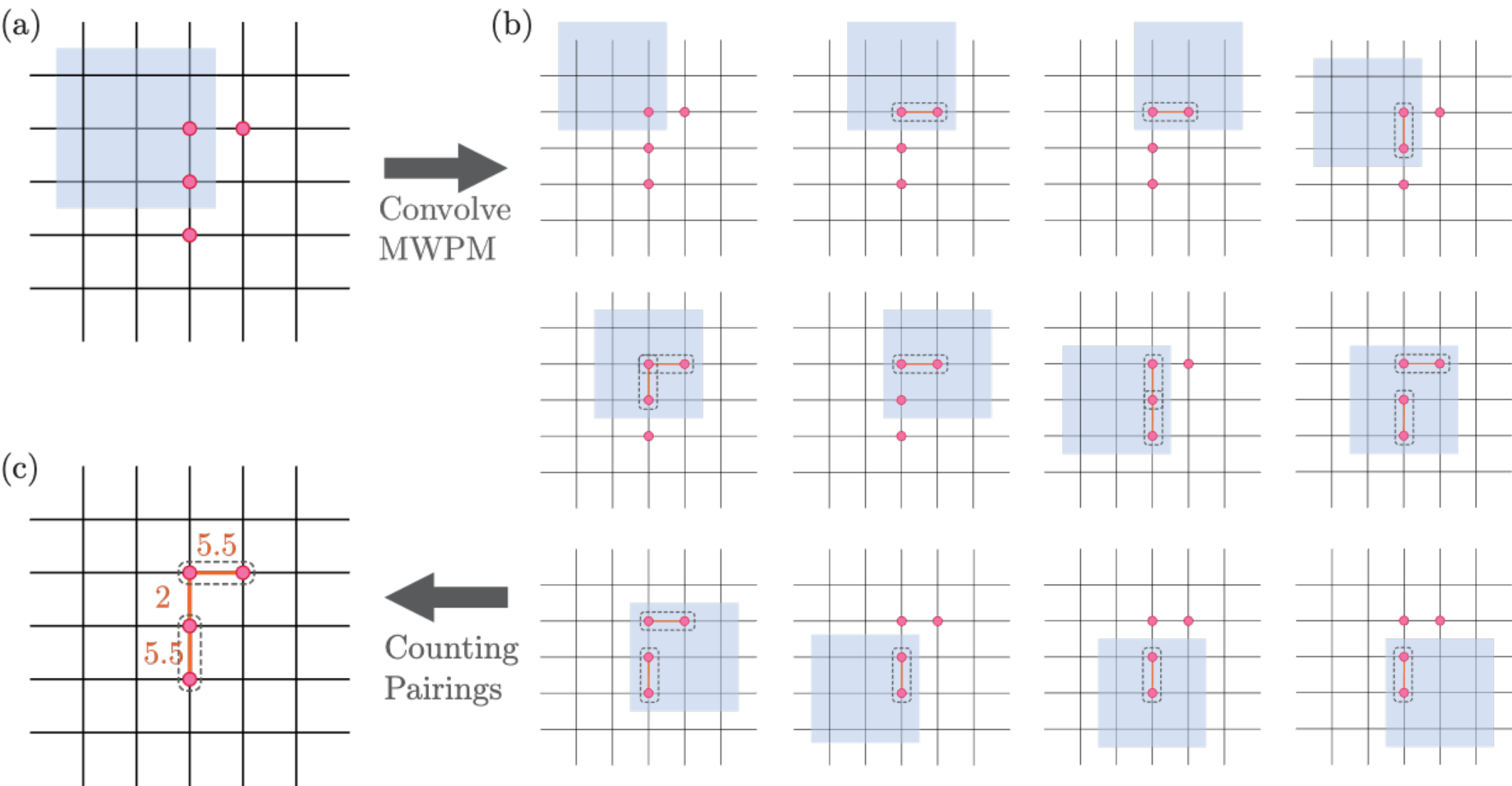}
    \caption{MWPM-based patch decoder. (a) Example of an error chain which creates four $e$-anyons. (b) The decoding algorithm performs correction using only local information by splitting the large system into smaller overlapping regions, within each of which the MWPM algorithm is used to find the lowest-weight pairing of anyons. These local regions have open boundaries, hence MWPM can also pair anyons to the boundaries if this is of lower weight. In practice, a slight boundary bias is added to break ties in favor of boundary pairing. (c) The final step requires locally combining the pairing outputs to determine the final pairing. In particular, we count the number of times each site $p$ is paired to sites $q > p$. In the diagram, two equal-weight pairings contribute 0.5 each, though we randomly break the tie in practice.  Then, the algorithm pairs $p$ with the $q$ that appears most often. In this example diagram, we connect two pairs which have weight$=5.5$, and do not form the weight$=2$ pairing.  We see in the simple four-anyon case depicted above, the procedure correctly recovers the pairing with windows of size $l=3$. In general, this patch-based decoder can correct errors up to distance $d = \lfloor{l/2} \rfloor$; moreover, the distance by which it spreads information and the thickness of any associated LED operators are both proportional to $l$.}
    \label{fig:methods-fig3}
\end{figure*}

\subsection{Decoder Details for the Ruby Lattice Spin Liquid}

\begin{figure*}
    \centering
    \includegraphics[width=0.8\textwidth]{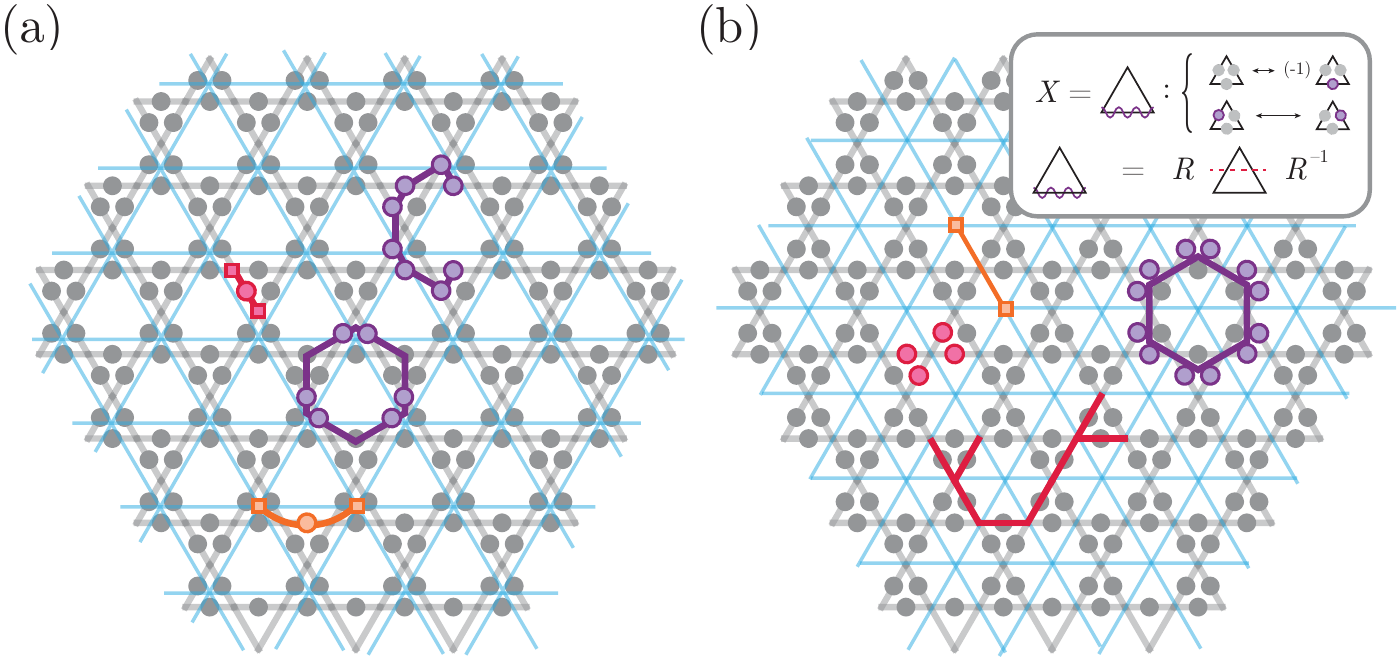}
    \caption{Decoding for the ruby lattice spin liquid realized in Ref.~\cite{semeghini_probing_2021}. (a) For $Z$-loops, two layers of LED can be performed. In both layers, we use the pairing decoder, which flips a qubit (e.g., red or orange circle) if and only if both neighboring stabilizers (e.g., red or orange squares) are equal to $-1$. Stabilizers in the first layer (e.g., red squares) are given by $(-1)\prod_{i \in v} Z_i$ for each vertex $v$ of the kagome lattice. The coarse-graining procedure after the first decoding step maps three stabilizers to a single stabilizer (e.g., orange square) in the coarse-grained lattice (blue lines), whose value is determined by the product of the qubits along a loop enclosing a triangle (e.g., purple closed loop). The open strings considered in the main text start and end at hexagons (e.g., purple open string). 
    (b) To measure $X$-loops, a basis rotation is first performed within each triangle of the kagome lattice, so that the $X$-string operators become diagonal in the measurement basis (inset and Refs.~\cite{semeghini_probing_2021,verresen_prediction_2020}). Each configuration is then mapped to a triangular lattice (blue lines), where each edge of the triangular lattice is determined by the product of four qubits in the original lattice (e.g., red circles); moreover, the $X$ stabilizers of the dimer model become vertex stabilizers in the triangular lattice (e.g., purple hexagons). As before, the pairing decoder flips qubits (orange edges) conditioned on the values of stabilizers (e.g., orange squares).
    Open strings on the triangular lattice also map to open strings in the kagome lattice (e.g., red string), although the resulting strings are slightly different from the ones measured in Refs.~\cite{semeghini_probing_2021,verresen_prediction_2020}.}
    \label{fig:methodsG-fig}
\end{figure*}

We now explain the decoding procedure for a dimer model where qubits lie on the vertices of the ruby lattice, or equivalently, on the links of a kagome lattice. This dimer model supports a $\mathbb{Z}_2$ spin-liquid phase, whose fixed-point is a resonating valence-bond (RVB) state~\cite{verresen_prediction_2020}. This state is in the same universality class as the toric code, as it supports $e$ and $m$ anyons with similar string operators. 

We first describe the decoding procedure for
$e$ anyons, which correspond to vertices with an even number of adjacent dimers~\footnote{Notice that this is an odd $Z_2$ spin liquid, and the trivial empty state corresponds to maximal occupation of $e$ anyon states.}. 
In the first correction step, we apply the pairing decoder between adjacent vertices. We then coarse-grain the kagome lattice to a triangular lattice by grouping vertices within each upward-pointing triangle.
This transforms vertex stabilizers in the kagome lattice to vertex stabilizers in the triangular lattice (Figure~\ref{fig:methodsG-fig}a).
The pairing decoder is then applied between adjacent triangles in the second correction step.
In the main text, we study the flow from uncorrected loops to vertex-paired and triangle-paired loops, which are denoted as as layers 0, 1, and 2, respectively. 

We next consider the $m$ anyons, which are associated with hexagonal plaquettes. A rotation is first performed within each triangle, such that the string operators associated with $m$ anyons become diagonal in the measurement basis. This allows us to map each configuration onto a triangular lattice, whose vertices are located at the center of each hexagon in the kagome lattice; this mapping transforms  $X$-stabilizers of the dimer model into vertex $Z$-stabilizers in the triangular lattice (Figure~\ref{fig:methodsG-fig}b).
Due to the small experimental system size, we can only perform one layer of correction, and we use the pairing decoder on the triangular lattice. 
We note that open strings on the triangular lattice map onto open strings on the ruby lattice, although the resultant strings are slightly different from the ones measured in \cite{verresen_prediction_2020,semeghini_probing_2021}. 

\subsection{Quantum Circuit Formulation of LED}

\begin{figure*}
    \centering
    \includegraphics[width=0.63\textwidth]{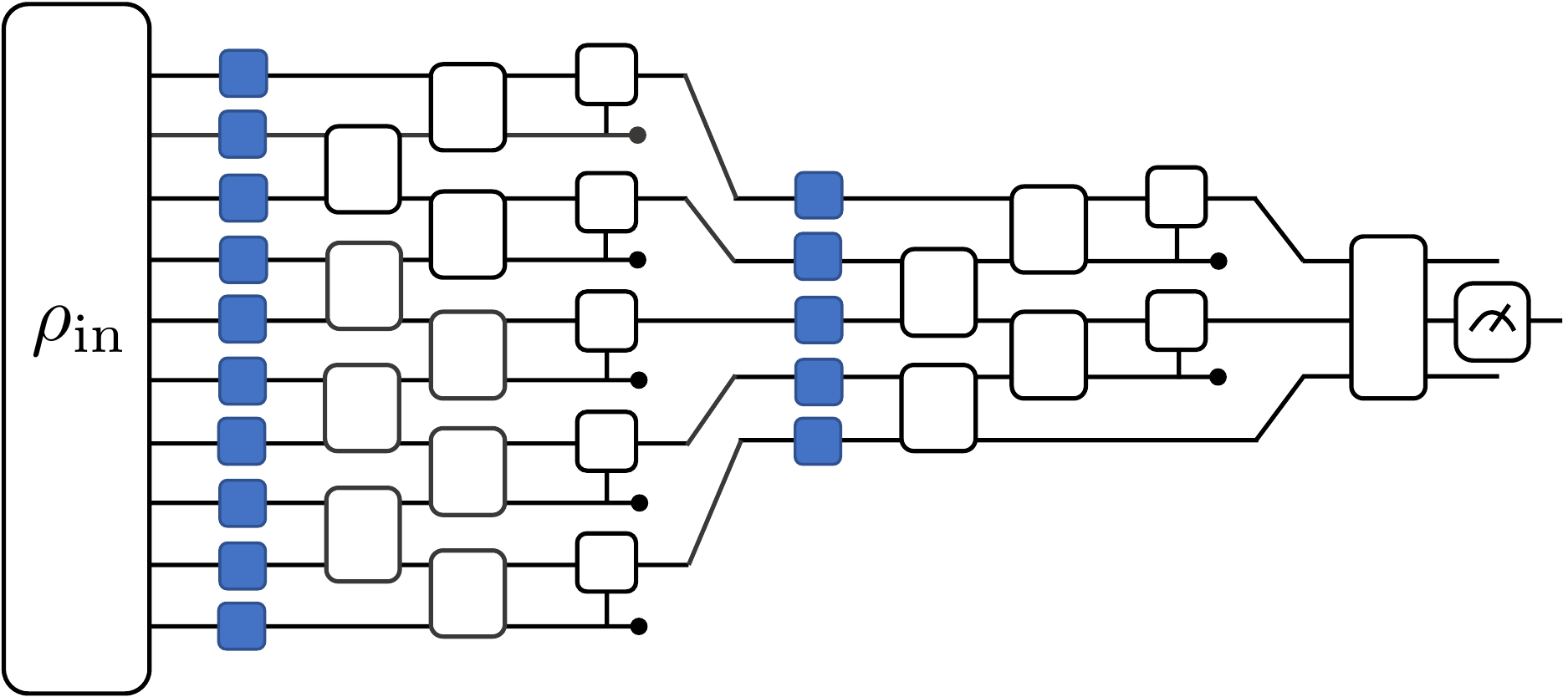}
    \caption{One-dimensional illustration of a hierarchical LED circuit, inspired by the QCNN circuit of Ref.~\cite{cong_quantum_2019}. Here,  stabilizer measurements are performed at each LED layer using local unitary gates (white boxes), and the LED error-correction step is performed through controlled-unitary gates. While this particular circuit model maps stabilizer values in each LED layer to qubits in the initial system ($\rho_{\textrm{in}}$) which are measured in that layer, the stabilizer values can also be obtained by introducing ancillary qubits in a known state (e.g $|0\rangle$) and performing local gates in the same fashion as for surface-code quantum computation~\cite{fowler_surface_2012}. 
    To accommodate states which differ from a known, fixed-point state by local rotations, variational unitary operations (blue boxes) can be introduced before each LED stabilizer-measurement step, and the parameters to these unitaries can be optimized adaptively through a hybrid quantum-classical feedback loop to achieve high LED operator expectation values.}
    \label{fig:LED-circuit}
\end{figure*}

As discussed in the main text, the most general formulation of LED uses a hierarchical quantum circuit like the QCNN circuit introduced in Ref.~\cite{cong_quantum_2019}. The structure of such a circuit is illustrated in Figure~\ref{fig:LED-circuit} in a one-dimensional example for simplicity of illustration, but can be easily generalized to the two-dimensional cases considered in LED.

In this framework, stabilizer measurements are performed at each layer using quantum circuits to preserve the coherence of qubits in the system, in the same fashion as for surface-code quantum computation~\cite{fowler_surface_2012}. 
When the lattice is coarse-grained as in Figures~\ref{fig:fig1},~\ref{fig:methods-fig2}, a fraction of the system's qubits are measured, and local operations are applied to each remaining qubit based on nearby stabilizer measurement values, to correct for local errors. One example circuit construction of LED stabilizer measurement, decoding, and coarse-graining for recognizing the toric code phase is presented in the Supplementary Information. For more general string-net models, ancillas can be used to detect the presence of anyons, and the decoding steps perform anyon transport and fusion via procedures described in Ref.~\cite{zhu_universal_2020}; meanwhile, the coarse-graining circuit is constructed as the inverse circuit of a multiscale entanglement renormalization ansatz (MERA) representation of the fixed-point state (Figure~\ref{fig:fig5})~\cite{konig_exact_2009}. 
Additionally, a layer of variational unitary operations 
is placed in front of each anyon-detection step.

These variational unitaries can  be tuned to optimize the LED order parameter values, especially in the presence of (quasi-)local rotations of qubits on top of a known fixed-point state.
For example, if every qubit in a perfectly-prepared toric code state underwent a Haar-random, single-qubit operation, both the bare and snapshot-based LED Wilson loop operators will be exponentially small. However, the layer of variational unitaries in front of the first local decoding step enables one to ``un-do'' these single-qubit operations and again achieve a high LED signal. In particular, one uses here an adaptive procedure, whereby a hybrid quantum-classical feedback loop is used to tune each unitary to optimize LED loop values. 
More generally, variational unitaries in front of subsequent local decoding steps $l$ allow us to compensate for local operations acting on multiple qubits of the system. 
This is a major step towards achieving a necessary and sufficient criterion for topological order, which is not possible using a single, fixed observable such as a bare Wilson loop operator~\cite{huang_provably_2022}.  Moreover,
due to the special hierarchical structure of QCNN and LED circuits, the optimization of the variational unitaries can be done efficiently without encountering the so-called ``barren plateau'' challenges of variational quantum circuits~\cite{mcclean_barren_2018,pesah_absence_2021}. 

Finally, one other advantage of circuit-based LED is that it enables the simultaneous measurement of loop operators in multiple bases in each experimental repetition. This allows us to capture anyonic braiding statistics, which is critical to the application of LED to non-abelian phases. In particular, the final measurement of $s_{ij}$ (Figure~\ref{fig:fig5}d) can be performed by initializing an ancilla qubit in the state $|+\rangle =\frac{1}{\sqrt{2}} (|0\rangle+|1\rangle)$, and applying a controlled operation which, conditioned on the ancilla being in $|1\rangle$, creates anyon pairs $\alpha_i,\overline{\alpha_i}$ and $\alpha_j, \overline{\alpha_j}$, braids $\overline{\alpha_i}$ around $\alpha_j$,  and fuses the pairs $\alpha_i,\overline{\alpha_i}$ and  $\alpha_j, \overline{\alpha_j}$. $s_{ij}$ is then measured in two steps:
First, the magnitude $|s_{ij}|^2$ is equal to the probability of $\alpha_j, \overline{\alpha_j}$ fusing to vacuum when the ancilla is in $|1\rangle$; this probability can be obtained by measuring local energy densities (e.g., by performing stabilizer measurements). Then, when $|s_{ij}|^2 > 0$, we post-select on both $\alpha_i,\overline{\alpha_i}$ and $\alpha_j, \overline{\alpha_j}$ fusing to vacuum and measure the ancilla's final state
\begin{equation}
    |S_{ij}\rangle \propto |0\rangle + s_{ij} |1\rangle
\end{equation}
in an appropriate basis to obtain the phase of $s_{ij}$.

For topological phases described by abelian quantum double models~\cite{kitaev_fault-tolerant_2003}, the quantum circuit and snapshot-based versions of LED can be combined by measuring all qubits in a fixed basis after some chosen depth $d$, and performing snapshot-based LED using the resulting stabilizer measurement values (see Supplementary Information). 
The choice of $d$ is then determined by a tradeoff between the quantum circuit depth/fidelity and the generality of local rotations which can be compensated for.

\subsection*{Topological Order Witness}

\begin{figure}
    \centering
    \includegraphics[width=0.49\textwidth]{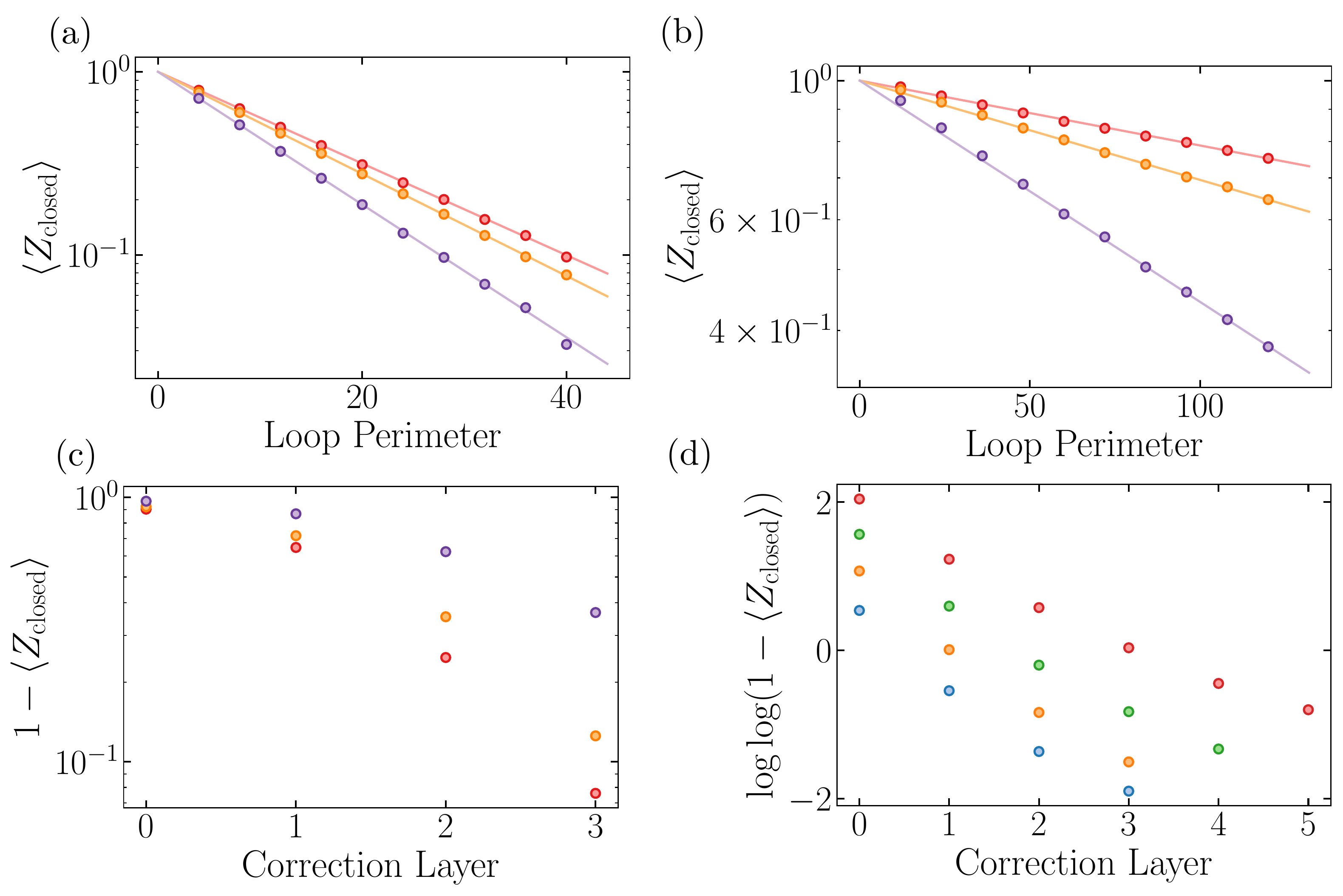}
    \caption{Perimeter-law decay of Wilson loops is clearly visible at various points in the topological phase---(orange) $g_Z=0, g_X=0.18, p_{\mathrm{flip}}=0$, (red) $g_Z=0.18, g_X=0.18, p_{\mathrm{flip}}=0$, (purple) $g_Z=0.10, g_X=0.18, p_{\mathrm{flip}}=0.03$. This is observed for both (a) uncorrected loops and (b) $d=6$ corrected loops under two layers of $d=3$ MWPM patch decoding. (c) LED Wilson loops appear to approach one faster than exponential in $n$. (d) In a model with only incoherent errors ($p_{\mathrm{flip}}=$ $0.02$ (blue), $0.03$ (orange), $0.04$ (green), $0.05$ (red)), we can study the effect of even more layers, where 
    we see hints that the decay is doubly-exponential in $n$, or exponential in $d \sim 2^n$.}
    \label{fig:methods-fig7-plawexpdecay}
\end{figure}

Here, we show that LED provides a topological order witness---that is, it does not misclassify any trivial product state as topological. 
For simplicity, we study the case of $\mathbb{Z}_2$ topological order on a surface with trivial topology, where the fixed-point state is the unique ground state $|\psi_\textrm{TC} \rangle$ of $H_\textrm{TC}$. 
We begin by considering the ideal case where LED operators go to one. 

\begin{theorem}
Let $\rho$ be an arbitrary input state defined on a surface with trivial topology. 
Then, after performing LED with correction distance $d$, assume the resultant state $\rho_f$ has, as a subsystem, qubits living on the links of a square lattice, as in the toric code.
Then, if the stabilizer expectation values $\left\langle \frac{1+A_v}{2} \right\rangle = \left\langle \frac{1+B_p}{2} \right\rangle = 1$ at every vertex $v$ and plaquette $p$ of the subsystem, then, the input state $\rho$ is topologically-ordered, in the sense that it is connected to an output state of the form $\rho_f = \ket{\psi_{\mathrm{TC}}}\bra{\psi_{\mathrm{TC}}} \otimes \alpha_{\mathrm{anc}}$ by generalized local unitary (gLU) transformation of depth $O(d)$. 
\label{thm:localGLU}
\end{theorem}

The key to the proof is a unitary implementation of LED by introducing product state ancillas and performing local unitary gates to perform stabilizer measurement and correction (see Supplementary Information for details). 
These operations, which cannot change the long-range entanglement structure of the state, are known as gLU transformations~\cite{chen_local_2010}, and preserve phase boundaries. Thus, if we further assume the output ancillas $\alpha_{\mathrm{anc}}$ are in a trivial state, Theorem~\ref{thm:localGLU} guarantees the input state is in the toric code phase. 
However, we do not certify this condition holds, which is in general more difficult: measurements in multiple bases are needed to uniquely determine $\alpha_{\mathrm{anc}}$.
Instead, LED certifies that the toric code state can be \textit{distilled} from the input state by gLU transformations. 
Because long-range entanglement cannot be created from a trivial state by gLU transformations~\cite{chen_local_2010}, Theorem~\ref{thm:localGLU} implies that LED operators flowing to unity forms a sufficient condition for topological order, or equivalently, a topological order witness (see also Ref.~\cite{haah_invariant_2016}).

While the above argument works well in theory, any practical system cannot measure LED observables equal to one with infinite precision. Indeed,  
even infinitesimal local perturbations to the toric code ground state, such as $e^{-i \epsilon H} \ket{\psi_\textrm{TC}}$ for arbitrarily small $\epsilon$ and some local Hamiltonian $H$, can create error strings larger than the correction length $d$. This causes LED loop expectation values to decay exponentially, even in the topological phase.
To show that LED still provides a topological order witness in the presence of local perturbations, finite measurement errors, and finite system size, we show the following Theorem:

\begin{theorem}
Consider an arbitrary input state $\rho$ and LED with correction distance $d$, as in Theorem~\ref{thm:localGLU}. 
Suppose the corresponding subsystem of $\rho_f$ has stabilizer expectation values $\left\langle \frac{1+A_v}{2} \right\rangle > 1-\epsilon$, $\left\langle \frac{1+B_p}{2} \right\rangle > 1-\epsilon$ at every vertex $v$ and plaquette $p$. 
Then, the input state $\rho$ exhibits topological ordering at least up to a length-scale $O(\mathscr{L}-d)$; 
that is, no purification of $\rho$ can be prepared using a local quantum circuit of depth less than $O(\mathscr{L}-d)$, where $\mathscr{L} \sim 1/\sqrt{\epsilon}$.
\label{thm:epsilon-witness}
\end{theorem}
Our proof of Theorem~\ref{thm:epsilon-witness} hinges on the following two Lemmas, proved in the supplement. 

\begin{lemma}
Given an output state $\rho_f$ satisfying the conditions of Theorem~\ref{thm:epsilon-witness}, and a simply connected $(\mathscr{L}-2) \times (\mathscr{L}-2)$ square region $R$ on the system part, the reduced density matrix $\rho_d = \tr_{R^c}[\rho_f]$ is indistinguishable from the toric code reduced density matrix $\sigma_{\mathrm{TC}} = \tr_{R^c}[\ket{\psi_{\mathrm{TC}}}\bra{\psi_{\mathrm{TC}}}]$ defined on the same region, up to the bound $|| \rho_d - \sigma_{\mathrm{TC}} || \leq \max\left(\sqrt{\epsilon}, 2 \mathscr{L}^2 \epsilon \right)$
\label{thm:many-body-fidelity}
\end{lemma}

\begin{lemma}
Consider an input state $\rho$ and an LED procedure satisfying the conditions of Theorem~\ref{thm:epsilon-witness}.
Then the final state $\rho_f$ after LED cannot be prepared using a local quantum circuit with depth less than $O(\mathscr{L}) \sim O(1/\sqrt{\epsilon})$.
\label{thm:depth-after-LED}
\end{lemma}

Upon combining the result of Lemma~\ref{thm:depth-after-LED} with the fact that our LED procedure corresponds to a local quantum circuit with depth $O(d)$, we find that the original input state $|\psi \rangle$ cannot be prepared using a quantum circuit of depth smaller than $O(\mathscr{L}-d)$---which is precisely the statement of Theorem~\ref{thm:epsilon-witness}. 
So, if we measure loops of length $L \gg d$ to be $1-\epsilon$, this shows that LED provides a topological order witness up to length-scales of $O(L/\sqrt{\epsilon})$.

We now discuss how these theoretical results are reflected in our numerical simulations. First, when fluctuations are local, the probability of having an error string of length $\ell$ decays exponentially with $\ell$, and the exponent is determined by the characteristic length-scale $\xi$ of fluctuations. 
In these systems, we expect the error rate after an optimal LED procedure with correction distance $d$ to be given by $\epsilon(d) \propto e^{-d/\xi}$, so correction distance $d = \Omega(\xi \log \mathscr{L})$ is sufficient to certify topological order up to length-scale $\mathscr{L}$. Second, when LED uses the hierarchical, anyon-pairing decoder, the anyon density is observed to decrease faster than exponentially in the number $n$ of LED steps (Figure~\ref{fig:methods-fig7-plawexpdecay}). In this case, both the measured stabilizer size and the correction distance $d$ grow exponentially with $n$, which implies that the certification length-scale $\mathscr{L}$ grows at least exponentially with $n$ as well.
Third, our argument does not certify topological order to any length-scale when $L < d$; this is because the support of such an LED operator no longer has an interior, potentially giving rise to signal even in the trivial phase. Indeed, this is reflected in our numerics as well (Figure~\ref{fig:methods-fig5-overcorrection}). 

\subsection{Connection to Topological Entanglement Negativity}

The entanglement negativity of a mixed state $\rho_S$ is defined as $S_N(\rho) = \log ||\rho||_1 = \log(\sum \lambda_i)$, where $\lambda_i$ are the eigenvalues of $\rho$.
Prior works have shown, via a combination of analytical arguments and numerical results, that in a topological phase, $S_N$ obeys an area-law with a constant correction, i.e. $S_N = \alpha L - \gamma$.
Further, recent results have also shown that the topological term $\gamma$ vanishes at finite-temperature~\cite{Lu_negativity_2022}, or for high incoherent error rates~\cite{fan2023diagnostics}.
Thus, the negativity appears to capture important features of mixed state topological order.

The unitary circuit construction of LED also enables us to connect a positive classification under LED, to the topological entanglement negativity of the input state.
In particular, theorem~\ref{thm:localGLU} implies that states classified as topological are connected to an output state $\rho_f = \ket{\psi_{\mathrm{TC}}}\bra{\psi_{\mathrm{TC}}} \otimes \alpha_{\mathrm{anc}}$ via local unitary circuits.
If we further assume the ancillas contain no long-range order (see SM for rigorous definition), then since $\ket{\psi_{\mathrm{TC}}}$ is topologically ordered, the output state indeed has a topological correction in the entanglement negativity.
It is further believed that $\gamma$ is a topological invariant, i.e. it should remain invariant under local unitary circuits. As such, this should be sufficient to certify the input state $\rho_S$ has topological order.

We show this in the SM, for the special case where the LED circuit is composed of Clifford gates, by extending the stabilizer formalism introduced in Ref.~\cite{Lu_negativity_2022}.
Interestingly, there, the topological correction to $S_N$ comes from the presence of decorated Wilson loops operators with non-trivial twist product in the input state $\rho_S$ (see also proof of Lemma~\ref{thm:depth-after-LED}). Thus we conjecture a connection to topological entanglement negativity holds for LED Wilson loops more generally.

\begin{figure}
    \centering
    \includegraphics[width=0.49\textwidth]{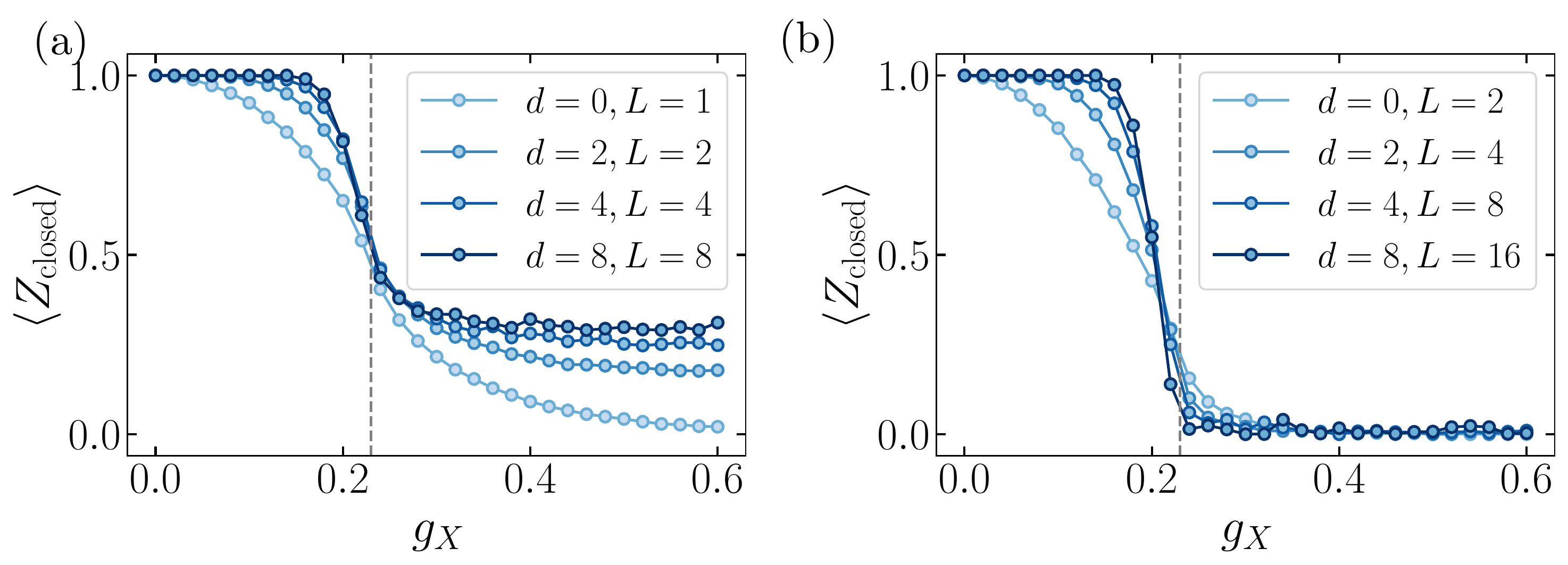}
    \caption{Over-correction for small Wilson loops. (a) Small LED loops with $d>L$ can give non-zero signal deep in the trivial phase, as in this regime, correction can pump anyons from the interior of the annulus to the exterior. (b) However, once $d < L$, the effects of over-correction become insignificant. LED Wilson loops are theoretically expected to certify topological order in the regime of $d \ll L$.}
    \label{fig:methods-fig5-overcorrection}
\end{figure}

%TC:ignore
\section{Supplementary Information}

%\tableofcontents

\subsection{PEPS Sampling Contraction Details}

We start by computing the left boundary MPS of an infinite strip. For sites with $y \leq 0$, we average over measurement outcomes, and hence the doubled PEPS tensor for each site, which contains both the bra and the ket tensors, is $T_{i,j,k,l}^{i',j',k',l'} = \sum_{pq} (A^*)_{i',j',k',l'}^{pq} A_{i,j,k,l}^{pq}$.
The boundary conditions we choose are $\ket{+}_i \ket{+}_{i'}$ at the lower $x=1$ boundary and $\delta_{i,i'}$ at the $x=L_x$ boundary.
Contracting the doubled tensor along an entire column $x=1,...,L_x$ results in a matrix product operator (MPO) $T_y$ acting on the boundary MPS. 
Then, we can compute the effect of an infinite environment by repeatedly applying $T_y$ to some initial boundary MPS until it converges. The right boundary MPS can be similarly computed, by exchanging the input and output directions of the MPO.
Note that at each application of $T_y$, we use singular-value decomposition truncation to prevent the bond dimension of the boundary MPS from growing too rapidly, rendering the method approximate. However, only singular values smaller than $<10^{-8}$ are discarded, so truncation errors should be insignificant. 
We refer to the resulting tensors as the left and right fixed-point of $T_y$.

The algorithm continues by using the left and right boundary MPS to sample a column of sites. Computing the marginal probability $P(ab)$ on the first site requires contracting a 1D tensor network (see Figure~\ref{fig:methods-fig1-PEPS}), where the doubled tensor at $x,y=1,1$ is replaced by a measurement-dependent one $T(ab)_{i,j,k,l}^{i',j',k',l'} =  (A^*)_{i',j',k',l'}^{ab} A_{i,j,k,l}^{ab}$, while the doubled tensor at sites $x > 1$ remain measurement-independent.
Note that the tensor network outputs an unnormalized probability distribution; however, since there are only four states per site, the normalization can be computed with little overhead.
After drawing a sample $ab_{11}$ for the first site, we replace the doubled tensor at $(x,y)=(1,1)$ by $T(ab_{11})$, and then compute the distribution of the second site. The process is repeated until the entire column is sampled.
To minimize repeat 1D contractions, upper and lower environment tensors can be stored and updated during the sweep $x=1,...,L_x$.

Next, the history of samples along the column are used to construct a measurement-dependent MPO $T_y(ab_{1},ab_{2},...,ab_{L_x})$, where the doubled tensor at each site is replaced by a measurement-dependent one.
Then, the left boundary MPS can be updated by contraction with $T_y(\{ab_i\})$, and the process repeated for the second column.
Thus, the left boundary MPS keeps track of the effect of past measurements on future measurements, as the algorithm sweeps from left to right. Meanwhile, the right boundary MPS remains unchanged, as it is modelling a static, infinite environment.

\subsection{Calculation of Phase Diagram}

To compute the phase diagram, we use the PEPS tensors to construct a transfer matrix with periodic boundary conditions, on small cylinders with circumference $L_y$ measured in unit cells.
The largest few eigenvalues of the transfer matrix can be efficiently computed using Krylov algorithms in this regime, and the degeneracy of the largest eigenvalue serves as an alternative signature of the topological transition~\cite{iqbal_entanglement_2021,duivenvoorden_entanglement_2017}.
In particular, the local $\mathbb{Z}_2$ gauge symmetry of the PEPS tensor becomes a $\mathbb{Z}_2 \times \mathbb{Z}_2$ symmetry of the doubled tensor (bra and ket). Hence the topological, $e$-condensate ($Z$-paramagnet), and $m$-condensate ($X$-paramagnet) correspond to three distinct symmetry broken phases from the point of view of the virtual legs, with degeneracy two, one, and four respectively.

As such, along the transition from topological to $e$-condensate, which occurs for $g_X$ small and $g_Z \approx 0.2-0.3$, the relevant ratio is between the first and second eigenvalues, $\Delta_{12}(g_Z, g_X) = \log(\lambda_1 / \lambda_2)$.
In contrast, for the transition from topological to $m$-paramagnet, the relevant gap is between the first and third eigenvalues $\Delta_{13}(g_Z, g_X)$.
Furthermore, the model is self-dual, so the wavefunction at $\ket{\psi(g_z, g_x)}$ is equivalent to the wavefunction at $\ket{\psi(g_x, g_z)}$ by a basis rotation and spatial translation. 
We will use this duality to compute the phase boundary in a way which minimizes finite-size effects, by computing the two gaps at their dual points (Fig~\ref{fig:supp_phase_boundary}a).
Therefore, we introduce two parameters $g_0,g_1$, and let $g_1$ be the parameter which changes across the transition. The relevant parameter is therefore  $\Delta_{12}(g_1, g_0)$ for the topological to $e$-condensate transition, and $\Delta_{13}(g_0, g_1)$ for the topological to $m$-condensate transition.

As $\Delta_{12}$ grows with increasing $g_1$, while $\Delta_{13}$ reduces with increasing $g_1$, these two ratios will eventually cross. Indeed, in the limit $L_y \rightarrow \infty$, the crossing point should exactly correspond to the phase boundary. In general, finite $L_y$ may shift the boundary.
Empirically, we see that at $g_0=0$, the solvable point, there are essentially no finite size effects, and the agreement with the analytical value $g_c = 0.2203434$ is almost exact (Fig.~\ref{fig:supp_phase_boundary}ac). For larger $g_0$, the dependence on $L_y$ appears minimal until around $g_0=0.01$: indeed, even in this case, $L_y=4$ only overestimates the phase boundary compared to $L_y=6$ by a few percent (Fig.~\ref{fig:supp_phase_boundary}d).
As such, we use $L_y=4$ and compute the phase boundary by identifying points where the difference $\Delta_{12}-\Delta_{13}$ is close to zero (Fig~\ref{fig:supp_phase_boundary}b).

\begin{figure}
    \centering
    \includegraphics[width=0.49\textwidth]{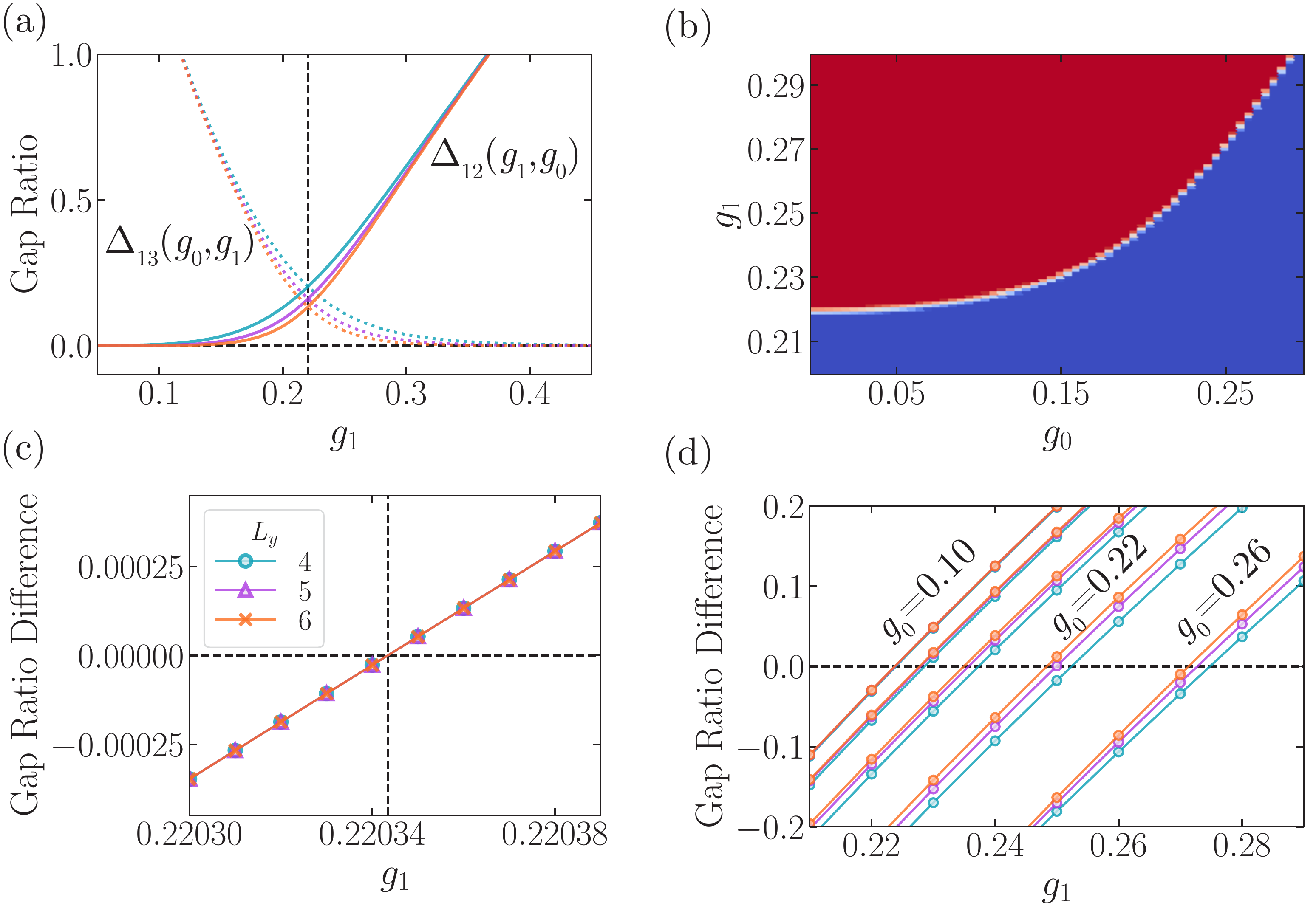}
    \caption{(a) For a given point $(g_0,g_1)$ on the phase diagram, the phase transition from topological to $e$-condensate is signaled by the transfer matrix gap ratio $\Delta_{12}(g_0, g_1)$, while the transition from topological to $m$-condensate is signaled by $\Delta_{13}(g_1, g_0)$, where both quantities are evaluated at the self-dual point with $g_z \leftrightarrow g_x$ switched. When these two gaps are equal, the system must be at the critical point, due to the self-duality. We observe that the crossing point coincides exactly with the known critical point $g_c$ (vertical dotted line), for $L_y=4,5,6$ (blue, purple, orange resp.). (b) Quatitatively, we extract the phase boundary using $L_y=4$ by finding points where the difference $|\Delta_{12}(g_0,g_1) - \Delta_{13}(g_1,g_0)| \leq 0.01$ is close to zero. (c) Along the cut of the phase diagram with $g_0$ small, the phase boundary computed from the difference has no visible finite size effects. (d) For larger $g_0 \geq 0.1$, small finite size effects start to appear. Shown are $g_0=0.1,0.14, 0.18, 0.22, 0.26$ from left to right. The phase boundary clearly shifts to larger $g_1$ with increasing $g_0$, but $L_y=4$ (blue) appears to slightly overestimate the transition point compared to $L_y=6$ (orange).}
    \label{fig:supp_phase_boundary}
\end{figure}

\subsection{Sample Complexity}

\begin{figure*}
    \centering
    \includegraphics[width=0.98\textwidth]{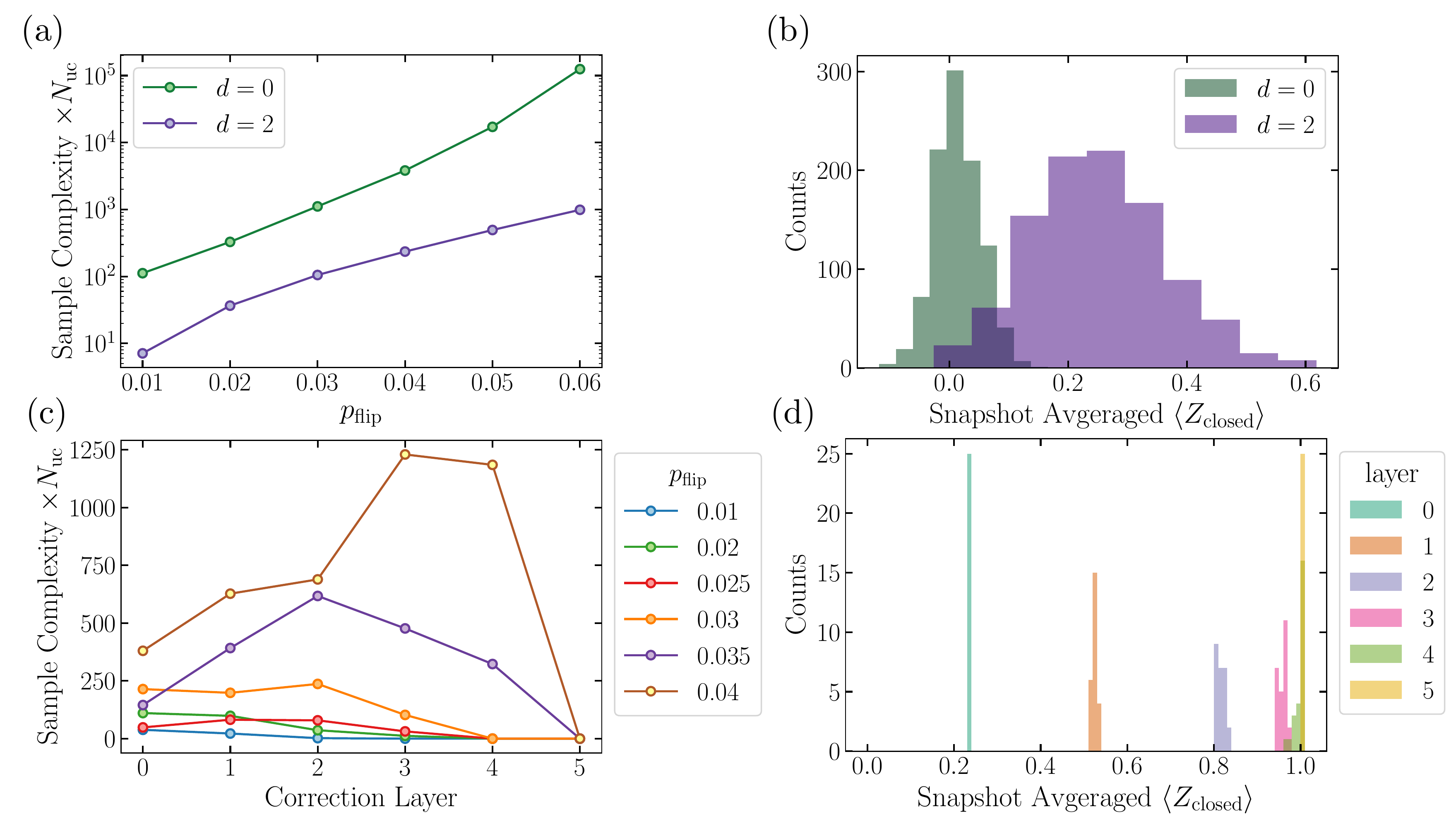}
    \caption{Sample complexity. (a) Effect of LED on loops of fixed size $L=10$. Snapshots are drawn from a toric code defined on a $25 \times 25$ square lattice with various bit-flip error rates. Sample complexity is computed as described in the main text. Note that sample complexities less than one imply that only a fraction of the system needs to be measured to certify topological order with 95\% confidence. As such, the $y$-axis is also multiplied by the number of unit cells $N_{\mathrm{uc}}$.
    (b) The histogram of expectation values averaged over a single snapshot confirms that these expectation values are approximately Gaussian distributed. The distribution for LED loops (purple) has much lower weight at zero than uncorrected loops (green), highlighting how fewer samples are needed to verify with high confidence that the closed loops are non-zero. (c,d) We also study the effect of coarse-graining on sample complexity, where the loop length $L=5 \times 2^n$ grows exponentially with correction layer $n$. Here, snapshots are taken from a $1024\times 1024$ square lattice.  Initially, the sample complexity increases due to a reduction in the number of independent loops available at higher layers. However, it eventually reduces and approaches zero in the topological phase. (d) This turnaround occurs in the limit $\xi \ll d$, where correction is able to remove almost all errors, and loop expectation values approach one. 
    }
    \label{fig:methods-fig6-sampcomp}
\end{figure*}

A key figure of merit for certification of phases is the sample complexity, defined here as the number of samples required to confirm with 95\% confidence that the measured loop operator is non-zero.
To compute sample complexity, we approximate the LED Wilson loop expectation value evaluated on a large but finite size system as a Gaussian random variable. In this scenario, the important quantity is the ratio of the standard deviation $\sigma$ to the mean $\langle Z_{\mathrm{closed}} \rangle$. 
The estimator of the expectation value has a standard deviation that decreases as $1/\sqrt{S}$ where $S$ is the number of samples.
Thus, to confirm the mean is non-zero to two standard deviations (95\% confidence), we require approximately $S = (2\sigma / \langle Z_{\mathrm{closed}} \rangle)^2$ samples.

In Figure~\ref{fig:methods-fig6-sampcomp}, we compare the sample complexity of certifying non-zero Wilson loops using bare and LED observables.
Two scenarios are considered. In the first, bare and LED Wilson loops are compared at fixed length-scale. Indeed, the sample complexity decreases by an order of magnitude for a range of incoherent error rates below the correction threshold.
In the second scenario, coarse-graining is considered, where larger length-scales are probed at each layer.
Here, the sample complexity in fact increases at early layers for moderate error rates before falling dramatically. 
This is because the variance of the signal initially increases, since coarse-graining reduces the number of loops available for averaging in a fixed size system. 
However, for sufficient correction layers, LED reliably removes almost all errors. This is the regime where $\xi \ll d$, so the Wilson loops saturate at one and their variance approaches zero (see histograms in Figure~\ref{fig:methods-fig6-sampcomp}).

Nevertheless, the initial increase in sample complexity is not simply due to information being removed. 
Upon further examination of the scenario
without coarse-graining, we find a similar initial increase in sample complexity. We interpret this as correction causing adjacent LED loops to become correlated.
Finally, we note that the sample complexity measured in this way only improves in the topological phase. In the uncorrectable, disordered phase, the inverse ratio $\langle Z_{\mathrm{closed}} \rangle / \sigma$ rapidly approaches zero.

\subsection{Decoder dependence}

\begin{figure*}
    \centering
    \includegraphics[width=0.85\textwidth]{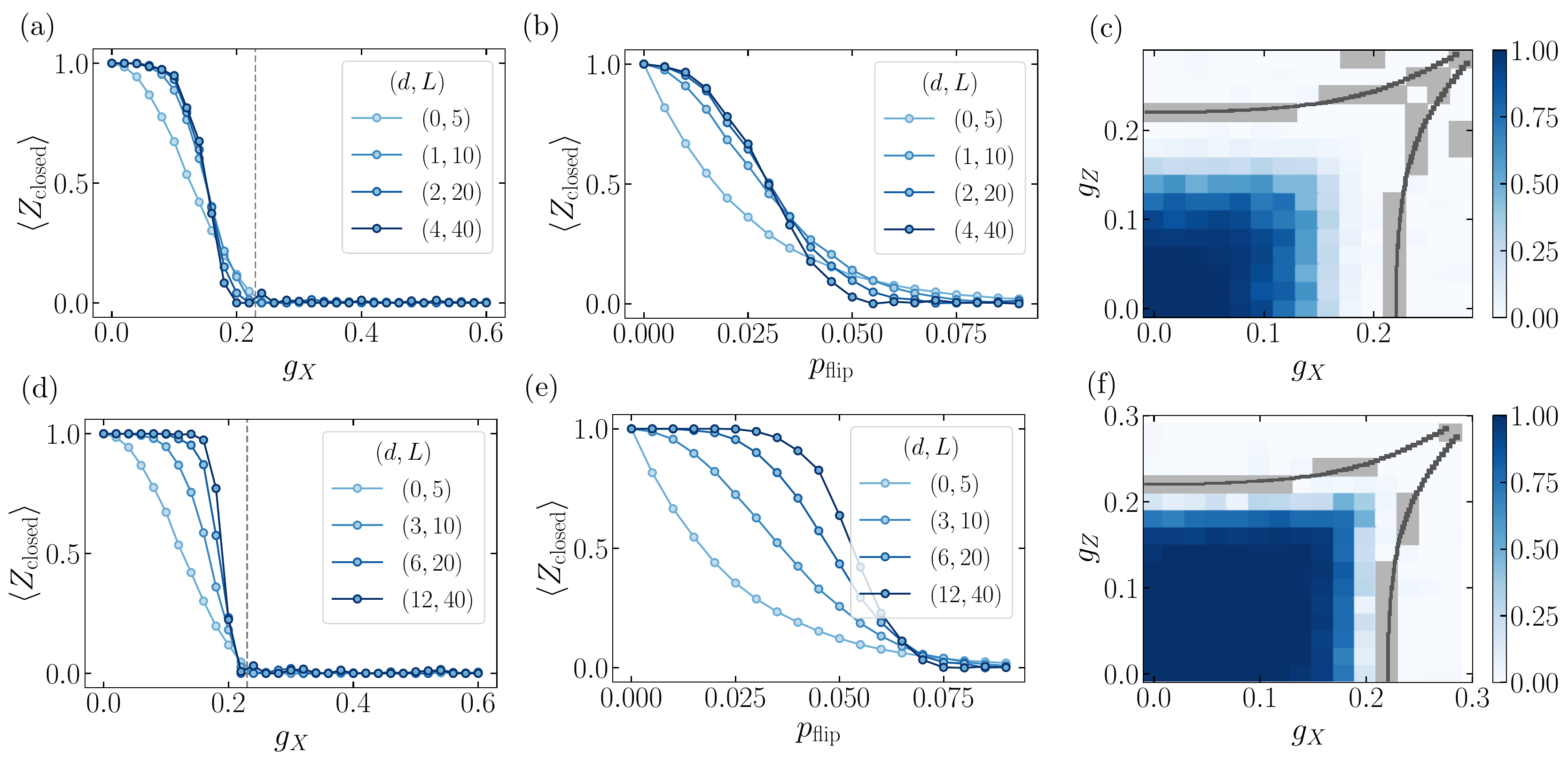}
    \caption{Decoder choice can significantly affect the extent of the region certified as topological. Here we show LED Wilson loops using coarse-graining with block size $b=2$ together with (a-c) a $d=1$ pairing decoder and (d-f) a $d=3$ patch decoder. Coherent perturbations are considered in panels (a,d), while incoherent errors are studied in (b,e). Panels (c,f) show LED expectation values for various combinations $g_X$,$g_Z$ after three layers of correction. Note that in the main text, analogous plots are made using the $d=2$ patch decoder, whose classification boundaries roughly match the $d=3$ decoder shown here.}
    \label{fig:methods-fig4}
\end{figure*}

We now examine how different choices of LED decoders can change the size of the ``correctable'' region---that is, the region classified as topological. The main text demonstrates results for an $l=4$ (or 
$d=2$) patch-based decoder with coarse-graining size $b=2$; in Figure~\ref{fig:methods-fig4}, we compare this with the pairing decoder and
an $l=6$ ($d=3$) patch-based decoder which also uses $b=2$. 
We see clearly that the $l=4$ and $l=6$ decoders both produce significantly larger correctable regions than the simple pairing decoder, 
for both coherent and incoherent errors.
Meanwhile, the $l=4$ and $l=6$ decoders perform similarly, so it appears that the decoder threshold saturates with $l$.
Intriguingly, we observe saturation at an
incoherent error rate which is significantly below the known error correction threshold of $p_c \approx 10.9\%$.
An interesting open question is to determine
whether this discrepancy arises because the patch-based decoder is suboptimal, or because local decoders have some fundamental limit.
In the Supplementary Information, we show that a ``annulus-based" decoder which applies MWPM in a non-local fashion results in a much larger correctable regime.

\subsection{Definition and Properties of Fixed-Point States}

As discussed in the main text, one important component for defining any LED procedure is to identify a fixed-point state of the target phase of matter. Here, we review the definition and key properties of a fixed-point state.

It is well-known in the literature that gapped quantum ground states can be classified by a real-space RG flow~\cite{chen_local_2010,schuch_classifying_2011}. To implement such an RG flow, one notes that any two states in the same phase can be connected by finite-depth local unitary transformations. The hallmark property of topologically-ordered states is the presence of long-range entanglement, and finite-depth local unitary transformations can add or remove local short-range entanglement; thus, for any given state in a topological phase of matter, one can construct a procedure which hierarchically removes all short-range entanglement from this state at increasing length-scales. Then, the resulting state then has zero correlation length, as all short-range entanglement has been removed; this state is known as a fixed-point state of the topological phase.

We also consider important properties of the fixed-point state in the context of our LED procedure and the closely related QCNN procedure of Ref.~\cite{cong_quantum_2019}. In both of these procedures, a fixed-point state of the phase under consideration is chosen, and the protocol identifies states within this phase by removing local errors or perturbations on top of this fixed-point state; this is done by performing a decoding operation and a coarse-graining operation, and repeating them $n$ times. In these settings, the fixed-point state has a few special properties: If the input state is equal to the fixed-point state, no errors are detected within each decoding step, and furthermore, the state after each layer of decoding and coarse-graining is equal to the input state, defined on a subsystem of the original system with fewer qubits. Finally, because fixed-point states have zero correlation length, the bare Wilson loops defined in the main text have expectation value equal to $+1$ without performing any correction or LED.

\subsection{LED Circuit for $\mathbb{Z}_2$ Toric Code}

\begin{figure}
    \centering
    \includegraphics[width=0.49\textwidth]{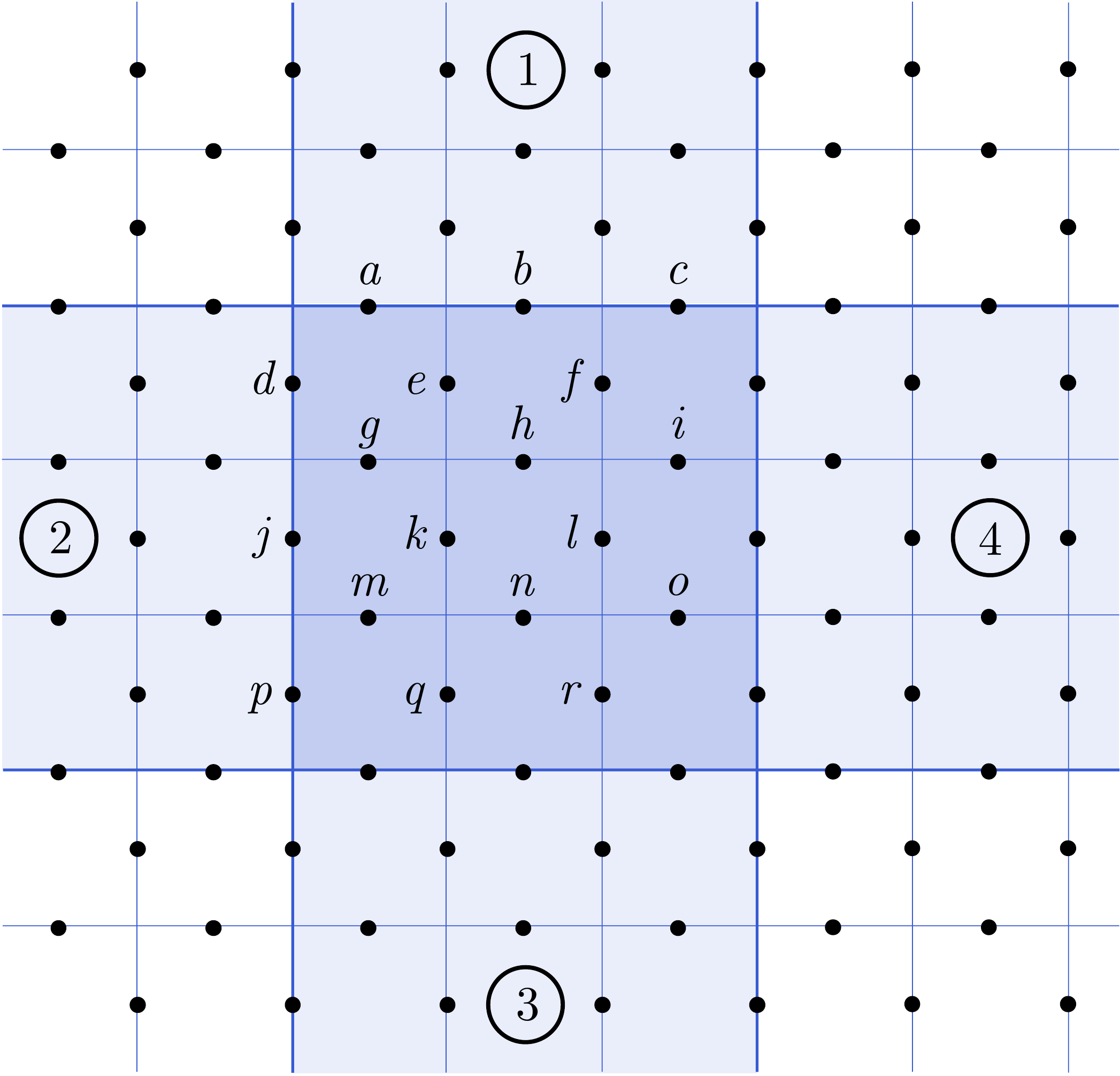}
    \caption{Setup for a quantum circuit which implements LED for the $\mathbb{Z}_2$ toric code phase without introducing any ancillary qubits. The coarse-graining performed by this circuit combines nine unit cells into a single, large unit cell in the subsequent layer (thick vs. thin lines); we label each of the 18 qubits in this large unit cell using the letters $a$ through $r$ as shown above. Here, we focus on the local unitaries involving the qubits within a single, large unit cell (darker blue shaded cell); however, some of the local unitaries in our circuit span two large unit cells, so we also consider the neighboring cells (lighter blue shaded cells). We label the neighboring cells 1 through 4, and refer to a qubit in such a cell by its letter and the number of the cell (e.g., the uppermost qubits of unit cell 2 would be $a_2$, $b_2$, and $c_2$ from left to right).}
    \label{fig:tc-circuit}
\end{figure}

In the quantum circuit formulation, the stabilizer measurement, local decoding, and coarse-graining steps of LED are implemented through local controlled-unitary gates, single-qubit rotations, and (optionally) measurements and local feed-forward operations. Meanwhile, local variational unitary gates are introduced before each LED layer to enable efficient and robust identification of a large set of topologically-ordered states (see Methods and Figure~\ref{fig:LED-circuit}). The quantum circuit for stabilizer measurements, local decoding, and coarse-graining can be implemented in different ways, as discussed below.

As a first example, one can introduce an ancillary qubit for each stabilizer, and perform local controlled-NOT (CNOT) gates between the system qubits and the ancillary qubit at each vertex or plaquette. These local CNOT gates are designed in the exact same fashion as stabilizer measurements of the surface code (see, for example, ~Ref.~\cite{fowler_surface_2012}). Decoding can then be implemented either by measuring the ancillary qubits and performing local feed-forward operations on the system qubits to correct for arbitrary single-qubit errors, or by using local controlled-unitary operations between the ancillary and system qubits to achieve the same result. Finally, the MERA circuit of Ref.~\cite{aguado_entanglement_2008} can be used to perform coarse-graining at each LED layer.

Alternatively, one can avoid introducing ancillary qubits by carefully constructing a circuit which maps stabilizer values in each layer to the qubits which are removed during the coarse-graining process of that layer. This scheme combines the stabilizer measurement and coarse-graining procedures together into one large set of unitary operations. As before, the decoding step can be implemented by measuring stabilizer qubits and performing local feedforwarding to the remaining qubits, or by using local controlled-unitary operations between the ancillary and system qubits.

Because the number of removed qubits in each unit cell is always smaller than the number of stabilizers in this second case, only a fraction of stabilizers can be measured in every layer. Thus, one must design the circuit meticulously in order to still correct for all single-qubit errors. 
As a concrete example, we illustrate here such a quantum circuit implementation of the LED stabilizer measurement, local decoding, and coarse-graining procedures for the $\mathbb{Z}_2$ toric code which does not require additional ancillary qubits. 

In this circuit, nine unit cells are combined to a single, large unit cell in the next layer by the coarse-graining procedure (i.e., a 3-to-1 reduction is performed in each dimension). The circuit is invariant under translation by the large unit cells, so we  consider only the gates involving a single large unit cell. The setup and notations for qubits are defined in Figure~\ref{fig:tc-circuit}.

Our circuit consists primarily of two-qubit CNOT gates, together with some single-qubit rotations and multi-qubit controlled-unitary operations. Thus, to compactly specify our circuit, we introduce some notations for the controlled-unitary operations of our circuit: specifically, we use the notation (control qubits $\mapsto$ target qubit) to refer to the gate which performs a bit-flip ($X$) gate on the target qubit if and only if all control qubits are in the $Z=-1$ state, and identity otherwise. For example, $(a \mapsto b)$ denotes a CNOT gate where $a$ is the control qubit and $b$ is the target qubit. Likewise, $(a,b \mapsto c)$ denotes a Toffoli gate where $a$ and $b$ are the control qubits and $c$ is the target qubit.

With all notations defined, we are now ready to specify our circuit. Due to translation invariance, it is understood that whenever we list a gate here, it will also be applied simultaneously on all qubits which differ from the labeled sites by translation of an integer number of large unit cells; thus, our circuit consists of many layers of non-overlapping local gates. The circuit proceeds as follows:
\begin{enumerate}
    \item Perform layers of CNOT gates in the following order: $(i \mapsto o)$, $(l \mapsto o)$, $(j_4 \mapsto o)$, $(f \mapsto i)$, $(c \mapsto i)$, $(d_4 \mapsto i)$, $(b \mapsto f)$, $(e \mapsto f)$, $(h \mapsto f)$, $(a \mapsto e)$, $(d \mapsto e)$, $(g \mapsto e)$, $(j \mapsto g)$, $(m \mapsto g)$, $(k \mapsto g)$, $(p \mapsto m)$, $(a_3 \mapsto m)$, $(q \mapsto m)$, $(q \mapsto n)$, $(b_3 \mapsto n)$, $(r \mapsto n)$, $(q \mapsto k)$, $(r \mapsto l)$, $(b_3 \mapsto h)$, $(k \mapsto h)$, $(l \mapsto h)$, $(a \mapsto b)$, $(c \mapsto b)$, $(d \mapsto j)$, $(p \mapsto j)$. 
    \item Perform layers of Toffoli gates in the following order: $(e,m_1 \mapsto b)$, $(f, n_1 \mapsto b)$, $(i_2, e \mapsto j)$, $(g, o_2 \mapsto j)$.
    \item Perform single-qubit $X$ gates on $e$, $f$, $g$, $h$, $n$, $o$.
    \item Perform two layers of multi-qubit controlled-unitary gates $(i, e_4, f, o \mapsto b)$, $(m, e_3, g, h \mapsto j)$
    \item Perform single-qubit Hadamard gates on $a$, $b$, $c$, $d$, $j$, $k$, $l$, $p$.
    \item Perform layers of Toffoli gates in the following order: $(a,c \mapsto b)$, $(d,k \mapsto j)$, $(c_3, l \mapsto b)$, $(c_3, l \mapsto j)$, $(c_3, l \mapsto b_3)$, $(d, p \mapsto j)$, $(r, c_3 \mapsto b_3)$, $(r, p_4 \mapsto b)$, $(r, p_4 \mapsto j)$, $(r, p_4 \mapsto j_4)$.
    \item Perform single-qubit Hadamard gates on $b$ and $j$.
\end{enumerate}

\subsection{LED for Abelian Quantum Double Models}

\begin{figure}
    \centering
    \includegraphics[width=0.38\textwidth]{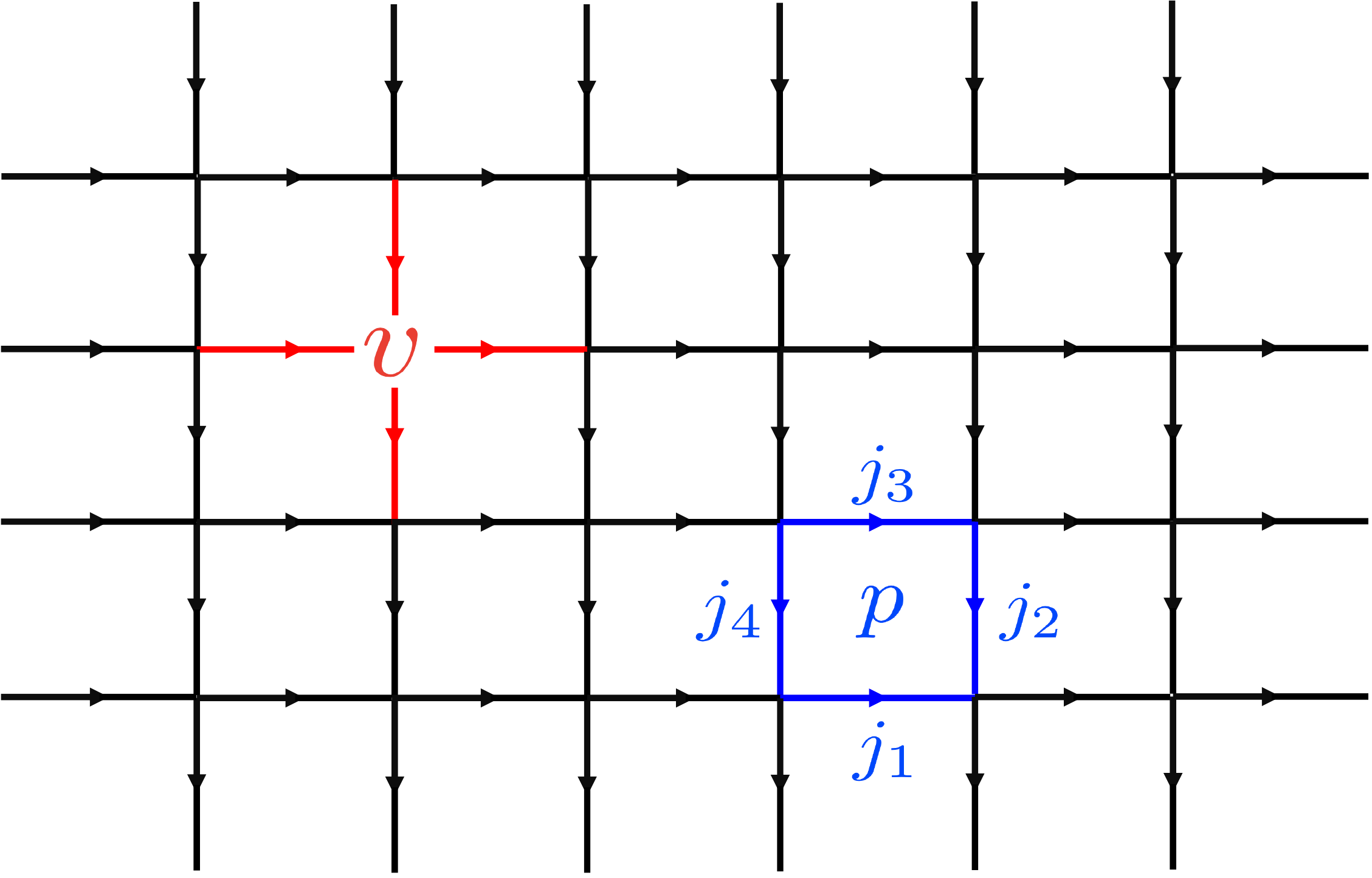}
    \caption{Setup for Kitaev's quantum double models on a square lattice~\cite{kitaev_fault-tolerant_2003}. For generic groups $G$, one must define an orientation for each edge; for simplicity, we choose edges to be oriented upwards and to the right.}
    \label{fig:non-abelian-lattice}
\end{figure}

The snapshot-based LED framework can be extended to enable detection of Kitaev's quantum double models based on any finite abelian group $G$. Since the quantum double model $\mathfrak{D}(G)$ of a direct product of finite groups $G = G_1 \times G_2 \times ... \times G_k$ is equivalent to the model where $k$ quantum doubles $\mathfrak{D}(G_1)$, ..., $\mathfrak{D}(G_k)$ are stacked together, it suffices to consider the quantum double of cyclic groups 
$G_d = \{ 0,1,...,d-1\}$. 
In such models, an orientation must be assigned to each link of the square lattice, which is occupied by a {\it qudit} belonging to the Hilbert space spanned by $\{ |j\rangle, j \in G_d \}$ (Figure~\ref{fig:non-abelian-lattice}). To define the vertex and plaquette stabilizers for these quantum double models, we first introduce the single-qudit {\it shift} and {\it clock} operators---$X_\pm$ and $Z_\pm$, respectively---which are generalizations of the Pauli matrices $X$ and $Z$ for qubits. We define these operators based on their action on the basis states:
\begin{equation}
X_\pm |j\rangle = |j \pm 1 \rangle
\end{equation}
\begin{equation}
Z_\pm |j\rangle = \omega^{j} |j\rangle.
\end{equation}
where addition is performed modulo $d$, and $\omega = e^{2 \pi i/d}$ is a $d^{\textrm{th}}$ root of unity. 
Using these operators, we can define generalized vertex and plaquette operators:
\begin{equation}
\label{eq:kitaev-vertex-g-term}
A(v) = \prod_{j \in \text{adj}(v)} X(j,v)
\end{equation}
\begin{equation}
\label{eq:kitaev-plaquette-h-term}
B(p) = \prod_{j \in \text{adj}(p)} Z(j,v).
\end{equation}
Here, the $X$ and $Z$ operators depend on the orientation of the edge $j$ relative to the vertex $v$ or plaquette $p$ under consideration: if the directed edge $j$ points away from $v$, $X(j,v) = X_- (j)$, otherwise $X(j,v) = X_+(j)$; if $p$ is on the left (resp., right) of the directed edge $j$ when the lattice is rotated such that  $j$ points upwards, $Z(j,p) = Z_-(j)$ (resp., $Z_+(j)$)~\cite{kitaev_fault-tolerant_2003}. 
The Hamiltonian for this model is then 
\begin{equation}
\label{eq:kitaev-hamiltonian}
H = -\sum_v \sum_{k=0}^{d-1} A(v)^k - \sum_p \sum_{k=0}^{d-1} B(p)^k.
\end{equation}

As in the case of the toric code Hamiltonian (Equation~(\ref{eq:tc-ham}) in the main text), all $A^k(v)$ and $B^k(p)$ terms in the above Hamiltonian $H$ commute with each other. Each $A(v)$ and $B(p)$ has $d$ possible eigenvalues: 1, $\omega$, $\omega^2$, ..., $\omega^{d-1}$. The ground state(s) of $H$ are then simultaneous $+1$ eigenstate(s) of all $A(v)$ and $B(p)$. Meanwhile, vertex and plaquette violations correspond to anyonic excitations: a plaquette where $B(p) = \omega^a$ hosts an $e^a$ anyon, while a vertex where $A(v) =\omega^b$ hosts an $m^b$ anyon. Thus, at any site $(v,p)$, there are $d^2$ possible topological charges $e^a m^b$ ($a,b \in \{0, 1, ... , d-1\}$).

We now demonstrate how a generalized snapshot-based LED procedure can be used to recognize the quantum double phase. In this case,
we begin by measuring all qudits in the basis $\{ |0\rangle, |1\rangle, ..., |d-1\rangle\}$, which we refer to as the {\it group basis}. Such a measurement allows us to compute all plaquette terms $B(p)$, and identify $e$-type anyons. The location of such an anyon can be shifted one cell away by applying $X_\pm^{B(p)}$ to one of the edges in $\textrm{adj}(p)$ (the sign depends on the edge's orientation relative to the plaquette): for example, if the anyon is located on the plaquette $p$ in Figure~\ref{fig:non-abelian-lattice}, applying $X_-^{B(p)}$ to $j_3$ will move the anyon up by one cell. 
The patch-based decoder we use for LED thus corrects errors by grouping together, when possible, two or more plaquettes $p_1, p_2, ...$ with non-trivial $B(p_i)$ such that the product $\prod_i B(p_i) = 1$; this can be implemented simply by multiplying qudits by group elements $B(p_i)^{\pm 1}$. The groupings are chosen to minimize the total number of qudits to modify, while still removing as many errors as possible within each patch. 

The above procedure allows us to correct for $e$-type errors. The same LED procedure can be performed to address $m$-type errors, by measuring qudits in another basis---the {\it representation basis}. Representation-basis measurements are performed by first applying a generalized Hadamard operator
\begin{equation}
\label{eq:generalized-hadamard}
H_d = \frac{1}{\sqrt{d}}\sum_{a=0}^{d-1} \sum_{b=0}^{d-1} \omega^{ab} |a\rangle\langle b|,
\end{equation}
and then measuring in the group basis. Because the shift operator is diagonal in the representation basis, measurements in this basis allow us to identify $m$-type anyons, and utilize the patch-based decoding scheme described above for LED.

\subsection{Background on Generic Topological Phases}

Here, we provide background on generic topological quantum field theories (TQFTs), which are characterized by modular tensor categories $\C$. The possible topological charges (a.k.a. anyon types) in such a system are given by the simple objects $\{ \alpha_0, \alpha_1, \alpha_2, ... \}$ of $\C$, where $\alpha_0 = \mathbf{1}$ is the trivial or vacuum topological charge.  Abelian anyons $\alpha_i$ have quantum dimension $d_i = 1$, meaning that the outcome of fusing $\alpha_i$ with any other anyon $\alpha_j$ is deterministic: $\alpha_i \otimes \alpha_j = \alpha_{k(i,j)}$ for some integer $k(i,j)$. On the other hand, non-abelian anyons $\alpha_i, \alpha_j$ have quantum dimensions $d_i, d_j > 1$, and their fusion can result multiple possible anyon types as governed by {\it fusion rules}
\begin{equation}
\label{eq:fusion-rules}
\alpha_i \otimes \alpha_j = \oplus_k n^k_{ij} \alpha_k.
\end{equation}
Here, the fusion coefficients $n^k_{ij}$  must satisfy
\begin{equation}
    d_i d_j = \sum_k n^k_{ij} d_k.
\end{equation}
Moreover, each anyon $\alpha_i$ has a unique {\it conjugate} anyon $\alpha_j = \overline{\alpha_i}$ for which $n^{\mathbf{1}}_{ij} = 1$ (i.e., $\alpha_i$ and $\overline{\alpha_i}$ can annihilate each other by fusing to vacuum); for all other $\alpha_k$, $n^{\mathbf{1}}_{ik} = 0$.

In addition to anyon types and fusion rules, several other quantities are needed to characterize a TQFT. In particular, for every pair of anyons $\alpha_i, \alpha_j$, one can compute the {\it Hopf link}
\begin{equation}
\label{eq:hopf-link}
\vcenter{\hbox{\includegraphics[width = 0.225\textwidth]{modular_s_2.pdf}}}.
\end{equation}
Moreover, for each anyon $\alpha_i$, its {\it topological twist} is defined as
\begin{equation}
\label{eq:twist}
\vcenter{\hbox{\includegraphics[width = 0.16\textwidth]{modular_t_2.pdf}}}.
\end{equation}
These quantities are used to define the modular $S$ and $T$ matrices of $\C$: 
\begin{equation}
    S_{ij} = \frac{1}{\mathcal{D}} \tilde{s}_{ij}, \qquad T_{ij} = \delta_{ij} \theta_i.
\end{equation}
where $\mathcal{D} = \sqrt{\sum_i d_i^2}$ is the global quantum dimension of $\C$.

It is conjectured that the modular $S$ and $T$ matrices uniquely define a unitary modular tensor category, or equivalently a topological phase of matter~\cite{wang_topological_2010}.  In the main text and Methods, we illustrate the application of LED to Levin and Wen's string-net models based on arbitrary unitary fusion categories $\mathcal{A}$. The MTC describing such a string-net model is the Drinfeld center $\mathcal{C} = \mathcal{Z}(\mathcal{A})$~\cite{levin_string-net_2005}. Additional background on general TQFTs and MTCs can be found in Refs.~\cite{bakalov_lectures_2001,wang_topological_2010}.

\subsection{Arbitrary Local Perturbations}

In this section, we show that LED Wilson loops flow to one for
any state which differs from the fixed-point state by an arbitrary local perturbation. 
This implies that LED loop operators are independent of the exact perturbation, unlike the fattened Wilson loops of Refs.~\cite{hastings_quasiadiabatic_2005,levin_detecting_2006}. For concreteness, we examine perturbations on top of a toric code ground state. 

To prove our claim, we consider a 
local unitary operator $\mathcal{O}$ supported on a local region $A$ of diameter $l$. 
It follows that $\mathcal{O}$ can only flip a stabilizer from $+1$ to $-1$ if it overlaps with $A$, and $\mathcal{O}$ cannot couple any ground state $|\psi \rangle$ to another ground state $|\psi'\rangle \neq |\psi \rangle$. 
We now show that LED removes all flipped stabilizers after $1+\log_b d$ layers, where $b$ is the coarse-graining length-scale.
Because the coarse-graining step effectively reduces $l \rightarrow l/b$, after $\log_b d$ layers, there are three possibilities: (1) $A$ becomes fully contained within a single $b \times b$ region at some layer $c < \log_b d$. Then, $A$ has zero support after another layer of coarse-graining and  disappears. 
(2) Before iteration $\log_b d$, $\mathcal{O}$ is supported on two adjacent $b \times b$ regions. Then, $\mathcal{O}$ becomes a single-qubit error after this iteration and is removed by the subsequent LED step.
(3) Before iteration $\log_b d$, $\mathcal{O}$ is supported at the corner of three or four regions. In this case, it becomes a two-qubit diagonal error after this iteration, which can also be removed by the subsequent LED step. Notice that handling case (3) requires the inclusion of diagonal pairing in the pairing decoder. Finally, while this proof focuses on the pairing decoder, it also generalizes directly to 
more advanced local decoders, such as the patch-based decoder, when they are combined with coarse-graining.

\subsection{Proof of Theorem 1}

\textbf{Theorem 1.} 
\textit{Let $|\psi\rangle$ be an arbitrary input state defined on a surface with trivial topology. 
Then, after performing LED with correction distance $d$, assume the resultant state $|\psi_d \rangle$ has, as a subsystem, qubits living on the links of a square lattice, as in the toric code.
Then, if the stabilizer expectation values $\left\langle \frac{1+A_v}{2} \right\rangle = \left\langle \frac{1+B_p}{2} \right\rangle = 1$ at every vertex $v$ and plaquette $p$ of the subsystem, then, the input state $|\psi \rangle$ is topologically-ordered, in the sense that it is connected to an output state of the form $\ket{\psi_d} = \ket{\psi_{\mathrm{TC}}} \otimes \ket{\phi_{\mathrm{anc}}}$ by generalized local unitary (gLU) transformation of depth $O(d)$.}

\begin{proof}
The LED procedure forms a local quantum channel, and we begin our proof by
constructing a purification of this channel.
To mediate stabilizer measurement and local error correction, one can first introduce an ancilla in the state $\ket{0}$ at every vertex and plaquette.
Next, a sequence of Hadamard and controlled-NOT (CNOT) gates is applied such that a $Z$-basis measurement on an ancilla is equivalent to the associated stabilizer measurement of $A_v$ or $B_p$.
Then, local quantum error correction is performed using a local unitary evolution on the combined system, which contains the original state  and the added ancilla qubits. This local unitary evolution applies gates which perform $X$ and $Z$ spin flips on the system qubits, conditioned on the state of the ancilla qubits. Finally, the coarse-graining step can also be performed with local unitary transformations by using a quantum circuit corresponding to a multiscale entanglement renormalization ansatz (MERA) representation of the fixed-point state~\cite{aguado_entanglement_2008}.
The transformations generated by introducing product state ancillas and performing local unitary operations are called generalized local unitaries (gLU); this class of transformations includes our LED procedure described above and is known to preserve phase boundaries~\cite{chen_local_2010}. 

If the system part of the final state, $|\psi_d\rangle \in \mathcal{H}_{\mathrm{sys}} \otimes \mathcal{H}_{\mathrm{anc}}$, has stabilizer expectation values $\left\langle \frac{1+A_v}{2} \right\rangle = \left\langle \frac{1+B_p}{2} \right\rangle = 1$ at every vertex $v$ and plaquette $p$, then $|\psi_d\rangle$ must belong to the ground state space of the toric code. This is because the projector onto the ground state space is given by the product of all the stabilizers $P_\textrm{TC} = \prod_v \frac{1+A_v}{2} \prod_p \frac{1+B_p}{2}$.  
On a surface with trivial topology, there is a unique state $\ket{\psi_{\mathrm{TC}}}$, so the output state factors into $\ket{\psi_d} = \ket{\psi_{\mathrm{TC}}} \otimes \ket{\phi_{\mathrm{anc}}}$.
\end{proof}

\subsection{Proof of Lemmas 3 and 4}

\textbf{Lemma 3.}
\textit{Given an output state $\ket{\psi_d}$ satisfying the conditions of Theorem~\ref{thm:epsilon-witness}, and a simply connected $(\mathscr{L}-2) \times (\mathscr{L}-2)$ square region $R$ on the system part, the reduced density matrix $\rho_d = \tr_{R^c}[\ket{\psi_d}\bra{\psi_d}]$ is indistinguishable from the toric code reduced density matrix $\sigma_{\mathrm{TC}} = \tr_{R^c}[\ket{\psi_{\mathrm{TC}}}\bra{\psi_{\mathrm{TC}}}]$ defined on the same region, up to the bound $|| \rho_d - \sigma_{\mathrm{TC}} || \leq \max\left(\sqrt{\epsilon}, 2 \mathscr{L}^2 \epsilon \right)$
}

\begin{proof}
To bound the trace distance, we will use the fact that our state $\rho_d$ locally looks almost the same as the toric code state. Specifically, trace distance is related to distinguishability by~\cite{nielsen_quantum_2010}
\begin{align}
    ||\rho-\sigma|| =\frac{1}{2}\sup_{||O||\leq1}\tr\left[O(\rho-\sigma)\right].
\end{align}
To upper bound this, we can consider all possible unit norm operators $O$. Specifically, $O$ can always be written as a linear combination of Pauli strings. These Pauli strings can be analyzed by considering two cases, closed strings and open strings. First consider operators $C$ supported on $R$, which commute with all stabilizers $A_v, B_p$ supported on a slightly larger $\mathscr{L} \times \mathscr{L}$ region (1-ball or $R$) , constructed by expanding on all sides by one unit cell.
These operators must be products of contractible Wilson loops, and hence can be written as product of stabilizers. Therefore, $\tr[\sigma_{\mathrm{TC}} C] = 1$ for the toric code state.
For the LED output state, we instead bound the expectation value of $P_\textrm{TC} = \prod_v \frac{1+A_v}{2}  \prod_p  \frac{1+B_p}{2}$, where the product over $v$ (resp., $p$) runs over all vertices (plaquettes) within the $\mathscr{L} \times \mathscr{L}$ region. 
The expectation value of this projector, evaluated on the output state, is given by $\langle \psi_d| P_\textrm{TC} |\psi_d \rangle$.
To lower bound this quantity, we first notice that every term in $P_\textrm{TC}$ has spectrum $\{0,1\}$, and that the terms are mutually commuting.

This allows us to approximate $\langle P_\textrm{TC} \rangle$ using the individual expectation values $\left\langle \frac{1+A_v}{2} \right\rangle > 1-\epsilon$ and $\left\langle \frac{1+B_p}{2} \right\rangle > 1-\epsilon$. To do so, we note that if two commuting operators $A$ and $B$, each with spectrum $\{0,1\}$, satisfy $\langle A\rangle > 1-\epsilon$ and  $\langle B\rangle > 1-\epsilon$, then $\langle AB \rangle > 1-2\epsilon$. To show this, we add the two individual bounds to obtain $\langle A\rangle  + \langle B\rangle > 2-2\epsilon$; moreover, since $AB$ has spectrum $\{0,1\}$, we have $\langle A\rangle + \langle B\rangle - \langle AB \rangle < 1$. It thus follows that $\langle AB \rangle > 1-2\epsilon$, and upon applying this recursively to include all vertex and plaquette terms within the $\mathscr{L}\times \mathscr{L}$ region, we obtain $\left\langle \psi_d | P_\textrm{TC} | \psi_d \right\rangle > 1-2\mathscr{L}^2 \epsilon$.

Using this fact, along with $C P_{\mathrm{TC}} = P_{\mathrm{TC}}$ we can similarly lower bound the expectation value $\tr[\rho_d C] \geq 1 - 4\mathcal{L}^2 \epsilon$.
Thus, for any $C$, we have $\frac{1}{2} |\tr\left[ C (\rho_d - \sigma_{\mathrm{TC}}) \right]| \leq 2 \mathcal{L}^2 \epsilon$.

Next, we consider operators $O$ which anti-commute with some stabilizers. In particular, these operators are Pauli strings with at least one endpoint (open strings). Naturally, their expectation value vanishes in the toric code state. To see this, let $S$ be one of the stabilizers which anti-commutes with $O$.
Now, $S$ and $O$ form an anti-commuting pair of Pauli operators, so they satisfy an uncertainty relation
$\left\langle S\right\rangle ^{2}+\left\langle O\right\rangle ^{2} \leq 1$.
As such, in the toric code where $\left\langle S\right\rangle=1$, this implies $\left\langle O\right\rangle=0$.
Similarly, the condition $\tr[\rho_d S] \geq 1 - 2\epsilon$ leads to the upper bound $\tr[\rho_d O] \leq 2\sqrt{\epsilon }$.
Combining these two results, we see the trace distance is at most $T(\rho_d, \sigma_{\mathrm{TC}}) \leq \max\left(\sqrt{\epsilon}, 2 \mathscr{L}^2 \epsilon \right)$.
\end{proof}

\textbf{Lemma 4.}
\textit{Consider an input state $\rho$ and an LED procedure satisfying the conditions of Theorem~\ref{thm:epsilon-witness}.
Then the final state $|\psi_d \rangle$ after LED cannot be prepared using a local quantum circuit with depth less than $O(\mathscr{L}) \sim O(1/\sqrt{\epsilon})$.}

\begin{figure}
    \centering
    \includegraphics[width=0.45\textwidth]{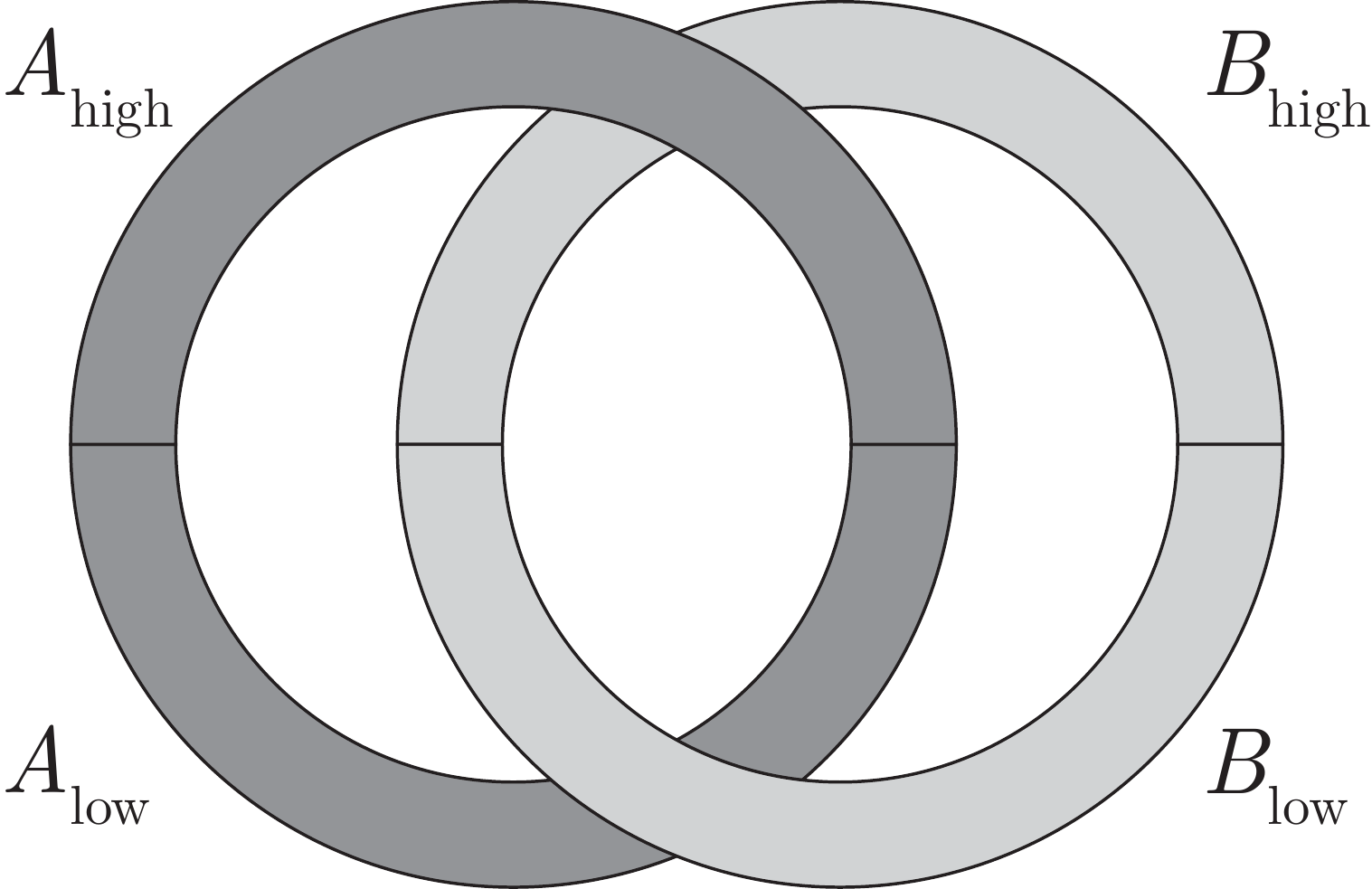}
    \caption{Consider two operators $A$ and $B$, each supported on an annulus which intersects at two regions, one higher and one lower. The twist product $A\infty B$ consists of applying $B$ first in the higher region, and $A$ first in the lower region (also see Ref.~\cite{haah_invariant_2016} Fig. 1 and Eq. 4). In the diagram the order of operations goes back to front. Furthermore, for tensor product operators which can be written as $A = A_{\mathrm{low}} A_{\mathrm{high}}$, $B=B_{\mathrm{low}} B_{\mathrm{high}}$, we can write the twist product as $A \infty B = B_{\mathrm{low}}A_{\mathrm{high}} A_{\mathrm{low}}B_{\mathrm{high}}$.}
    \label{fig:fig_supp_twistproduct}
\end{figure}

To prove these results, we extend and generalize the proof techniques developed by Ref.~\cite{haah_invariant_2016}. 
Note that for notional simplicity, here we work with an input state $\ket{\psi}$ that is a purification of $\rho$. This is done without loss of generality, since all of the operations act on the original degrees of freedom.

Consider a pair of $Z$-basis and $X$-basis Wilson loops $A$ and $B$, supported on two overlapping annuli. The twist product of the two operators $A \infty B$ is defined such that at one intersection region, operator $B$ is applied first, while at the other operator $A$ is applied first (Fig.~\ref{fig:fig_supp_twistproduct}). By the arguments of Ref.~\cite{haah_invariant_2016}, such a pair of locally non-commuting observables, whose twist product does not factorize into a product of the individual observables, can serve as a witness for long-range entanglement.
In our case, since $Z$- and $X$-strings locally anti-commute, we can remove the twist to get $A \infty B = -AB$.

Then, Ref.~\cite{haah_invariant_2016} showed that, assuming the observables $A$ and $B$ satisfy an additional important property called local invisibility (see below), then the following correlation serves as a witness for long-range order.
\begin{align}
    \left|\bra{\psi} A \infty B \ket{\psi} - \bra{\psi} A \ket{\psi}\bra{\psi} B \ket{\psi}\right| > 0.
\end{align}
Specifically, the state $\ket{\psi}$ cannot be prepared from a trivial state by a circuit of depth $O(L)$, where $L$ is the separation between the two intersection regions.
Indeed, in the exact case, where the expectation value of large LED Wilson loops are one, the results of Ref.~\cite{haah_invariant_2016} can be directly applied.

However, the proofs in Ref.~\cite{haah_invariant_2016} do not immediately apply to the realistic case considered here, where stabilizers have expectation value $1 - \epsilon$, and residual entanglement between the system and ancilla or environment qubits prevents exact knowledge of the state. Nevertheless, with sufficient care and a few additional assumptions, approximate versions of key results in Ref.~\cite{haah_invariant_2016} can be recovered. 

First, we develop a notion of \textit{approximate local invisibility}. Throughout, we follow the spirit of the proofs in Section III of Ref.~\cite{haah_invariant_2016}; the reader is encouraged to consult the original reference for additional details and insights.

\begin{definition}[Approximate ($\Delta, r,t$)-local invisibility.]
Let $A$ be a region of radius $r$ and $B$ be a $t$-ball around $A$.
An operator $O$ with unit norm is ($\Delta, r, t$)-locally invisible with respect to a state $\ket{\psi}$ if, for any state $\ket{\phi}$ whose reduced density matrix on $B$ is equivalent to $\ket{\psi}$, it satisfies
\begin{align}
    \left\vert \left\vert \frac{\tr_{A^c}[O \ket{\phi}\bra{\phi} O^\dagger]}{\tr[O \ket{\phi}\bra{\phi}O^\dagger]} - \tr_{A^c}[\ket{\psi}\bra{\psi}] \right\vert \right\vert \leq \Delta,
\end{align}
where the norm is the standard trace norm. 
In other words, locally invisible operators leave local reduced density matrices approximately unchanged. 
Note that we restrict to states $\ket{\phi}$ for which the expectation value does not vanish, such that this remains well-defined. 
This subtlety is also present in the original definition of Ref.~\cite{haah_invariant_2016}.
\end{definition}

Next, we will show that Wilson loops that nearly stabilize $\ket{\psi_d}$ are approximately locally invisible.
Let $A$ be a region of radius $r$, which can only cover a patch of the loop.
Furthermore, let $t=0$, i.e. region $B$ is identical to region $A$.
Since the Wilson loop is a tensor product of local unitaries, we can write $O = O_{B} \otimes O_{B^c}$ and $O_{B^c}^\dagger O_{B^c} = 1$.
This allows us to work directly with the reduced density matrices of region $B$, and we can use Lemma 3 to reduce to the toric code case
\begin{align}
    \tr_{A^c}[O\rho_d O] = \tr_{A^c}[O\sigma_{\mathrm{TC}} O] + \tr_{A^c}[O (\rho_d - \sigma_{\mathrm{TC}}) O]
\end{align}
Indeed, since $O$ is locally invisible with respect to to the toric code, this gives us our result, where the error term depends on the size of $B$.
\begin{align}
    \left\vert \left\vert \tr_{A^c}\left[O\rho_d O\right] - \tr_{A^c}\left[\sigma_{\mathrm{TC}} \right] \right\vert \right \vert \leq \max\left(\sqrt{\epsilon}, 2 (r+1)^2 \epsilon \right)
\end{align}
More microscopically, $O$ spans the region, so locally looks like a logical operator. The reduced density matrix $\sigma_{\mathrm{TC}}$ on region $B$ is an equal weight mixture of all logical states, so $O$ leaves it invariant.

This shows that Wilson loops are ($\Delta, r, t$)-locally invisible with respect to $\ket{\psi_d}$ for $t \geq 0$ and $\Delta = \max\left(\sqrt{\epsilon}, 2 (r+1)^2 \epsilon \right)$. 
When combined with the fact that Wilson loops have large expectation value on $\ket{\psi_d}$, this will serve as a witness for long-range topological order.

To prove this, we need to confirm that, even for the weaker notion of approximate local invisibility, the twist product approximately factorizes for trivial states.
\begin{lemma}
The twist product of two $(\Delta, r,t)$ locally invisible operators $A$ and $B$, acting on a trivial product state $\ket{\psi} = \ket{00...0}$, must satisfy
\begin{align}
    \left\vert \bra{\psi} A \infty B \ket{\psi} - \bra{\psi} A \ket{\psi}\bra{\psi} B \ket{\psi} \right\vert \leq O(\sqrt{\Delta R/r})
\end{align}
where $A$ and $B$ are supported on two annuli which intersect at two regions (Fig.~\ref{fig:fig_supp_twistproduct}) whose separation is $\geq 2(r+t)$.
\end{lemma}

\begin{proof}
In the first step of the proof, we bound the expectation value of $\Pi_\mathcal{\mathcal{R}} = \prod_{i \in \mathcal{R}} \ket{0}\bra{0}_i$, the projector onto $\ket{\psi}$ supported on region $\mathcal{R}$ evaluated with respect to $O\ket{\psi}$ for unitary $O$.

Invoking the definition of local invisibility, we show for $A$ of radius $r$, that $\bra{\psi} O^\dagger \Pi_A O \ket{\psi} \geq 1-\Delta$. 
\begin{align}
    \left\vert \left\vert \tr_{A^c} \left[ O \ket{\psi} \bra{\psi} O^\dagger \right] - \Pi_A \right\vert \right\vert &\leq \Delta
\end{align}
Thus, the expectation value of the observable $\Pi_A$ satisfies:
\begin{align}
    | \tr\left[ \Pi_A O \ket{\psi} \bra{\psi} O^\dagger \right] - 1| &\leq \Delta \\
    \bra{\psi}O^\dagger \Pi_A O \ket{\psi} &\geq 1 - \Delta
\end{align}
Next, we can use the fact that $\Pi_{\mathcal{R}}$ can be written as a product of $\Pi_A$ approximately $R/r$ times. 
It follows from the same union bound argument in Lemma 3, that $\bra{\psi} O^\dagger \Pi_{\mathcal{R}} O \ket{\psi} \geq 1 - \Delta R / r$. This result can be used to bound the distance between the two states,
\begin{align}
    ||\Pi_{\mathcal{R}} O \ket{\psi} - O \ket{\psi}||^2 \leq \Delta R / r,
\end{align}
by noticing the left side is equal to $1 - \bra{\psi}O \Pi_{\mathcal{R}} O \ket{\psi}$.
We will use this below, to show that we can replace $O \ket{\psi}$ with $\Pi_{\mathcal{R}} O \ket{\psi}$ without incurring significant error. 

To prove the main result, we use the same construction as Ref.~\cite{haah_invariant_2016}.
Specifically, we want to use the above result to show that 
$A \infty B \ket{\psi} = (A \Pi_\mathcal{R}) \infty B \ket{\psi} + O(\Delta \mathcal{L}^2)$.
We do this by carefully inserting projectors. 
For Wilson loops, which are tensor product operators, we the twist product can be split up as follows (see Fig.~\ref{fig:fig_supp_twistproduct}),
\begin{align}
A \infty B = B_\textrm{low} A_\textrm{high} A_\textrm{low} B_\textrm{high}.
\end{align}
Our above result implies $\Pi_{\mathcal{R}} B_\textrm{high} \ket{\psi} = B_\textrm{high} \ket{\psi} + O(\Delta R/r)$.
We can subsequently pull projectors from $\ket{\psi}$ to cover the region of support of $A$, e.g. low and the parts of high.
Thus, we get 
\begin{align}
\begin{split}
    A \infty B \ket{\psi} = &B_\textrm{low} A_\textrm{high} \Pi_\textrm{high} A_\textrm{low} \Pi_\textrm{low} B_\textrm{high} \ket{\psi}\\ 
    & + O(\sqrt{\Delta R/r}) \\
    = &(A \Pi) \infty B \ket{\psi} + O(\sqrt{\Delta R/r})
\end{split}
\end{align}
as we wanted. 
Finally, we this implies the expectation value of the twist product approximately factorizes
$\bra{\psi} A \infty B \ket{\psi} = \bra{\psi} A \ket{\psi}\bra{\psi} B \ket{\psi} + O(\sqrt{\Delta R/r})$.
\end{proof}

Although we proved that the twist product must factorize for the trivial product state, this holds for a much wider class of short-range entangled states, generated from a trivial state by a finite-depth unitary circuit. In particular, using the same argument as Haah Lemma III.3 shows that, given a ($\Delta,r,t$)-locally invisible operator $O$ and state $\ket{\psi}$, if we evolve under a local unitary circuit $W$ of depth $d$, then the operator $WOW^\dagger$ is $(\Delta,r-d,t+2d)$-locally indistinguishable with respect to $W\ket{\psi}$.
This will be used to show that non-factorizability of the twist product lower bounds the depth of a quantum circuit required to produce the state from the trivial product state.

Finally, we can use the trace distance Lemma 3, to show the twist product does not factorize for $\ket{\psi_d}$. Specifically, let $A$ and $B$ be Wilson loops supported on an $\mathscr{L} \times \mathscr{L}$ region. Then,
\begin{align}
    \vert \bra{\psi_d} A \infty B \ket{\psi_d} - \bra{\psi_d} A \ket{\psi_d} \bra{\psi_d} B \ket{\psi_d} \vert \geq 2 - c \mathscr{L}^2 \epsilon
\end{align}
for a constant $c$.
Combining this Lemma 6 for trivial states, we see the bound is violated when
\begin{align}
    2 - c \mathscr{L}^2 \epsilon &> c' \sqrt{\Delta \mathscr{L} / r} \\
    2 - c \mathscr{L}^2 \epsilon  - c' 2 (r+1) \sqrt{\mathscr{L} \epsilon / r} &\geq 0
\end{align}

We choose $r$ to be a constant fraction of $\mathscr{L}$, and see that we can roughly make $\mathscr{L} \sim O(1/\sqrt{\epsilon})$ and still certify long-range order.
In particular, the state $\ket{\psi_d}$ cannot be generated by a finite depth circuit of depth smaller than $r \sim O(\mathscr{L})$.
Since $\ket{\psi_d}$ is connected to the input state by a depth-$d$ quantum circuit, this implies LED Wilson loops close to one certify topological order up to length-scales $O(\mathscr{L}-d)$.

\subsection{Connection Between LED and Entanglement Negativity}

The LED framework suggests a characterization of topological order based on the ability to distill the fixed-point wavefunction using generalized local unitary operations.
In this section, we connect this definition to the topological entanglement negativity of the input state, a typical observable used to detect topological order in mixed states~\cite{peres_separability_1996, horodecki_separability_1996}.

As described above, the circuit construction for LED involves applying local unitary transformations to the input mixed state $\rho_S$ and product state ancillas $\ket{0}_A\bra{0}$. States that are classified as topological by LED are connected via a finite-depth local unitary to the toric code fixed-point state and an ancillary register (see Theorem~\ref{thm:localGLU}).
$$\rho_{S} \otimes \ket{0}_{A}\bra{0} \rightarrow \ket{\psi_\textrm{TC}}_{S'}\bra{\psi_\textrm{TC}} \otimes \sigma_{A'}$$
In particular, all of the entropy in the input state, associated with both incoherent and coherent fluctuations away from the fixed point state, is transferred to the ancilla $\sigma_{A'}$. 

We conjecture that the ability to ``distill'' the toric-code fixed point state also implies the presence of a topological correction to the entanglement negativity of the input state.
Although we cannot prove this statement rigorously, we make the connection more precise in this section.
First, we make the plausible assumption that the ancilla $\sigma_{A'}$ is in a trivial mixed state.
This ensures all of the topological contribution to the entanglement negativity comes from the system part $S'$.
Then, we argue the topological contribution remains after applying inverse LED circuits (coarsening of the Wilson loops). In particular, if we consider local Clifford circuits, we can show this is indeed the case.
Together, these arguments suggest that the topological order witness provided by LED should also serve as a witness for other quantities like entanglement negativity in many scenarios of practical interest.

\begin{figure}
    \centering
    \includegraphics{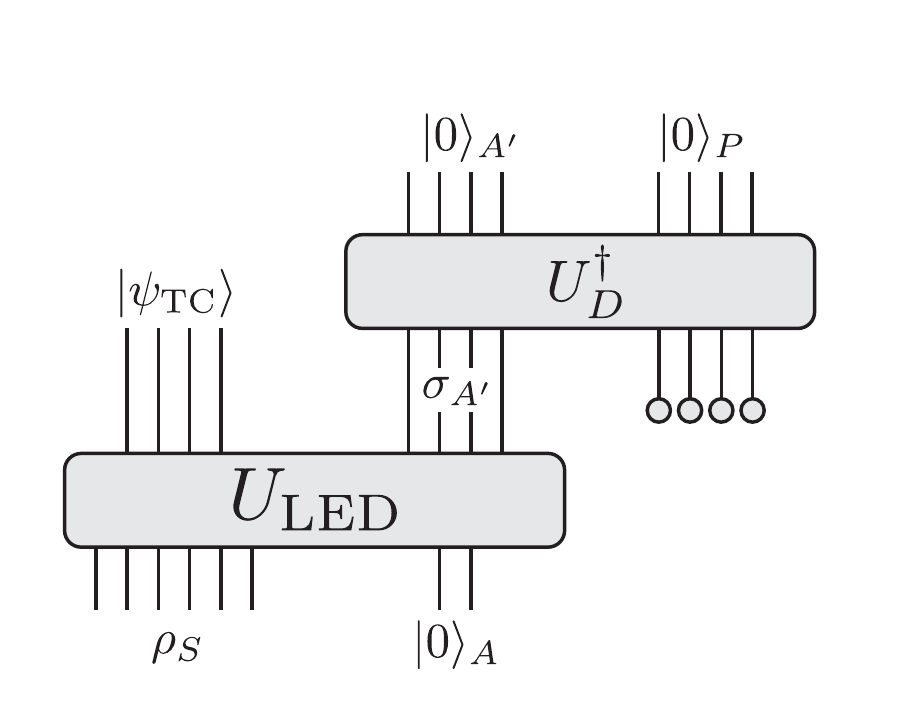}
    \caption{Theorem 1 shows that a positive LED classification implies the input mixed state $\rho_S$ and product state ancilla's ($A$) can be connected to an output state of the form $\ket{\psi_\textrm{TC}}_{S'}\bra{\psi_\textrm{TC}} \otimes \sigma_{A'}$, where $S'$ is typically smaller than $S$ due to coarse-graining. To make a connection between LED and entanglement negativity of $\rho_S$, we further assume $\sigma_{A'}$ is trivial. In particular, we assume there is a local purification of $\sigma_{A'}$ that is connected to the product state via local unitaries. Equivalently, we assume $\sigma_{A'}$ can be generated from a product state by applying local unitaries and tracing out some degrees of freedom.}
    \label{fig:local_purification}
\end{figure}

The key assumption we make is that the ancilla $\sigma_{A'}$ is in a trivial mixed state.
Concretely, we assume that there exists a \textit{local purification} of $\sigma_{A'}$ which is connected to a product state under local unitary circuits (Figure~\ref{fig:local_purification}).
This local purification is constructed by doubling the number of qubits, such that $\sigma_{A'} = \tr_{P} [\vert \Sigma \rangle \langle \Sigma \vert]_{A' P}$. Here, $P$ is the subsystem which implements the purification. Qubits in $P$ are placed next to qubits in $A'$, so that the local structure is preserved.
Then, we assume that there is a local unitary $U^\dagger_D$ which maps $\vert \Sigma \rangle$ to a product state. Thus, $\sigma_{A'} =  \tr_P [ U_D \vert 0 \rangle \langle  0 \vert U_D^\dagger ]$ can be ``locally purified'' to a product state.
We note that in the circuit model, the assumption about $\sigma_{A'}$ may be checkable via state tomography.

Taking into account the purified register of qubits, the output state can now be written as a pure state. Further, it is a stabilizer state, with toric code stabilizers $A_v, B_p$ for the $S'$ subsystem, and trivial stabilizers $Z_i$ for the combined $A' P$ subsystem. Note we assume for simplicity a topologically trivial manifold, so there are no non-trivial logical operators.
Our construction then implies that this stabilizer state is connected to a purification of the input state via local unitary circuits. If we then trace out the degrees of freedom associated with $P$, this produces the input mixed state.
We can leverage the stabilizer structure to make concrete statements about the entanglement negativity of this input mixed state.

First, we review the calculation of the entanglement negativity from Ref.~\cite{Lu_negativity_2022}, for the toric code output state. 
Consider a partition of the system into region $R$ and its complement $R'$. Note that $R$ is associated with a region of space, and hence contains sites from all three subsystems $S'$, $A'$, and $P$.
The entanglement negativity is given by $S_R = \log || (\rho')^{T_A} ||_1 = \log(\sum_i p'_i)$, where $T_A$ is the partial transpose of the output density matrix, and $p'_i$ are the associated eigenvalues. 
The partial transpose of the density matrix $\rho$, which is defined as
\begin{align}
    \rho = \sum_{i,i',j,j'} \rho_{i,i',j,j'} \vert i \rangle_{A} \langle i' \vert \otimes \vert j \rangle_{A^c} \langle j' \vert \\
    \rho^{T_A} = \sum_{i,i',j,j'} \rho_{i,i',j,j'} \vert i \rangle_{A} \langle i' \vert \otimes \vert j' \rangle_{A^c} \langle j \vert
\end{align}
where $\vert i\rangle_{A}$ and $\vert j \rangle_{A^c}$ form an orthonormal basis for the two parts of the system.
To start, notice that the only non-zero eigenvalue of the bare output state $\rho'$ corresponds to the simultaneous $+1$ eigenstate of all the stabilizers, $A_v, B_p$, and $Z_i$.

Under the inverse LED operations, each individual stabilizer is coarsened, making exact computation of the entanglement negativity difficult.
However, the LED circuit provides additional structure. Importantly, the toric code stabilizers map onto the system part. These become the LED Wilson loops which we measure in the input state.
In contrast, the trivial stabilizers can also map onto subsystem that will be traced out.
As a result, we expect the coarsened toric code stabilizers will still contribute a topological term in the entanglement negativity.
We conjecture this is true for any local unitary circuit.
However, we are able to calculate it if we consider the specific case of local Clifford circuits.
Note that the LED circuits which implement local pairing decoders are not Clifford circuits, and hence the argument does not rigorously apply in that case.

Let $A_v, B_p$ label the toric code stabilizers in the original state, and $Z_i$ the stabilizers of the ancilla.
Under the Clifford unitary, these stabilizers are mapped to coarsened versions, $\tilde{A}_v, \tilde{B}_p, \tilde{Z}_i$.

The reduced density matrix of the input state, before tracing out the auxilliary qubits, can be written as
\begin{align}
    \rho = \prod_v \frac{1+\tilde{A}_v}{2} \prod_p \frac{1+\tilde{B}_p}{2} \prod_i \frac{1+\tilde{Z}_i}{2}.
\end{align}
Suppose there are $N=V+P+A$ stabilizers and an equivalent number of qubits. This is the case if we are working on a topologically trivially manifold. 
Then, $\rho$ has a unique non-vanishing eigenstate with eigenvalue one.
Thus, $\log ||\rho||_1 = 0$.
Following Ref.~\cite{Lu_negativity_2022}, we can equivalently write
\begin{align}
    \rho = 2^{-N} \sum_{s,r,t} \prod_{v,p,i} \tilde{A}_v^{s_v} \tilde{B}_p^{r_p} \tilde{Z}_v^{t_i}.
\end{align}

Next, consider splitting the system into two regions, $R$ and $R'$, and take the partial transpose of $\rho$ on $R$. 
\begin{align}
    \rho^{T_R} = 2^{-N} \sum_{s,r,t} \prod_{v,p,i} \left( \tilde{A}_v^{s_v} \tilde{B}_p^{r_p} \tilde{Z}_v^{t_i} \right)^{T_R}
\end{align}
The key to computing the entanglement negativity is to simplify $\left( \tilde{A}_v^{s_v} \tilde{B}_p^{r_p} \tilde{Z}_i^{t_i} \right)^{T_R}$.

In the output state, which is simply the toric code state with product state ancillas, the entanglement negativity can be exactly computed, and has been shown to have a topological correction coming from the long-range topological order~\cite{lee_entanglement_2013,Lu_negativity_2022}.
We extend this computation to show the topological correction survives upon performing local Clifford circuits and tracing out the ancillary degrees of freedom.

We start by reviewing the computation in the output, fixed-point state.
The first thing to notice is that $Z_i$ trivially commutes with $A_v$ and $B_p$, because they are supported on different qubits. Second, $Z_i$ is invariant under partial transpose, since the operators are supported on a single site and hence cannot span the boundary between $R$ and $R'$.
Thus, we can pull the $Z_i$ out from the product.
$$\left( A_v^{s_v} B_p^{r_p} Z_i^{t_i} \right)^{T_R} \rightarrow \left( A_v^{s_v} B_p^{r_p} \right)^{T_R} Z_i^{t_i}$$

Next, we discuss the effect of partial transpose on the toric code stabilizers, following Ref.~\cite{Lu_negativity_2022}. 
The only stabilizers which are affected non-trivially are those that span the boundary.
In particular, if we consider a pair of stabilizers $A_v$ and $B_p$ that intersect at two sites, but only one of the sites is supported in $R$, then they pick up a minus sign under partial transpose (see Figure~\ref{fig:negativity_SPT}).
More generally, $(A_v^{s_v} B_p^{r_p})^{T_A} = (-1)^{C(s_v, r_p)} A_v^{s_v} B_p^{r_p}$ picks up an overall sign if the number of such intersections $C(s_v, r_p)$ is odd.

We can now discuss the spectrum of the partially transposed density matrix. Note that an orthonormal basis for the Hilbert space can be formed from simultaneous eigenstates of all the stabilizers. These states are labelled by the stabilizer eigenvalues, which can be $\pm 1$.
For non-trivial eigenstates of $\rho^{T_A}$, stabilizers which are unaffected by partial transpose must have eigenvalue one.
Thus, we can reduce the spectrum calculation to only stabilizers on the boundary.

Therefore, we can associate eigenstates with expectation values of ``strange correlators'' of the form $\langle + \vert \cdot \vert \psi \rangle$~\cite{Lu_negativity_2022}.
Order stabilizers around the boundary, so $A_v$ are on even sites and $B_p$ on odd sites, and consider specifically the correlators,
$$\bra{+} Z_{2i}^{(1+A_v)/2} Z_{2i+1}^{(1+B_p)/2} \ket{\psi}.$$
Here, $\psi(s) = (-1)^{\sum_i s_i s_{(i+1) \textrm{mod} L}}$ is a wavefunction that encodes the signs coming from partial transpose, and $\bra{+}=2^{-L} \sum_s \langle s \vert$ is a uniform superposition over all configurations.
The wavefunction for $\psi(s)$ is the same as the 1D cluster state.
Non-vanishing correlators are generated by products of $Z_{2i-1} X_{2i} Z_{2i+1}$ and $Z_{2i} X_{2i+1} Z_{2i+2}$, the stabilizers of the cluster state.
There are $2^{L-2}$ of such non-vanishing correlators, each with expectation value $\pm \langle + \vert \psi \rangle$.
The overlap $\langle + \vert \psi \rangle$ is $2^{-n+1}$, where $L=2n$.
Thus, the entanglement negativity is $S_N = (n-1)\log 2$ where $n$ is the size of the perimeter.
The constant correction is a signature of topological entanglement entropy.

Next, we discuss how the effect of finite-depth Clifford circuits on the entanglement negativity.
In particular, since Clifford circuits map stabilizers to stabilizers, eigenstates of $\rho^{T_A}$ can still be labelled by stabilizer expectation values, and the entanglement negativity can be calculated by enumerating non-zero strange correlators.
However, a few important things change.
The set of stabilizers which span the boundary, and hence may have non-trivial twist products with other stabilizers, is considerably larger. Any $A_v, B_p,$ or $Z_i$ stabilizer that is within distance $\ell$ of the boundary could contribute.
Applying the same procedure as in the case of the toric code, we can enumerate non-trivial strange correlators of the form $\langle + \vert O \vert \psi' \rangle$ for a different state $\psi'(s)$.
However, now the Hilbert space of $\psi'(s)$ has one qubit associated with each relevant $A_v, B_p$, or $Z_i$.

Now, we argue that new state $\psi'(s)$ preserves the same cluster-state structure as in the fixed-point case and still has long-range SPT order.
First, we show that $\psi'$ still has a $\mathbb{Z}_2 \times \mathbb{Z}_2$ symmetry, generated by the product of Pauli-$X$ operators on all qubits associated with $A_v$, or all qubits associated with $B_p$.
To see this, notice that the product of all relevant $A_v$ decomposes into two independent Wilson loops, each supported distance $\ell$ away from the boundary on the interior and exterior sides. Each of these Wilson loops individually has trivial twist-product with the other transformed stabilizers, implying that it can be applied (on either side) to the partially-transposed state $\rho^{T_A}$ without changing the state. This further implies that $\psi(s)$ is invariant under applying the Pauli-$X$ associated with each $A_v$ qubit.

Second, we show that there exists long-range SPT order in $\psi'(s)$. In particular, to do this, we show that there are coarse-grained stabilizers which have the same behavior under partial transpose as the stabilizers at the toric-code fixed point. 
Consider large Wilson loops generated by multiplying either $A_v$ or $B_p$ in a large contiguous region spanning the boundary (see Figure~\ref{fig:negativity_SPT}). In particular, if we look at regions that are offset, the large Wilson loops associated with $A_v$ or $B_p$ necessarily anti-commute on either side of the boundary, leading to a twist product of $-1$. 
Therefore, we can construct coarse-grained Pauli-`$Z$' and `$X$' operators in the Hilbert space of $\psi'$ simply by taking products of single-qubit Pauli operators. This argument shows that the coarse-grained $ZXZ$ operator is a stabilizer of the state.

Finally, we discuss the effect of the remaining degrees of freedom on the entanglement negativity.
The $\tilde{Z}_i$ stabilizers split into two groups: those that contain support on $P$, and those that are fully supported on $SA$.
Let us start with the first group. Since we trace out all of subsystem $P$, the effect on the $\tilde{Z}_i$ stabilizers is to replace the stabilizer state with the maximally mixed state
\begin{align}
    \mathrm{tr}_P\left[ \frac{1+\tilde{Z}_i}{2} \right] \rightarrow \frac{1}{2}I_i.
\end{align}
The operator $I/2$ has two eigenvectors, each with eigenvalue 1/2. Thus, any stabilizer with support on $P$ becomes fully disentangled from the rest of the system, and has no effect on the entanglement negativity. 
Next, consider stabilizers that are fully supported on $SA$. The $\tilde{Z}_i$ which are near the boundary between $R$ and $R'$ will potentially contribute to the entanglement negativity. 
Furthermore, they may have non-trivial commutation with $\tilde{A}_v$, $\tilde{B}_p$, and other $\tilde{Z}_i$ after partial transposition. Thus, they can contribute to the topological entanglement entropy.
However, since only operators near the boundary are affected by partial transpose, these operators only contribute an area-law contribution. 
Further, we note that they cannot change the long-range SPT order of $\psi(s)$, since they commute with the $\mathbb{Z}_2 \times \mathbb{Z}_2$ symmetry. As such, they also should not affect the value of the topological term.

The argument presented above relied on the fact that the stabilizers of the state were Pauli operators. This meant the spectrum of $\rho^{T_A}$ could be exactly computed, and the techniques of Ref.~\cite{Lu_negativity_2022} could be applied.
However, the key part of the proof involved generating a sub-algebra of coarse-grained Wilson loops with twist-product $-1$. 
LED provides a constructive method for measuring such coarse-grained Wilson loops, with non-trivial twist products, in generic mixed states, as explained in our proof of Theorem~\ref{thm:localGLU}. Hence, it would be interesting to directly link the existence of operators with non-trivial twist product to the topological entanglement negativity in general mixed states.

\begin{figure}
    \centering
    \includegraphics[width=1.0\columnwidth]{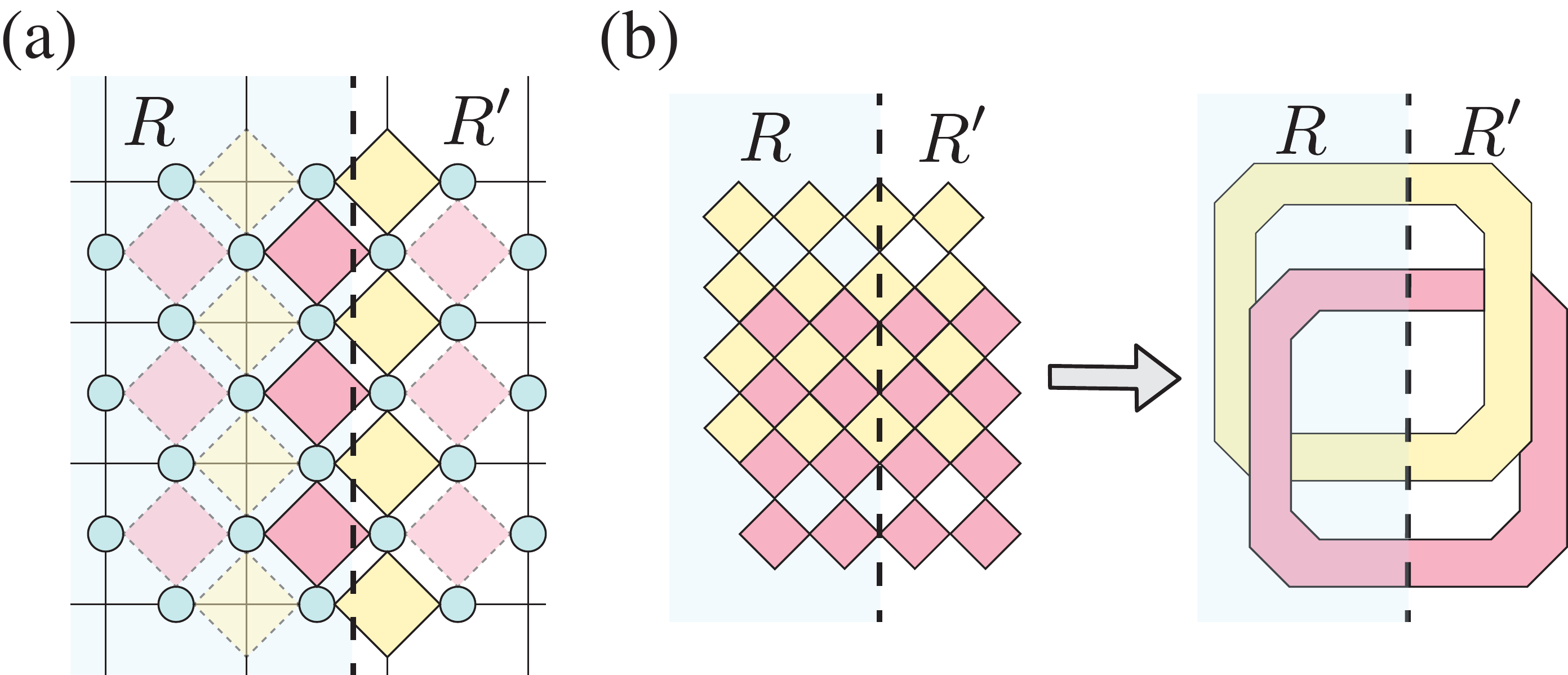}
    \caption{(a) Calculating entanglement negativity in the toric code state reduces to computing a ``strange correlator'' in a state that lives on the boundary, where qubits in the state encode stabilizer expectation values. Stabilizers which do not span the boundary do not commute. Hence, they are forced to be $+1$ and do not contribute to the strange correlators. In contrast, vertex and plaquette stabilizers, whose intersections span the boundary pick up $-1$ signs under partial transpose. This acts like a CZ gate on the boundary state, essentially preparing a cluster state (see text). The long-range SPT order of the cluster state leads to a topological correction to the entanglement negativity. (b) A local (Clifford) unitary circuit coarse-grains each of the stabilizers in the Heisenberg picture. As a result, a much larger number of stabilizers participate in the boundary state. However, we show the long-range SPT order survives, by looking at products of stabilizers. In particular, a product of a large block of stabilizers is supported only on the boundary. After the unitary circuit, this boundary gets coarsened, but for large enough blocks the intersection structure survives. This ensures the operators pick up $-1$ signs under partial transpose as well.}
    \label{fig:negativity_SPT}
\end{figure}

\subsection{Valence-Bond Solid Phase}

In the main text, we argued that the flow of closed-loop and open-string Wilson operators suggest that the large $\Delta/\Omega$ regime is most consistent with a decoherence-dominated disorder phase.
This in contrast to numerical simulations, which suggest the ground state in this regime is a valance bond solid (VBS) state~\cite{verresen_prediction_2020}.
The VBS fixed-point state can be understood as a single dimer covering, in contrast to the spin-liquid state which is a uniform superposition over dimer coverings. More generally, in VBS phase the $m$-anyons become condensed and the $e$-anyons become confined. 
As such, in the VBS phase, we would expect that closed $Z$-loops flow to $+1$, and open $Z$-strings flow to either $+1$ or $-1$ depending on how many dimers they intersect. 
However, our analysis shows that all $Z$-loops, both open and closed, stay near zero under LED flow.
In the main text, we average over all open $Z$-strings, so the vanishing signal could be due to cancellation of individual loops that flow to opposite values. In Figure~\ref{fig:novbs_flow}, we also plot a representative, single $Z$-string, and see the expectation value still remains close to zero at large $\Delta/\Omega$. 
Since the $X$-loops also remain at zero, this is the same signature of the high incoherent error regime of the toric code mixed-state phase diagram.
\begin{figure}
    \centering
    \includegraphics[width=0.95\columnwidth]{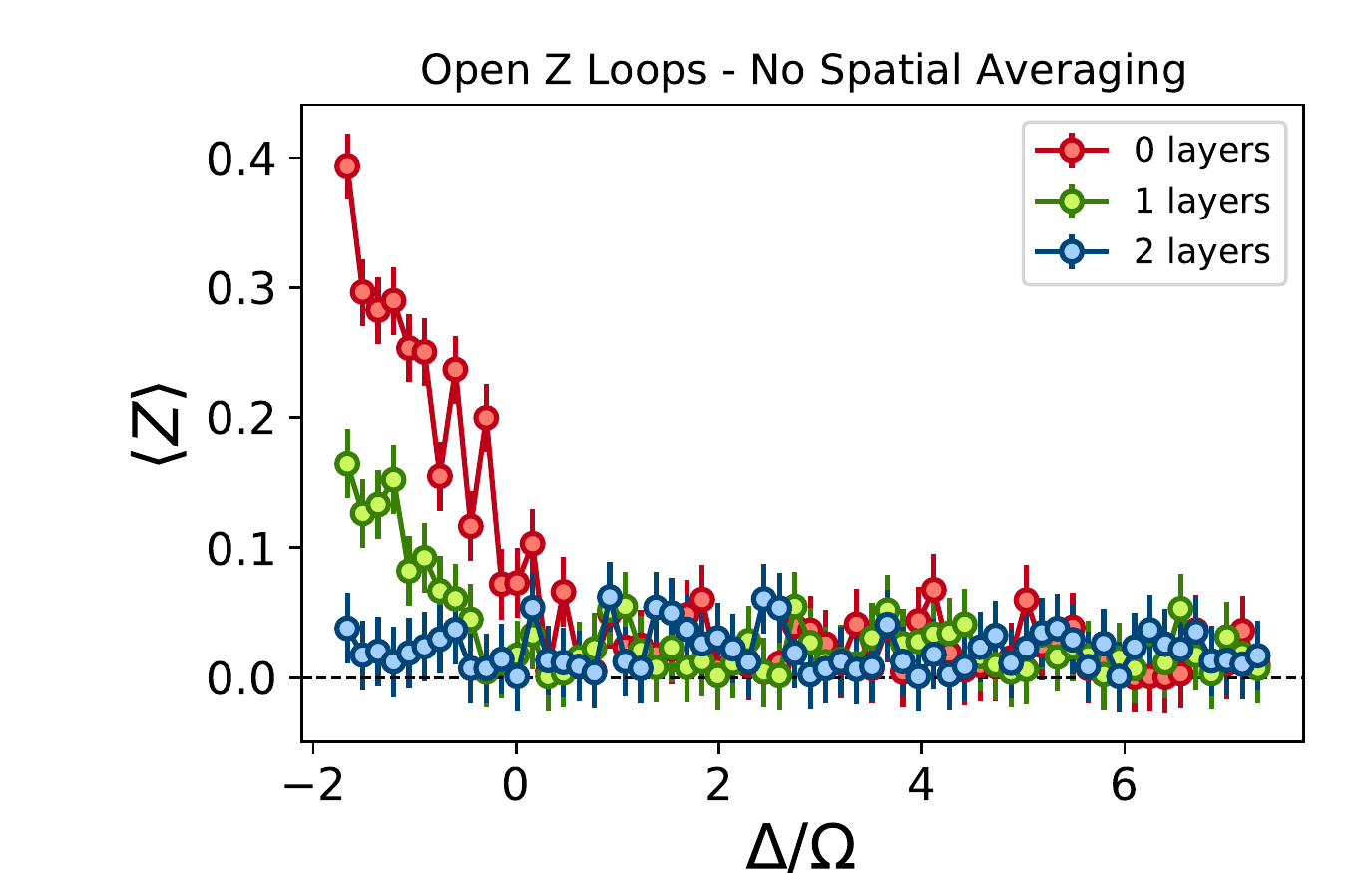}
    \caption{Flow under LED of a representative, single open $Z$-string. This string operator remains close to zero in the large $\Delta/\Omega$ regime, which is inconsistent with preparation of a VBS state.}
    \label{fig:novbs_flow}
\end{figure}

\subsection{Detecting Gauge-Glass Phases and More }

The LED techniques developed here can also be applied to characterize generalized toric code orders in higher dimensions.
The 2D toric code we study has point-like excitations. However, different models can have higher-dimensional excitations, such as line-like excitations appearing in the 2D Ising model and 3D gauge theories.
LED can also be applied in these cases.
In particular, simple local decoders for models with line-like excitations also exist. For example Toom's rule, as described in Ref.~\cite{breuckmann_local_2016}, can be made into a local decoder by restricting the number of update steps to some constant. The patch-based decoder can also be generalized to this case by locally identifying small contractible loops.
Another strategy would be to consider spatial slices of the model with line-like excitations. This would result in a lower-dimensional model with point-like excitations, for which the techniques presented in this paper can be directly applied.

A particularly interesting case study is the 3D random-plaquette gauge model. The degrees of freedom live on the links of the 3D square lattice. Products of spin operators around each plaquette measure the magnetic flux. Excitations are line-like magnetic flux tubes, coming from chains of excited plaquettes.
The model also has quenched disorder, so that the zero-temperature state has a finite density of magnetic flux tubes.
Hence, there are two sources of flux tubes, the quenched disorder as well as thermal fluctuations. Due to the gauge constraints, the combination of flux tubes from quenched disorder and thermal fluctuations always form closed loops.

Let $K_p$ be the parameter which determines the density of quenched disorder, and $K$ the effective temperature.
The model has three phases of interest to us. 
The first is an ordered phase for small $K$ and $K_p$, where the density of flux tubes is small.
The second is a disordered phase at large $K$ and $K_p$, where flux tubes condense.
These two phases can be distinguished by measuring the disorder averaged Wilson loop $[ \langle W(C) \rangle_K]_{K_p}$. In Ref.~\cite{Wang_gauge_glass_2003}, it is suggested that these two phases be identified by comparing perimeter-law decay vs. area-law decay in each of the distinct phases. 
Perimeter-law decay indicates that the fluctuations which cause the Wilson loop to decay are local, and hence characterizes the ordered phase.

The third phase is a disordered but glassy phase at large $K_p$ but small $K$. In this phase, quenched disorder is large, leading to a high density of flux tubes. 
As a result, the disorder averaged Wilson loop obeys an area-law, consistent with a disordered phase.
However, thermal fluctuations remain small, and the pattern of flux tubes is static when averaged over thermal fluctuations. Therefore, the square of the thermal expectation value, $[ \langle W(C) \rangle^2_K]_{K_p}$, should decay as perimeter-law and can hence distinguish the disordered phase from the glassy phase.

The LED observables we construct here can be used to more robustly classify these phases as well.
In particular, as explained in the main text, differentiating between perimeter-law and area-law is difficult in practice. Both of these signals decay exponentially in loop size, causing the sample complexity to increase significantly.
Further, area-law generically becomes mixed with perimeter-law in the presence of additional imperfections or perturbations, such as incoherent errors or gauge-constraint violating terms.
However, using LED, the Wilson loops will be amplified to one only in the ordered phase. Similarly, the squared Wilson loop $\langle W(C) \rangle^2_K$ would be amplified in the ordered and glassy phases. 
This observable can also be efficiently measured from snapshots, as long as we are able to sample multiple times from each quenched disorder realization. 
Another strategy, which could be more sample efficient, is the following: Given two independent snapshots $s_0,s_1$ from the same quenched disorder realization, one can compute the element-wise sum modulo two, $s = s_0 \oplus s_1$, and then measure the LED Wilson loop on the new snapshot $s$. Intuitively, such a procedure effectively removes the quenched disorder, which is fixed from snapshot-to-snapshot, leaving behind only the thermal fluctuations.
Identifying such gauge-glassy models is also difficult in numerical simulations, and hence LED and related ideas could also be applicable in those contexts.

\subsection{Annulus Decoders}

In this section, we consider an alternative approach to constructing error corrected operators, by applying  MWPM to an annulus.
Anyons supported on the annulus can be paired either with other anyons, or with the boundary of the annulus.
Interestingly, MWPM is a global decoder, where introducing or removing anyons can change the pairing far away.
However, by construction, the decoder cannot connect the interior of the annulus to the exterior.

We show in Figure~\ref{fig:supp_strip_decoder} the behavior of these corrected loop operators for coherent, incoherent, and mixed errors.
In particular, for the incoherent error model, we see the region classified as topological seems to coincide with the known error recovery threshold of $\approx 10.1\%$ for MWPM. This is much closer to the theoretical optimal error correction threshold of $\approx 10.9\%$, where we expect the phase transition from topological to disordered to occur.

Since this decoder is non-local, our existing theoretical arguments, and the connection to RG, may not apply. Nevertheless, an annulus decoder cannot change the super-selection sector---whether or not an anyon is contained within the annulus. Therefore, if the annulus-correction decorated loops go to one, then asymptotically large regions have well-defined super-selection sectors, suggesting the state is topologically-ordered. Indeed, it has been argued that such an annulus decoder could be considered a witness for topological entanglement entropy~\cite{wootton_witness_2012}.

\begin{figure}
    \centering
    \includegraphics[width=0.49\textwidth]{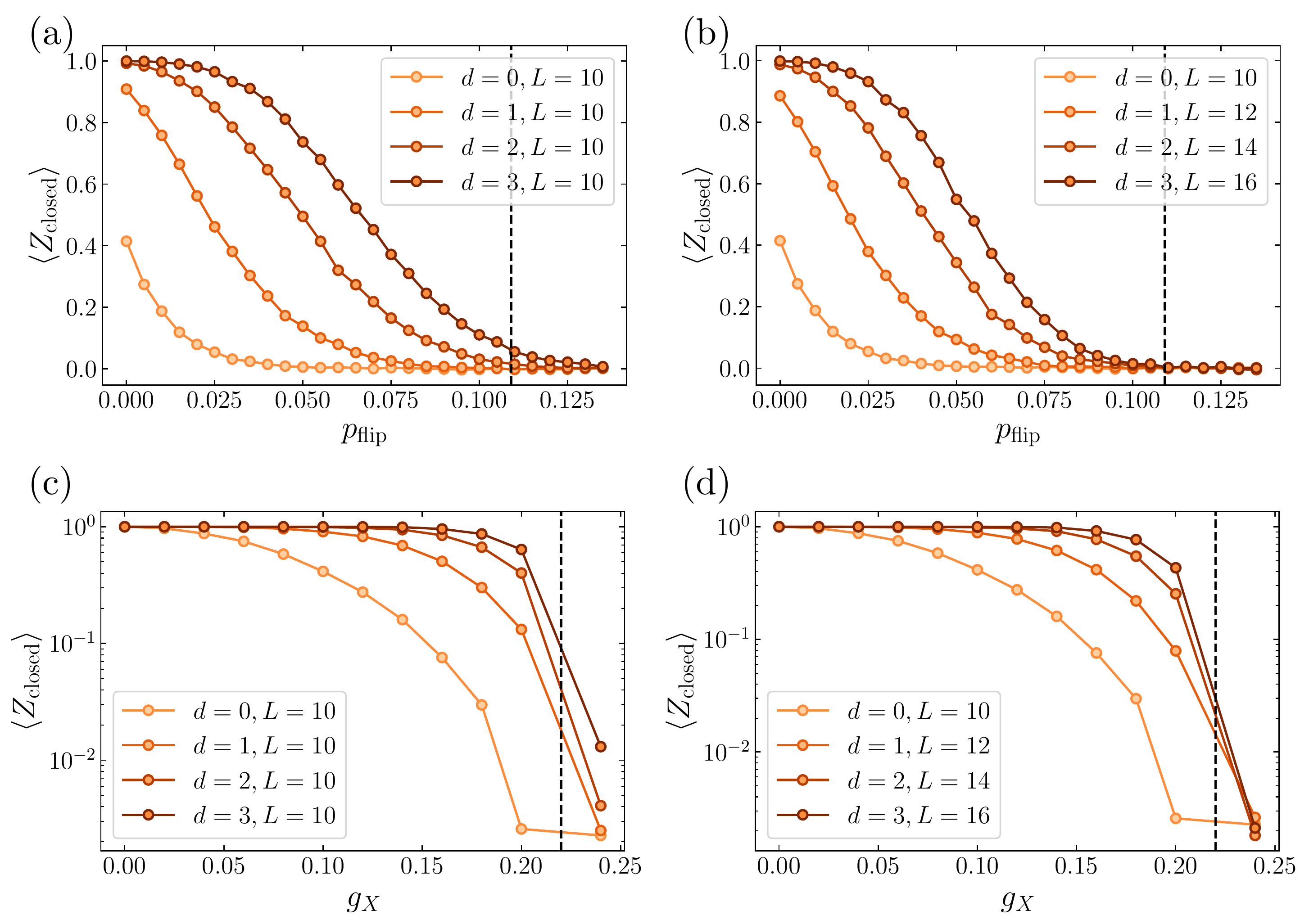}
    \caption{LED loops using MWPM on an annulus. We see the region classified as topological extends significantly further than the LED loops, for both incoherent (a,b) and coherent (c,d) perturbations. However, additional subtleties arise.  When we study fixed loop size and increasing annulus thickness (a,c), there appears to be amplification even on the trivial and disordered sides. However, if we scale the loopsize such that the size of the interior is fixed (b,d), then the improper amplification disappears. An interesting direction for future work is to develop a rigorous understanding of whether non-local annulus-based decorating can also serve as a topological order witness.}
    \label{fig:supp_strip_decoder}
\end{figure}

\vspace{2mm}
\noindent
{\bf Competing Interests.} M.D.L. is a co-founder and shareholder of QuEra Computing.

\vspace{1mm}
\noindent
{\bf Data Availability Statement.} The data that support the plots within this paper and other findings of this study are available from the corresponding author upon reasonable request.

\vspace{1mm}
\noindent
{\bf Code Availability Statement.} The code used to generate the plots within this paper and other findings of this study are available from the corresponding author upon reasonable request.

\vspace{1mm}
\noindent
{\bf Acknowledgments.} We thank E. Altman, Y. Bao, D. Bluvstein, Z.-P. Cian, S. Ebadi, G. Giudici, M. Hafezi, H.-Y. Huang, A. Kitaev, H. Levine, J. Preskill, S. Sachdev, R. Sahay, N. Tantivasadakarn, R. Verresen, A. Vishwanath, T. T. Wang, and X.-G. Wen for insightful discussions, and we especially thank D. Aasen and Z. Wang for providing helpful information and conversations on applying LED to non-abelian topological phases.  We also thank the Referee for the suggestion to consider gauge-glassy models. This work was supported by the US Department of Energy [DE-SC0021013 and DOE Quantum Systems Accelerator Center (contract no. 7568717)], the Defense Advanced Research Projects Agency (grant no. W911NF2010021), the National Science Foundation, the Department of Defense Multidisciplinary University Research Initiative (ARO MURI, grant no. W911NF2010082), and the Harvard-MIT Center for Ultracold Atoms. I.C. acknowledges support from the Alfred Spector and Rhonda Kost Fellowship of the Hertz Foundation, the Paul and Daisy Soros Fellowship, and the Department of Defense through the National Defense Science and Engineering Graduate Fellowship Program. N.M. acknowledges support from the Department of Energy Computational Science Graduate Fellowship under Award Number DE-SC0021110. HP acknowledges support from the ERC Starting grant no.~101041435 and the Erwin Schrödinger Center for Quantum Science and Technology.

\nocite{apsrev41Control}
\bibliographystyle{apsrev4-1}
\bibliography{refs}

%TC:endignore

\end{document}